\documentclass[11pt,onecolumn]{article}

\usepackage[margin=1in]{geometry}
\usepackage{setspace}
\usepackage{parskip}
\usepackage{authblk}
\usepackage{amsmath,amssymb,amsfonts,amsthm}
\usepackage{graphicx}
\usepackage{caption}
\usepackage{dcolumn}
\usepackage{bm}
\usepackage[colorlinks=true,linkcolor=blue,citecolor=blue]{hyperref}
\usepackage{braket}
\usepackage{xcolor}
\usepackage{array}
\usepackage{booktabs}
\usepackage[numbers,sort&compress]{natbib}
\usepackage[expansion=false]{microtype}
\usepackage{tikz}
\usetikzlibrary{shadows,calc,arrows.meta,positioning,fit,backgrounds}

\newcommand{\Kop}{\mathcal{K}}

\newcommand{\tr}{\mathrm{Tr}}
\newcommand{\id}{\mathbf{1}}
\newcommand{\diag}{\mathrm{diag}}

\newcommand{\alphaop}{\alpha^{\mathrm{op}}}
\newcommand{\betaop}{\beta^{\mathrm{op}}}
\newcommand{\tauop}{\tau^{\mathrm{op}}}
\newcommand{\Jop}{J^{\mathrm{op}}_m}
\newcommand{\tbetaop}{\tilde\beta^{\mathrm{op}}}

\newtheorem{theorem}{Theorem}
\newtheorem{proposition}[theorem]{Proposition}
\newtheorem{lemma}[theorem]{Lemma}

\newtheorem{remark}{Remark}

\begin{document}

\title{Quantum Krylov Learning}

\author[1,2,3,4]{Shunji Matsuura}
\author[4]{Yoji Kawamura}
\author[2,5,6]{Joseph Salfi}
\author[1,7,8]{Satoshi Iso}

\affil[1]{RIKEN Center for Interdisciplinary Theoretical and Mathematical Sciences (iTHEMS), RIKEN, Wako, Saitama 351-0198, Japan}
\affil[2]{Department of Electrical and Computer Engineering, University of British Columbia, Vancouver, BC V6T 1Z4, Canada}
\affil[3]{Department of Physics, University of Guelph, Guelph, ON N1G 2W1, Canada}
\affil[4]{Center for Mathematical Science and Advanced Technology, Japan Agency for Marine-Earth Science and Technology, Yokohama 236-0001, Japan}
\affil[5]{Department of Physics and Astronomy, University of British Columbia, Vancouver, BC V6T 1Z4, Canada}
\affil[6]{Stewart Blusson Quantum Matter Institute, University of British Columbia, Vancouver, BC V6T 1Z4, Canada}
\affil[7]{KEK Theory Center, Institute of Particle and Nuclear Studies, Oho 1-1, Tsukuba, Ibaraki 305-0801, Japan}
\affil[8]{Graduate University for Advanced Studies (SOKENDAI), Oho 1-1, Tsukuba, Ibaraki 305-0801, Japan}

\date{\today}

\maketitle

\begin{abstract}
The Hamiltonian of a natural quantum system is usually determined only indirectly.
In conventional approaches, microscopic models are proposed first and then tested
by comparing their predicted observables with experimental measurements.
A more direct approach is to extract information about the underlying Hamiltonian
from measurements performed on the physical system itself.
In many experimentally relevant settings, however, direct access to the full system is not available,
and measurements can be performed only through controllable probe degrees of freedom.
This situation arises naturally in quantum sensing and boundary spectroscopy.
Existing Hamiltonian-learning protocols under such restricted access
typically require either a special symmetry or a known parametric form of the Hamiltonian,
and many also rely on controlled operations beyond free evolution.
We introduce \emph{Quantum Krylov Learning (QKL)}, a restricted-access
Hamiltonian-learning framework that avoids these requirements.
The QKL uses single-temperature probe autocorrelators and their operator Krylov structure
to extract model-independent spectral information about the unknown system.
Rather than fitting a parametrized Hamiltonian, it reconstructs
the operator Krylov Jacobi matrix
which represents the dynamical modes accessible to the probe.
The method assumes only that the unobserved part of the system is at, or near,
infinite temperature.
Under this condition, we show that the operator Krylov Jacobi matrix can be reconstructed directly from
the autocorrelation function of a probe observable,
without oracle access to the Hamiltonian.
The required operations are only free evolution under the unknown $H$,
together with probe state preparation and measurement.
The learning core therefore requires no symmetry, no parametric model,
and no controlled-unitary gates.
We show a Lieb--Robinson bound for the operator Krylov chain
and use it to derive a data-driven criterion for the required observation window.
We also establish an end-to-end sample complexity guarantee.
Under generic conditions on the relevant spectral gaps and spectral weights,
a total number of measurement shots polynomial in $1/\varepsilon$
is sufficient to achieve spectral precision $\varepsilon$,
with explicit polynomial degrees.
We verify QKL numerically on an XXZ spin ring with a transverse field.
The fully data-driven pipeline reproduces the leading oracle Jacobi coefficients
to within a few percent, the higher coefficients accessible within the present
observation window to $10$--$34\%$, and the transport-time scale $\tauop$
to $0.9$--$6.6\%$.

\end{abstract}

\section{Introduction}
\label{sec:intro}

The microscopic Hamiltonian of a natural quantum system is usually obtained only indirectly.
Experiments measure observables such as order parameters, correlation functions,
or responses to external fields, and these measurements are compared with predictions
derived from a candidate microscopic model.
The Hamiltonian itself is therefore not determined directly from the experimental data.

A more direct approach to determining a microscopic Hamiltonian has been developed within quantum information theory, where it is commonly known as Hamiltonian learning. 
Rather than proposing a microscopic model first and testing its predictions against measured observables, Hamiltonian-learning methods aim to extract properties of the Hamiltonian directly from the observed dynamics of the system \cite{burgarth2009,burgarth2009gateway,difranco2009,wiebe2014,haah2023,huang2023,bairey2019,lapasar2012,hangleiter2024}. Many existing approaches, however, assume a level of experimental control characteristic of programmable quantum devices. In particular, they often rely on the ability to prepare a broad class of initial states, to apply controlled gate operations in addition to the free evolution generated by the unknown Hamiltonian, and to perform measurements in chosen bases. These assumptions are natural in the setting of a quantum processor, but they are generally unavailable when the goal is to learn the Hamiltonian of a naturally occurring quantum system. In such settings, state preparation, external control, and measurement access may all be strongly restricted, so that many of the standard assumptions underlying Hamiltonian learning no longer apply.

In many experimentally relevant settings, the microscopic degrees of freedom
of the system cannot be addressed or measured individually.
Information is instead obtained through controllable probe degrees of freedom
that remain coupled to the otherwise inaccessible system.
This restricted-access setting appears naturally in quantum sensing and
boundary spectroscopy.
Examples include nitrogen-vacancy and donor-spin sensors coupled to
unknown nuclear-spin environments, as well as local probes of correlated materials.
The measured probe dynamics then provide the available information about
the Hamiltonian governing the full interacting system.

Hamiltonian learning under such restricted access has also been studied
both theoretically and experimentally.
Identifiability conditions for local-probe settings are known
\cite{burgarth2012,sone2017}.
The protocol of Burgarth, Maruyama, and Nori
\cite{burgarth2009} reconstructs Heisenberg-chain couplings
from the survival amplitude of a boundary excitation,
but relies on excitation-number conservation,
$[H,\hat N]=0$.
Time-trace methods extract Hamiltonian parameters from probe dynamics
when the inaccessible subsystem is maximally mixed
\cite{sone2017,zhang2014}.
Quantum probe tomography
\cite{chen2025probe} reconstructs the parameters of a chosen Hamiltonian model
from probe measurements using Gibbs states prepared at several controlled inverse temperatures.
These approaches show that Hamiltonian information can be recovered under restricted access,
but each relies on additional structure or control beyond probe preparation,
free evolution, and probe readout.

Here we ask what can be learned when these assumptions are removed
and the available operations are limited to probe preparation,
free evolution under the unknown Hamiltonian, and probe measurement.

Krylov subspace methods provide a natural language for this problem.
They organize the dynamics generated by a Hamiltonian into the sequence of
operators reached from an experimentally accessible probe observable
~\cite{Klymko2022,Yoshioka2022,CortesGray2022,Shen2023,Kirby2024}.
We use this structure to develop a model-independent framework
for Hamiltonian learning under restricted quantum access.

In this paper,
we develop a new framework for restricted access Hamiltonian learning,
in which an infinite-temperature probe autocorrelation $C_\alpha(t)$ is used as a probe-only observable.
The autocorrelation can be measured using only probe preparation, free evolution under the unknown Hamiltonian, and probe readout.
The quantum-probe autocorrelation measured at a single fixed bath temperature
(ideally $\beta = 0$ for the exact construction below)
is first converted into the probe spectral measure $\mu_\alpha$
and its operator Krylov Jacobi representation $\Jop$,
which governs the dynamical evolution of the probe.
We call this framework \emph{Quantum Krylov Learning (QKL)}.
The spectral reconstruction requires measurements at only a single fixed bath temperature, rather than a scan over several temperatures as in quantum probe tomography.

The infinite-temperature condition plays a specific role in this construction.
The observable used in the protocol is a correlation of operators at different times.
This is written as a trace over the whole system involving the state of the bath.
When the state of the bath is a maximally mixed state,
every bath configuration contributes equally.
A uniform weight then drops out of the comparison,
and what remains depends only on the system's own dynamics.
At a finite temperature, the bath configurations are no longer equally likely,
and this bias reappears as an extra, unknown weighting on the measured frequencies.
Section~\ref{sec:setup} makes this weighting precise (Eq.~\eqref{eq:weighting_operator}),
and Sec.~\ref{sec:finite_temp} bounds the resulting deviation at finite temperature.

The efficiency of the protocol is tied to spectral resolution.
A generic many-body system can have a number of distinct spectral contributions
that grows exponentially with system size,
with correspondingly shrinking gaps and weights between them.
Resolving every one of these contributions would require a number of measurements
that grows with that same exponential cost,
since finer gaps and smaller weights are intrinsically harder to resolve.
Our sample-complexity result avoids this cost
by fixing the spectral resolution in advance,
rather than shrinking it as the system grows.
At any such fixed resolution, the number of spectral features that must be resolved
stays bounded independently of system size,
so the number of measurement shots needed is polynomial in $1/\varepsilon$
(Theorem~\ref{thm:sample_complexity} in Sec.~\ref{sec:sample_complexity}).
The leading Jacobi coefficients and the transport time $\tauop(m)$
depend only on low moments of the spectral measure.
These low moments are already captured by a fixed-resolution, coarse-grained reconstruction,
so they can be estimated accurately without resolving every individual Bohr frequency.

The central result underlying QKL is the following.
When the autocorrelation $C_\alpha(t)$ is sampled faster than twice
the highest Bohr frequency (the Nyquist criterion) over a sufficiently long window,
the operator Krylov Jacobi matrix $\Jop$ is reconstructed exactly
from probe-only measurements of $C_\alpha(t)$,
with no oracle access to the Hamiltonian
(Theorem~\ref{thm:oracle_free} in Sec.~\ref{sec:oracle_free}).
We numerically verify this exact equivalence between the oracle and oracle-free constructions
in Sec.~\ref{sec:numerics_thm3} (Fig.~\ref{fig:oracle_free}).
Figure~\ref{fig:pipeline} shows the QKL pipeline from probe data to the Jacobi matrix.
For finite observation windows and finite measurement shots,
the same reconstruction remains controlled,
with the sample-complexity bound discussed above
(Theorem~\ref{thm:sample_complexity} in Sec.~\ref{sec:sample_complexity}).
We demonstrate this finite-data robustness numerically
in Sec.~\ref{sec:numerics_robust} (Fig.~\ref{fig:robust}).
The remaining results, namely,
the Jacobi structure, the Lieb--Robinson bound, and the transport formula,
support the construction or use its output.

\begin{figure}[tbp]
\centering
\includegraphics[width=\textwidth]{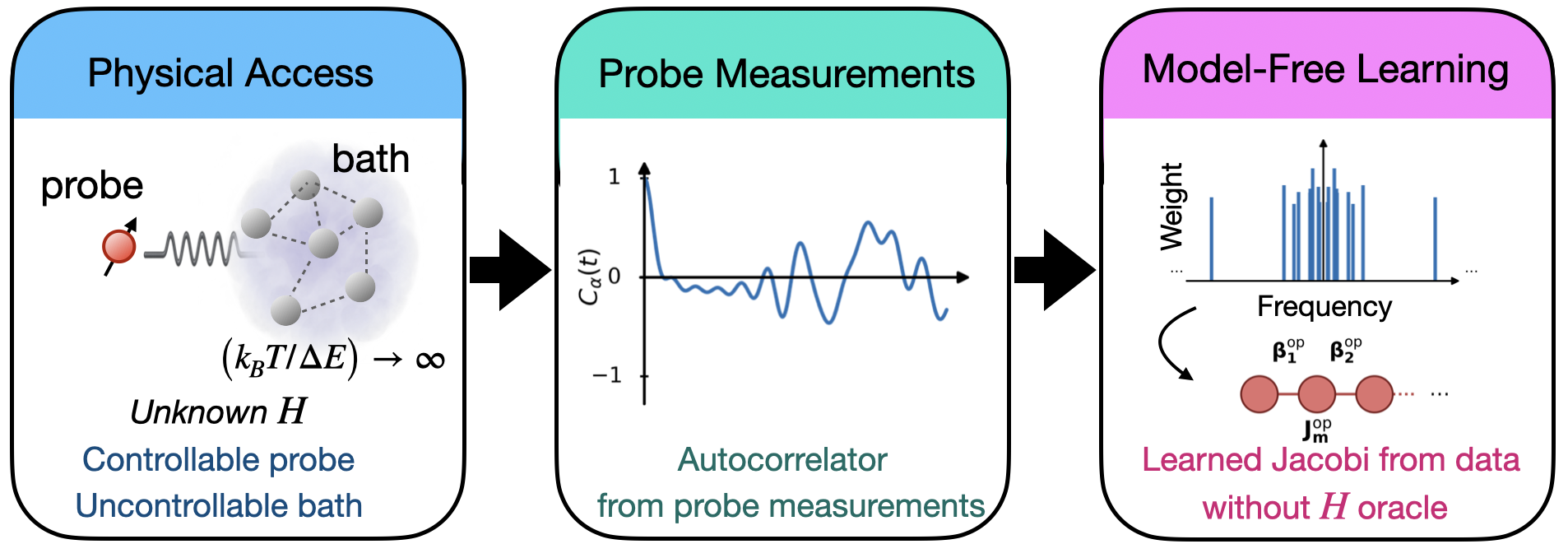}
\caption{Overview of Quantum Krylov Learning (QKL).
  \textbf{Physical access.} 
 A fully controlled probe is coupled to an uncontrolled and unmeasured bath.
The probe-bath system evolves under an unknown total Hamiltonian $H$.
The bath is assumed to be at effectively infinite temperature. More precisely it is in the regime $k_B T \gg \Delta E$, where $\Delta E$ denotes the characteristic
internal energy scale of the bath.
  \textbf{Probe time trace.} For a chosen probe operator $\sigma_\alpha$, the autocorrelation
$C_\alpha(t)$ is extracted from measurements of the probe after free time
evolution.
Appropriate initial probe states are chosen according to $\sigma_\alpha$,
so that only single-probe measurements are required.
  \textbf{Model-free learning.} The measured autocorrelation $C_\alpha(t)$ is converted directly from data
into its spectral measure
$\mu_\alpha=\{(\omega_\nu,p_\nu)\}$, supported on the Bohr frequencies of
the unknown Hamiltonian $H$.
The Stieltjes/Lanczos procedure then maps this spectral measure to the
operator Krylov Jacobi matrix $\Jop$, without requiring an oracle for $H$.
The resulting Krylov hopping coefficients $\betaop_n$ characterize the
effective Krylov chain seen by the probe and determine the data-driven
observation window $\tauop(m)$ associated with Krylov depth $m$.
Throughout the protocol, only free time evolution and single-probe
measurements are required, with no controlled gates, no symmetry
assumptions, and no parametric model of $H$.}
\label{fig:pipeline}
\end{figure}

The only physical requirement of the construction is
that the uncontrolled, unmeasured bath be at, or near, infinite temperature.
This is a natural state on many sensing platforms.
Examples are an NV-center sensor, which commonly operates at room temperature,
and a donor-spin sensor, which typically operates at cryogenic temperatures for high-fidelity readout.
In either case, the nuclear-spin bath's internal (dipolar) coupling scale is of order 1~kHz or less,
far smaller than the platform's operating temperature,
so the infinite-temperature condition holds regardless of which regime is used.
It can otherwise be emulated by classical bitstring randomization on programmable hardware.
We explicitly examine this requirement.
Section~\ref{sec:finite_temp} bounds the error incurred
when the bath temperature is high but finite,
and shows that for NV-center and donor-spin platforms
this error is subdominant to finite-window and shot-noise effects,
while for the bitstring-randomization route
the infinite-temperature condition is realized exactly.

Reference~\cite{matsuura2026state} develops the state Krylov analog
of the operator Krylov Jacobi chain constructed above.
For the resulting Jacobi chain $J_m$,
it proves a Lieb--Robinson (LR) bound with velocity $v_{\mathrm{LR}} = 2$
and a transport formula $t_\ast(m) \approx \tau_m = \sum_{n<m} 1/\beta_n$
with a universal slope $a \approx 1$, confirmed numerically for $N=6$--$12$.
The operator Krylov Jacobi matrix $\Jop$ constructed here
is mathematically a chain of the same kind.
Because it is a real symmetric tridiagonal matrix
with hopping coefficients $\betaop_n$ in place of $\beta_n$,
an analogous Lieb--Robinson bound and transport formula hold for it by the same argument,
with velocity $v_{\mathrm{LR}}^{\mathrm{op}} = 2$.
We record the precise statements and a short translation of the proof in Appendix~\ref{app:jacobi},
and verify the transport formula directly in operator space (Sec.~\ref{sec:numerics}).

Throughout, we use the term \emph{oracle-free} in a precise sense.
A protocol is oracle-free if, beyond the measured time series,
no property of $H$, no matrix element, no eigenvector, and no directly applied commutator enters.
The only allowed operations are single-probe state preparation,
free time evolution under the unknown $H$, and single-probe measurement.
This is strictly weaker than having oracle access to $[H, Q_j]$,
and strictly stronger than assuming nothing at all
(we assume that a finite-dimensional Hamiltonian governs the dynamics).

The paper is organized as follows.
Section~\ref{sec:setup} defines the problem setup and the single-probe observables.
Section~\ref{sec:jacobi} introduces the operator Krylov subspace, establishes its Jacobi structure,
and presents the Lieb--Robinson bound and the transport formula.
Section~\ref{sec:oracle_free} details the oracle-free construction of the Jacobi matrix from the infinite-temperature autocorrelation.
Section~\ref{sec:implementation} describes the experimental protocol and methods for realizing the required bath conditions.
Section~\ref{sec:sample_complexity} analyzes the robustness of the protocol against finite sampling and finite temperature.
Section~\ref{sec:numerics} provides numerical verification of the pipeline on an XXZ spin ring.
The XXZ-type Hamiltonian used as our numerical benchmark is also relevant to real devices.
Hole-spin qubits in Si and Ge nanostructures naturally realize anisotropic, XXZ-like exchange
via spin--orbit coupling~\cite{hetenyi2020exchange},
recently measured and controlled in a Si FinFET~\cite{geyer2024anisotropic}.
Because such qubits can also be arranged in ring geometries,
they are a promising hardware target for the protocol developed here.
Section~\ref{sec:discussion} discusses the results, comparing them with related Hamiltonian learning methods and identifying open directions.
Section~\ref{sec:conclusion} concludes the paper.
All full derivations, proofs, and the complete numerical study are collected in the appendices.

\section{Setup and Observables}
\label{sec:setup}

We first make the problem setting precise.
Consider a quantum system on a Hilbert space
\[
\mathcal{H}=\mathcal{H}_p\otimes\mathcal{H}_b,
\]
where $\mathcal{H}_p$ denotes an experimentally accessible probe
and $\mathcal{H}_b$ the remaining, inaccessible part of the system.
We decompose the total Hamiltonian as
\begin{equation}
H = H_p + H_b + H_{pb},
\label{eq:H_decomposition}
\end{equation}
where $H_p$ acts only on the probe, $H_b$ acts only on the bath,
and $H_{pb}$ describes the probe--bath interaction.
The probe Hamiltonian $H_p$ is assumed to be known,
and the probe itself is fully controllable.
The experimenter can prepare chosen initial states,
apply local gate operations, and perform measurements in chosen bases.
The learning protocol developed below, however, requires only probe-state preparation,
free evolution under the full Hamiltonian $H$, and probe readout.
By contrast, both $H_b$ and $H_{pb}$ are unknown,
and the bath cannot be individually initialized, controlled, or measured.
Information about these unknown terms is therefore obtained only through their effect on the probe dynamics.

The only macroscopic condition imposed on the bath is its temperature.
For the exact construction developed below, we take
\[
\rho_b=\frac{\id_b}{d_b},
\]
corresponding to the infinite-temperature state.
The role of this condition is made explicit below,
and deviations at finite temperature are analyzed in Sec.~\ref{sec:finite_temp}.

The probe $\mathcal{H}_p$ has dimension $d_p$ and
the inaccessible part $\mathcal{H}_b$ has dimension $d_b$, so $D = d_p d_b$.
The dynamics are generated by a time-independent Hamiltonian $H = H^\dagger$,
and unless stated later we assume no locality, symmetry, or parametric form.
The probe is a general qudit, and we do not assume $d_p = 2$ below.

Let $A_p = A_p^\dagger = \sum_{i=1}^{d_p} a_i|i\rangle \langle i|_p$
be a Hermitian probe observable and $A = A_p \otimes \id_b$.
For each eigenstate $|i\rangle_p$, prepare the quantum state of the whole system
$\rho_i = |i\rangle \langle i|_p \otimes \id_b / d_b$
(the unobserved part $\mathcal{H}_b$ at infinite temperature),
and measure $s_A(t|\rho_i) = \tr [\rho_iA(t)]$ (with $A(t) = e^{iHt}Ae^{-iHt}$).
A direct computation gives the probe-only measurement formula
\begin{equation}
C_A(t) := \frac{1}{D} \tr [A(t)A] = \frac{1}{d_p} \sum_{i=1}^{d_p} a_i \, s_A(t|\rho_i),
\label{eq:CA_def}
\end{equation}
that is, the infinite-temperature autocorrelation of any Hermitian probe observable
is a real linear combination of signals obtained from probe-eigenstate preparation and probe-only readout,
and requires neither controlled time evolution nor readout of the unobserved part $\mathcal{H}_b$.
The trade-off is that $\mathcal{H}_b$ must realize the maximally mixed state,
either physically (a high-temperature environment)
or by classical bitstring randomization on programmable hardware (Sec.~\ref{sec:implementation}).

The rest of the paper works, for concreteness, with an $N$-qubit realization.
Here $\mathcal{H} = (\mathbb{C}^2)^{\otimes N}$, $D = 2^N$, and the probe is a single qubit with
$A_p = \sigma^\alpha_p$, $\alpha \in \{x,y,z\}$ ($d_p = 2$, eigenvalues $a_\pm = \pm 1$).
Equation~\eqref{eq:CA_def} reduces to the two-state difference
\begin{equation}
C_\alpha(t) = \frac{1}{2} \bigl[ s_\alpha(t|\rho_+) - s_\alpha(t|\rho_-) \bigr],
\label{eq:two_state}
\end{equation}
and all statements below are written in this qubit notation.
The general qudit statement is recovered
under the replacement of $\sigma^\alpha_p$ and $C_\alpha(t)$ by $A_p$ and $C_A(t)$.
We do not assume $[H, \hat{N}] = 0$.
Any symmetry-breaking perturbations are allowed.

The probe Bloch vector
\begin{equation}
s_\alpha(t) = \tr \bigl[ \rho_0 \, \sigma^\alpha_p(t) \bigr], \qquad
\sigma^\alpha_p(t) = e^{iHt} \sigma^\alpha_p e^{-iHt},
\label{eq:bloch}
\end{equation}
is the primary observable, measured at discrete times $t_k = k \Delta t$.
Expanding in the energy eigenbasis $H|E_j\rangle = E_j |E_j\rangle$,
\begin{equation}
s_\alpha(t) = \sum_{j,k} c^\alpha_{jk} \, e^{i(E_j-E_k)t}, \qquad
c^\alpha_{jk} = \langle E_k|\rho_0|E_j\rangle \, \langle E_j|\sigma^\alpha_p|E_k\rangle.
\label{eq:spectral}
\end{equation}
The observable frequencies are the energy differences $\omega_{jk} = E_j - E_k$,
that is, the Liouvillian eigenvalues, not the energies themselves.
The residues $c^\alpha_{jk}$ are in general complex, which is different from the survival amplitude
where all spectral weights $w_j = |\langle E_j|v_0\rangle|^2$
for an initial state $|v_0\rangle$ are non-negative real numbers.

The central result requires the bath to be at, or near, infinite temperature,
$\rho_{\mathrm{bath}} = \id_{2\ldots N}/2^{N-1}$.
The reason for this choice is visible in the two-state difference,
\begin{equation}
  \tfrac{1}{2} \bigl[ s_\alpha(t|\rho_+) - s_\alpha(t|\rho_-) \bigr]
= \tfrac{1}{2} \tr \bigl[ (\sigma^\alpha_p \otimes \rho_{\mathrm{bath}}) \, \sigma^\alpha_p(t) \bigr],
\label{eq:weighting_operator}
\end{equation}
where the bath state is the \emph{weighting operator} of the correlation function.
Only $\rho_{\mathrm{bath}} \propto \id_b$ makes that weighting the Hilbert--Schmidt trace,
and only then does the measured correlation depend solely on the Hamiltonian dynamics,
rather than on an additional unknown bath-state weighting.
Any $e^{-\beta H_{\mathrm{bath}}}/Z$ with $\beta > 0$ depends on the Hamiltonian
being learned and mixes unknown thermal weights
into the measured object (Sec.~\ref{sec:finite_temp}).
Fixing this macroscopic temperature is the only handle retained on $\mathcal{H}_b$.
This regime requires neither measuring $\mathcal{H}_b$ nor any further control of it.

\section{Operator Krylov Subspace and the Jacobi Matrix}
\label{sec:jacobi}

We work in the Liouville space $\mathfrak{L}(\mathcal{H})$ of linear operators on $\mathcal{H}$,
equipped with the Hilbert--Schmidt inner product $(A|B) = \tr (A^\dagger B)/D$,
and we use round brackets $|A)$, $(A|$ to distinguish these from state-space kets and bras.
The Liouvillian superoperator is $\mathcal{L}(O) = [H,O]$,
the Heisenberg evolution is $O(t) = e^{i\mathcal{L}t}O$, and
$s_\alpha(t) = D\,(\rho_0|e^{i\mathcal{L}t}\sigma^\alpha_p)$.

In the vectorized representation,
$\mathcal{L} = [H,\cdot\,]$ corresponds to $H\otimes\id - \id\otimes H^\mathrm{T}$,
the difference of two commuting Hermitian matrices,
so $\mathcal{L}$ is Hermitian in the Hilbert--Schmidt inner product,
\begin{equation}
(A|\mathcal{L}(B)) = (\mathcal{L}(A)|B),
\label{eq:L_hermitian}
\end{equation}
with eigenvalues equal to the Bohr frequencies,
$\mathcal{L}(|E_j\rangle\langle E_k|) = \omega_{jk}|E_j\rangle\langle E_k|$, $\omega_{jk}=E_j-E_k$.
Hermiticity makes $e^{i\mathcal{L}t}$ unitary and, importantly,
reduces the Arnoldi process for $\mathcal{L}$ to the Lanczos process~\cite{lanczos1950,saad2003},
producing a real symmetric tridiagonal (Jacobi) matrix.
This relies on the system being closed.
For an open-system Liouvillian with dissipators,
the generator is generically non-Hermitian
and the output is a Hessenberg rather than a Jacobi matrix.

The signal generates the operator Krylov subspace~\cite{parker2019,nandy2025krylov}
\begin{equation}
\mathcal{K}_m(\mathcal{L},\sigma^\alpha_p) = 
\mathrm{span}\bigl\{\sigma^\alpha_p,\,[H,\sigma^\alpha_p],\,[H,[H,\sigma^\alpha_p]],\,\ldots,\,\mathcal{L}^{m-1}\sigma^\alpha_p\bigr\},
\label{eq:krylov_subspace}
\end{equation}
whose full dimension is at most $D^2 = 4^N$,
compared with at most $D = 2^N$ for the state Krylov space.
The orthonormal Krylov basis $\{Q_n\}$ of $\mathcal{K}_m$ is obtained from this sequence of operators
by Gram--Schmidt orthonormalization with respect to the Hilbert--Schmidt inner product
$(A|B) = \tr (A^\dagger B)/D$.
The seed is $Q_0 = \sigma^\alpha_p/\|\sigma^\alpha_p\|_{\mathrm{HS}}$,
and each subsequent $Q_{n+1}$ is obtained from $[H,Q_n]$
by subtracting its components along the already constructed $Q_n$ and $Q_{n-1}$ and normalizing.
This orthonormalization is exactly the Lanczos process for $\mathcal{L}$ seeded at $Q_0$,
which gives the following structure.

\begin{theorem}[Operator Krylov Jacobi structure]
\label{thm:jacobi}
Starting from $Q_0 = \sigma^\alpha_p/\|\sigma^\alpha_p\|_{\mathrm{HS}}$,
the Lanczos process for $\mathcal{L}$ in the Hilbert--Schmidt (HS) inner product gives the recurrence
\begin{equation}
\mathcal{L}(Q_j) = \betaop_{j+1}\,Q_{j+1} + \alphaop_j\,Q_j + \betaop_j\,Q_{j-1},
\label{eq:lanczos}
\end{equation}
with real coefficients $\alphaop_j = (Q_j|\mathcal{L}(Q_j))$ and $\betaop_j > 0$.
The resulting matrix
\begin{equation}
\Jop = \begin{pmatrix}
0         & \betaop_1 &           &        \\
\betaop_1 & 0         & \betaop_2 &        \\
          & \betaop_2 & 0         & \ddots \\
          &           & \ddots    & \ddots
\end{pmatrix}
\label{eq:jacobi_matrix}
\end{equation}
is real symmetric tridiagonal, and its diagonal entries vanish identically,
that is, $\alphaop_n = 0$ for all $n$ and for any Hamiltonian, without assuming any symmetry.
\end{theorem}

The proof is in Appendix~\ref{app:jacobi}.
Tridiagonality is the standard property of Lanczos applied to a Hermitian generator~\cite{saad2003}.
The vanishing of $\alphaop_n$ follows from the Hermiticity of the probe seed.
The Krylov vectors alternate in adjoint parity, $Q_n^\dagger = (-1)^n Q_n$,
and cyclicity of the trace then forces
$\alphaop_n = (Q_n|\mathcal{L}Q_n) = 0$ unconditionally.
This vanishing is verified numerically in Sec.~\ref{sec:numerics_structure}.

The Lanczos coefficients depend on $H$ only through a single scalar measure.
Define the operator Krylov spectral measure
\begin{equation}
d\mu_\alpha(\omega) = \frac{1}{D\,\|\sigma^\alpha_p\|_{\mathrm{HS}}^2}
\sum_{j,k} |M^\alpha_{jk}|^2\,\delta(\omega-\omega_{jk})\,d\omega, \qquad
M^\alpha_{jk} = \langle E_j|\sigma^\alpha_p|E_k\rangle.
\label{eq:measure}
\end{equation}
This is a normalized non-negative point measure supported on the Bohr frequencies,
with weights proportional to $|M^\alpha_{jk}|^2$.
Here and throughout,
$\|\sigma^\alpha_p\|_{\mathrm{HS}}^2 := (\sigma^\alpha_p|\sigma^\alpha_p) = \tr [(\sigma^\alpha_p)^2]/D = 1$
is taken with respect to the \emph{normalized} inner product $(A|B) = \tr (A^\dagger B)/D$,
which makes $\mu_\alpha$ a probability measure.
With the unnormalized Hilbert--Schmidt norm the prefactor would read $1 / \tr [(\sigma^\alpha_p)^2]$.
Because $\sigma^\alpha_p$ is Hermitian,
$|M^\alpha_{kj}| = |M^\alpha_{jk}|$ and $\omega_{kj} = -\omega_{jk}$.
Consequently, $\mu_\alpha$ is an even measure, meaning $\mu_\alpha(-S) = \mu_\alpha(S)$.
The Lanczos iterates take the form
$Q_n \propto \sum_{j,k} r_n(\omega_{jk})\,M^\alpha_{jk}\,|E_j\rangle\langle E_k|$
with real polynomials $r_n$, so every inner product entering the recursion is
\begin{equation}
(Q_n|f(\mathcal{L})|Q_{n'})
= \int r_n(\omega) \, f(\omega) \, r_{n'}(\omega) \, d\mu_\alpha(\omega).
\label{eq:inner_prod_main}
\end{equation}
Neither the eigenvectors of $H$ nor the phases of $M^\alpha_{jk}$ appear.
The Jacobi matrix $\Jop$ is a function of $\mu_\alpha$ alone,
and the recursion is the Stieltjes procedure for the orthonormal polynomials
$\{\pi_n\}$ of $\mu_\alpha$ (Appendix~\ref{app:measure_jacobi}).
In particular,
$\alphaop_n = \int \omega\, \pi_n^2 \, d\mu_\alpha = 0$,
because $\mu_\alpha$ is even.

The Jacobi matrix $\Jop$ defines a one-dimensional tight-binding chain in operator space,
with the site index $n$ as coordinate and $\betaop_n$ as hopping amplitude.
With the weighted operator Krylov distance
$\rho^{\mathrm{op}}(m,n) = \sum_{j=\min(m,n)+1}^{\max(m,n)} 1/\betaop_j$,
this chain obeys a Lieb--Robinson bound~\cite{lieb1972,nachtergaele2006}.

\begin{theorem}[Operator Krylov Lieb--Robinson bound]
\label{thm:LR}
For all $m,n$, all $t\in\mathbb{R}$ and all $\lambda>0$,
\begin{equation}
\bigl| (Q_m|e^{i\Jop t}|Q_n) \bigr|
\leq \exp\bigl[ -\lambda\,\rho^{\mathrm{op}}(m,n) + \Omega^{\mathrm{op}}_{M,\lambda} |t| \bigr],
\label{eq:LR}
\end{equation}
with $\Omega^{\mathrm{op}}_{M,\lambda} = \max_r \bigl[ f_\lambda(\betaop_{r-1}) + f_\lambda(\betaop_r) \bigr]$,
where $f_\lambda(\beta) = \beta\sinh(\lambda/\beta)$.
In the limit $\lambda \to 0$, the ratio $\Omega^{\mathrm{op}}_{M,\lambda}/\lambda$ approaches $2$.
This gives the Lieb--Robinson velocity $v_{\mathrm{LR}}^{\mathrm{op}} = 2$.
\end{theorem}

The proof (Appendix~\ref{app:jacobi}) is the same as
its state Krylov counterpart in Ref.~\cite{matsuura2026state}.
Each $1/\betaop_j$ is the time to cross one bond,
so $\rho^{\mathrm{op}}(m,n)$ is the time to cross the interval
and Eq.~\eqref{eq:LR} is a light cone.
The amplitude is exponentially small until
$|t| \gtrsim \rho^{\mathrm{op}}(m,n)/v_{\mathrm{LR}}^{\mathrm{op}}$.
The velocity is $2$ because each site couples in two directions.
Estimating the error of the $m$-step description
then gives the operator Krylov transport formula.
Consider the truncation error of the $m$-step description,
$\epsilon^{\mathrm{op}}_m(T) = \max_{t\leq T}|s^{(m)}_\alpha(t)-s_\alpha(t)|$.
The break time $t_{\ast}^{\mathrm{op}}(m)$ is the time at which this error
first exceeds a fixed threshold.
Because $\epsilon^{\mathrm{op}}_m(T)$ stays small for $T\lesssim\tauop(m)$,
the break time tracks $\tauop(m)$,
\begin{equation}
t_{\ast}^{\mathrm{op}}(m) \approx a\,\tauop(m) + b, \qquad
\tauop(m) = \sum_{n=1}^{m-1} \frac{1}{\betaop_n}.
\label{eq:transport}
\end{equation}
The $m \times m$ Jacobi matrix contains the hoppings $\betaop_1,\dots,\betaop_{m-1}$,
so the sum runs to $m-1$.
This convention matches the numerical figures,
whose axes read $\sum_{n<m} 1/\betaop_n$.
Here, the break time is defined from the \emph{round-trip} return quantity $C_\alpha - C_\alpha^{(m)}$.
The corresponding signal must propagate from the origin to the truncation boundary and then return.
It therefore travels approximately twice the Krylov distance $\rho^{\mathrm{op}}(m-1,0) = \tauop(m)$
at the Lieb--Robinson velocity $v_{\mathrm{LR}}^{\mathrm{op}} = 2$.
This leads to a break time proportional to $\tauop(m)$, with slope $a \approx 1$.
By contrast, a one-way matrix element such as $(Q_{m-1}|e^{i\Jop t}|Q_0)$
only needs to propagate from the origin to the boundary.
It therefore becomes appreciable at $t \approx \rho^{\mathrm{op}}(m-1,0)/2 = \tauop(m)/2$,
corresponding to slope $a \approx 0.5$.
See Appendix~\ref{app:jacobi}.
Because $\{\betaop_n\}$ becomes experimentally accessible through the oracle-free construction below,
Eq.~\eqref{eq:transport} serves as a data-driven criterion for the observation window.
From the measured time series alone, and with no model input,
the experimenter learns how much of the window carries information at a given Krylov depth.
We verify Eq.~\eqref{eq:transport} directly in operator space in Sec.~\ref{sec:numerics_structure}.
The intercept $b$ is a small offset
reflecting the initial propagation delay before reaching the threshold.

\section{Oracle-Free Construction of the Jacobi Matrix}
\label{sec:oracle_free}

In this section, we present the central algorithmic result of the paper, that is, 
the exact reconstruction of the operator Krylov Jacobi matrix from measurable probe data.
We first connect the infinite-temperature autocorrelation
to the characteristic function of the spectral measure (Sec.~\ref{sec:oracle_free_measure}),
and then detail the two-stage classical post-processing pipeline
that achieves this oracle-free construction (Sec.~\ref{sec:oracle_free_construction}).

\subsection{Autocorrelation and the spectral measure}
\label{sec:oracle_free_measure}

The oracle Lanczos of Sec.~\ref{sec:jacobi} builds $\Jop$
by applying $\mathcal{L}(Q_j) = [H,Q_j]$ directly to successive Krylov operators.
However, this approach requires full knowledge of $H$.
We now show that the same Jacobi matrix can be built from a single probe-measurable time series.

Inserting the energy eigenbasis gives
$\tr [\sigma^\alpha_p(t)\sigma^\alpha_p] = \sum_{j,k}e^{i\omega_{jk}t}|M^\alpha_{jk}|^2$.
Comparing this expression with Eq.~\eqref{eq:measure},
we see that the infinite-temperature autocorrelation is the characteristic function~\footnote{
``Characteristic function'' is the probabilists' name for $\int e^{i\omega t} d\mu_\alpha(\omega)$.}
of the measure $\mu_\alpha$.
Therefore,
\begin{equation}
C_\alpha(t)
= \frac{1}{D} \tr \bigl[ \sigma^\alpha_p(t) \, \sigma^\alpha_p \bigr]
= \int e^{i\omega t}\,d\mu_\alpha(\omega)
= \sum_\nu p_\nu\,e^{i\omega_\nu t},
\qquad
p_\nu \geq 0,
\quad
\sum_\nu p_\nu = 1.
\label{eq:C_measure}
\end{equation}
Here each atomic weight $p_\nu$ is
the pooled contribution of all transitions $(j,k)$ degenerate at the frequency $\omega_\nu$,
\begin{equation}
p_\nu = \frac{1}{D\,\|\sigma^\alpha_p\|_{\mathrm{HS}}^2} \sum_{(j,k):\,\omega_{jk}=\omega_\nu} |M^\alpha_{jk}|^2, \qquad
M^\alpha_{jk} = \langle E_j|\sigma^\alpha_p|E_k\rangle,
\label{eq:pbeta_def}
\end{equation}
which is precisely the atomic weight of the spectral measure $\mu_\alpha$ in Eq.~\eqref{eq:measure}.
The index $\nu$ labels the distinct frequencies $\omega_\nu$
in the spectral measure and their corresponding weights $p_\nu$.
Because the weights are fixed by the squared magnitudes $|M^\alpha_{jk}|^2$,
we have $p_\nu \geq 0$ (indicating a positive measure).
Furthermore, the Hermiticity of $\sigma^\alpha_p$ ensures
$|M^\alpha_{kj}| = |M^\alpha_{jk}|$ and $\omega_{kj} = -\omega_{jk}$.
Therefore, $\mu_\alpha$ is an even measure satisfying $\mu_\alpha(-S) = \mu_\alpha(S)$.
Because $\mu_\alpha$ is even, the exponentials pair into cosines,
\begin{equation}
C_\alpha(t) = \sum_{\omega_\nu>0} 2p_\nu\cos(\omega_\nu t) + p_0,
\label{eq:C_cosine}
\end{equation}
which is real and even in $t$, with $C_\alpha(0) = 1$.
Here, $p_0$ denotes the weight of the atom at $\omega = 0$, if present.

This is the recursion method of condensed-matter physics~\cite{parker2019,viswanath1994} applied in reverse.
Instead of resumming $C(t)$ from a known $H$, we invert the measured $C_\alpha(t)$ of an unknown $H$ to obtain $\Jop$.
Appendix~\ref{app:recursion} gives a detailed account of the recursion method.
It explains how the method expresses the autocorrelation as a continued fraction through the Jacobi coefficients.
It also shows that the present construction is the inverse of that map.

The state-dependent signal $s_\alpha(t)$ can have a \emph{dark sector}.
A dark sector is a set of transitions that are physically present,
but do not appear in a particular signal.
For $s_\alpha(t)$,
this happens when a transition has zero weight for the chosen initial state.
Such a transition is physically present in the dynamics,
yet it is invisible in the measured signal.
The autocorrelation $C_\alpha(t)$, by contrast, has no dark sector relative to $\Jop$.
A transition is missing from $C_\alpha$ only when its weight $p_\nu$ vanishes.
However, a vanishing $p_\nu$ means precisely that the transition does not enter $\Jop$ either.
Every transition that enters $\Jop$ therefore appears in $C_\alpha$.
A dark sector does exist relative to $H$ itself.
It is made up of the transitions for which $\langle E_j|\sigma^\alpha_p|E_k\rangle = 0$.
These transitions remain undetectable in this channel at any temperature,
no matter what initial state is chosen.
Changing the initial state can reduce the dark sector of $s_{\alpha}(t)$, whereas changing the probe observable \(\alpha\) can change the sector that is invisible to the probe channel itself.

Under the infinite-temperature bath condition,
the autocorrelation is obtained with probe-only operations
through the two-state measurement of Eq.~\eqref{eq:two_state}.
Indeed, preparing
\begin{equation}
\rho_\pm
= |\pm_\alpha\rangle \langle\pm_\alpha|_p \otimes \frac{\id_{2\ldots N}}{2^{N-1}},
\end{equation}
gives
\begin{equation}
  \frac{1}{2} \left[ s_\alpha(t|\rho_+) - s_\alpha(t|\rho_-) \right]
= \frac{1}{2} \tr \left[ (\rho_{+} - \rho_{-}) \sigma_p^\alpha(t) \right]
= \frac{1}{D} \tr \left[ \sigma_p^\alpha(t) \sigma_p^\alpha \right]
= C_\alpha(t).
\label{eq:probe_autocorr_derivation}
\end{equation}

At $\beta > 0$, the bath factor is no longer $\propto\id$,
and the middle equality fails.
The signal can then no longer be identified with $\mu_\alpha$,
but the spectral information in it is not lost.
The Bohr frequencies are unchanged and only their weights are thermally distorted
(Sec.~\ref{sec:finite_temp}).

\subsection{Two-stage construction of the Jacobi matrix and the oracle-free theorem}
\label{sec:oracle_free_construction}

The spectral-learning core of QKL proceeds in two stages.
First, the spectral measure,
$\mu_\alpha = \sum_\nu p_\nu\delta_{\omega_\nu}$,
is extracted from the sampled values $C_\alpha(t_k)$
by a matrix-pencil (Prony-type) estimator~\cite{hua1990}.
Treating the samples as a sum of complex exponentials,
a generalized eigenvalue problem on a Hankel matrix recovers the atoms $\omega_\nu$ and weights $p_\nu$,
and a non-negative least-squares step~\cite{lawson1974} imposes the positivity $\mu_\alpha \geq 0$.

Second, the recovered $(\omega_\nu,p_\nu)$ are put into the Stieltjes/Lanczos recursion,
exactly as if they had come from an oracle diagonalization of $H$, to generate $\Jop$.
The second stage reconstructs $\Jop$ exactly from the atoms and weights of $\mu_\alpha$ alone.
The mathematical justification is given in Appendix~\ref{app:measure_jacobi}.
The concrete algorithmic and experimental details of the full pipeline
are presented in Sec.~\ref{sec:implementation}.

The equivalence follows from two constructions of $\Jop$.
The first route starts from $H$ itself
and builds $\Jop$ by repeatedly applying the commutator $[H,\cdot\,]$,
exactly as in the oracle construction of Sec.~\ref{sec:jacobi}.
The second route starts from the measured autocorrelation $C_\alpha$,
obtains the spectral measure $\mu_\alpha$ from it via Eq.~\eqref{eq:C_measure},
and then obtains $\Jop$ from $\mu_\alpha$ via Eq.~\eqref{eq:inner_prod_main}.
These two routes agree because,
by Eq.~\eqref{eq:inner_prod_main},
the commutator oracle enters the construction of $\Jop$ only through $\mu_\alpha$.
By Eq.~\eqref{eq:C_measure},
the measured $C_\alpha$ already carries $\mu_\alpha$ in full.
The measured time series therefore replaces the oracle exactly.
The following theorem states this precisely.

\begin{theorem}[Oracle-free Jacobi matrix]
\label{thm:oracle_free}
Let $H$ be any finite-dimensional Hamiltonian, with no symmetry assumed,
and let the bath be at infinite temperature.
Let $\hat\mu_\alpha$ be the discrete measure
reconstructed from the samples $\{C_\alpha(t_k)\}_{k=0}^{K}$,
with $t_k = k\Delta t$,
using a consistent exponential-analysis method.
Suppose that
(a) $\Delta t < \pi/\omega_{\max}$ (Nyquist),
(b) $K+1 \geq 2 \,|\mathrm{supp}\,\mu_\alpha|$,~\footnote{
The support $\mathrm{supp}\,\mu_\alpha$ denotes
the set of distinct Bohr frequencies with nonzero weight in $\mu_\alpha$.
Several transitions with the same frequency therefore contribute to a single element of this set,
and $|\mathrm{supp}\,\mu_\alpha|$ denotes the number of such distinct frequencies.}
and
(c) the data are noiseless.
Then $\hat\mu_\alpha = \mu_\alpha$ exactly.
Applying the Stieltjes/Lanczos procedure to $\hat\mu_\alpha$ therefore
reproduces the oracle Jacobi matrix $\Jop$ exactly for all $m \leq |\mathrm{supp}\,\mu_\alpha|$.
The entire construction requires only probe preparation, free time evolution, and probe measurement.
No property of $H$ other than the measured time series is used at any stage.
\end{theorem}

\paragraph{Sketch of the proof.}
The argument consists of three steps.

(I) The autocorrelation
$C_\alpha(t) = \sum_\nu p_\nu e^{i\omega_\nu t}$
is a superposition of oscillations at the distinct Bohr frequencies $\omega_\nu$.
Condition (a) prevents aliasing of the high-frequency components.
Conditions (b) and (c) provide sufficiently many noiseless samples
for the Prony / matrix-pencil method to determine
the frequencies $\omega_\nu$ and their weights $p_\nu$ uniquely.
Equivalently, the distinct spectral lines are fully resolved.
The measured time series therefore determines the spectral measure $\mu_\alpha$ exactly.

(II) As shown in Appendix~\ref{app:measure_jacobi},
the Jacobi coefficients $\betaop_n$ are determined entirely by the measure $\mu_\alpha$
through the recurrence of orthogonal polynomials.
Since $\hat\mu_\alpha = \mu_\alpha$,
the Jacobi matrix constructed from $\hat\mu_\alpha$
is exactly the same as the one obtained by oracle Lanczos.

(III) The autocorrelation $C_\alpha(t)$ itself is obtained by preparing the probe,
allowing the system to evolve freely, and measuring the probe, as in Eq.~\eqref{eq:two_state}.
No matrix element, eigenvector, or commutator of $H$ is required.

The complete proof of these three steps is given in Appendix~\ref{app:oracle_free_proof}.

In practice, condition (b) is relaxed.
The estimator resolves the dominant atoms within the observation window $T = K\Delta t$.
Since $\Jop$ depends on $\mu_\alpha$ only through its first $2m$ moments,
the leading Jacobi coefficients become accurate as the captured weight approaches one.

Equation~\eqref{eq:C_measure} provides a consistency check.
The estimated weights $\hat p_\nu$ must be real and non-negative,
must sum to $C_\alpha(0) = 1$, and must be symmetric under $\omega_\nu \to -\omega_\nu$.
These are falsifiable constraints.
Violations quantify shot noise and imperfect bath thermalization.
We impose non-negativity directly in the least-squares step,
where it acts as a strong regularizer that suppresses spurious lattice frequencies.
By contrast, the complex residues $c^\alpha_\nu$ of the state-dependent signal
do not satisfy analogous reality and non-negativity constraints.

\section{Experimental Protocol and Realization of the Infinite-Temperature Bath}
\label{sec:implementation}

We show that the construction described above can be implemented as an experiment.
We describe, in turn, the single-probe circuit,
the measurement of the autocorrelation $C_\alpha(t)$,
and the ways of realizing the infinite-temperature bath,
which is the only physical requirement of the method.
Further details of the experimental pipeline and spectral estimation
are provided in Appendix~\ref{app:construction}.

\subsection{Single-probe circuit}

For a pure initial state $\rho_0 = |\psi_0\rangle\langle\psi_0|$,
writing $|\psi(t)\rangle = e^{-iHt}|\psi_0\rangle$ gives
$s_\alpha(t) = \langle\psi(t)|\sigma^\alpha_p|\psi(t)\rangle$,
so the circuit has three stages.
(1) Prepare $|\psi_0\rangle = |\!\uparrow\rangle_1\otimes|\mathbf{0}\rangle_{2\ldots N}$.
(2) Evolve the system under $H$ for a time $t_k$
(in the learning context $H$ is unknown and the system evolves on its own,
whereas in an engineered simulation $e^{-iHt}$ is implemented by Trotterization).
(3) Perform single-qubit measurement on spin 1 in the $\{\sigma^x,\sigma^y,\sigma^z\}$ bases.
No ancilla, no controlled gate, and no thermofield-double construction are needed.
For a mixed initial state
$\rho_0 = \sum_\mu p_\mu|\psi^{(\mu)}_0\rangle\langle\psi^{(\mu)}_0|$,
repeat and average the circuit with $\lfloor p_\mu M\rfloor$ shots for each $|\psi^{(\mu)}_0\rangle$.
This multi state strategy also increases the Liouvillian coverage.

\subsection{Measuring the autocorrelation}

The oracle-free construction requires the autocorrelation
$C_\alpha(t)$ defined operationally in Eq.~\eqref{eq:two_state}.
Its measurement requires only a minimal extension of the single-probe circuit.
Prepare the probe in $|+_\alpha\rangle$ or $|-_\alpha\rangle$,
place the bath at high temperature
(or, on engineered hardware, apply an independently drawn random bit-flip $X^{b_j}$ to each bath qubit),
evolve freely to $t_k$, and measure $\sigma^\alpha$ on the probe.
The estimator is
$\hat C_\alpha(t_k) = \tfrac{1}{2}[\hat s_\alpha(t_k|\rho_+) - \hat s_\alpha(t_k|\rho_-)]$,
averaged over $n_b$ bitstrings of $M$ shots each.
The bitstring-sampling variance adds to shot noise as
$\mathrm{Var}[\hat C_\alpha] \le 1/(2n_bM) + \mathrm{Var}_b[m_\alpha(b,t)]/n_b$,
where $m_\alpha(b,t)$ is the exact conditional signal for bitstring $b$.
Thus averaging independent bitstrings gives the expected $n_b^{-1/2}$ suppression,
but no universal $2^{-(N-1)/2}$ system-size law
is assumed for the computational-basis product-state ensemble.
For the $N = 6$ benchmark,
this contribution is quantified empirically in Sec.~\ref{sec:numerics_robust}.

The protocol requires free coherent evolution over the full observation window $T_C$,
so the platforms invoked above must supply a coherence time $T_2 \gtrsim T_C$.
The numbers are platform-dependent but easily estimated.
For an NV center coupled to a nuclear-spin environment with hyperfine couplings
$J/(2\pi) \sim 100$~kHz, the window $T_C = 160/J$ of Sec.~\ref{sec:numerics}
corresponds to $\approx 0.25$~ms.
This is within reach of dynamical-decoupling-extended coherence times
(reported at the millisecond scale for NV-center platforms),
though not of the bare $T_2^\ast$~\cite{ryan2010}.
For donor spins with $J/(2\pi) \sim 1$~MHz,
it corresponds to $\approx 25~\mu$s,
comfortably inside $T_2$~\cite{kobayashi2021}.
When decoherence over the window is not negligible,
it broadens the recovered poles into damped exponentials,
corresponding to the open-system situation of Sec.~\ref{sec:disc_open}.
The closed-system analysis of this paper applies as long as that broadening
stays below the window resolution $2\pi/T_C$,
that is, as long as $T_2 \gtrsim T_C$.

\subsection{Three routes to the infinite-temperature bath}

We outline the realizations of an infinite-temperature bath.
(i) Physical high-temperature bath.
When the bath temperature exceeds all bath energy scales,
$\rho_{\mathrm{bath}} \simeq \id_{2\ldots N}/2^{N-1}$ holds without active preparation,
as for an NV-center or donor-spin sensor.
Even at the platform's actual operating temperature
(room temperature for NV, or sub-kelvin to a few kelvin for donor spins, whichever regime is used),
this is far above the $\sim$1~kHz nuclear dipolar scale of the bath~\cite{lapasar2012,hangleiter2024}.
(ii) Classical randomization.
On programmable devices the maximally mixed state is emulated by averaging over uniformly random bitstrings,
$\id_{2\ldots N}/2^{N-1} = \mathbb{E}_b[|b\rangle\langle b|]$,
requiring only single-qubit bit-flip gates.
(iii) Purification by an EPR (Einstein--Podolsky--Rosen) pair.
When the device can prepare entangled states,
realize $\rho_{\mathrm{bath}} = \id_{2\ldots N}/2^{N-1}$
as the reduced state of a maximally entangled register of $(N-1)$ qubits,
\begin{equation}
|\Phi\rangle = \frac{1}{\sqrt{2^{N-1}}} \sum_{b\in\{0,1\}^{N-1}} |b\rangle_{\mathrm{bath}}\otimes|b\rangle_{\mathrm{aux}},
\label{eq:epr_state}
\end{equation}
for which $\tr_{\mathrm{aux}}(|\Phi\rangle\langle\Phi|) = \id_{2\ldots N}/2^{N-1}$ holds exactly,
so that the infinite-temperature trace appears as a pure-state expectation on the enlarged system
and Theorem~\ref{thm:oracle_free} applies without modification.
This realization suits solid-state spin-qubit processors that can prepare many-body states on all qubits
while readout is available only on a subset.

\section{Robustness to Finite Data and Finite Temperature}
\label{sec:sample_complexity}

Theorem~\ref{thm:oracle_free} relies on two idealizations.
The first is that the input autocorrelation $C_\alpha(t)$ is measured at infinite temperature,
so that the unobserved part $\mathcal{H}_b$ is maximally mixed.
The second is that the time series is sampled without noise,
with sufficiently fine time resolution and over a sufficiently long observation window.
In a real experiment, deviations occur in both respects.
However, as we show in this section, both can be controlled quantitatively,
while the required resources remain polynomial in the target precision.

The statistical deviation, arising from finite shots and a finite observation window,
can be suppressed using a number of shots polynomial in $1/\varepsilon$
(Sec.~\ref{sec:sample_complexity_stat}, Theorem~\ref{thm:sample_complexity}).
The systematic deviation due to finite temperature
is bounded by a small quantity proportional to the inverse temperature $\beta$.
Part of this deviation, namely the odd sector defined in Sec.~\ref{sec:finite_temp},
can be diagnosed directly from the data and is automatically removed by the pipeline.
The remaining even-sector bias is small and can be bounded explicitly (Sec.~\ref{sec:finite_temp}).

The results of this section are not required for the numerical verification that follows.
Readers interested primarily in the core construction may proceed directly to Sec.~\ref{sec:numerics}.

\subsection{Statistical deviation from finite shots and a finite observation window}
\label{sec:sample_complexity_stat}

Theorem~\ref{thm:oracle_free} stated exact recovery from noiseless samples.
Here we evaluate the number of measurements needed under realistic finite shots.
There are two sources of deviation.
One is shot noise. By performing $M$ measurements at each time step,
the estimated autocorrelation $\hat C_\alpha(t_k)$ concentrates around the true value.
Since the $\sigma^\alpha$ readouts are bounded by $\pm1$,
Hoeffding's inequality guarantees the exponential concentration
$P[|\hat C_\alpha(t_k)-C_\alpha(t_k)| > \delta]\leq 2 \exp(-M\delta^2)$
(see Lemma~\ref{lem:hoeff_explicit} in Appendix~\ref{app:explicit}).
Thus $M\sim1/\delta^2$ measurements suppress the signal error $\delta$.
The second source of deviation is frequency separation.
Extracting the individual Bohr frequencies from the time series
depends on resolving two nearby frequencies.
The ease of resolving them is governed by the observation window $T$
and the minimum frequency spacing
$q = \min_{\nu \neq \nu'} | \omega_\nu - \omega_{\nu'} |$.
Successfully separating these frequencies generally requires $qT \gtrsim 1$.
When frequencies are closely spaced,
separation becomes more difficult,
and a small error in the signal is more readily amplified into a large error
in the estimated frequency.
This error amplification arises from the ill-conditioning of the linear problem that maps the time series to frequencies. It can be bounded quantitatively by $q$, so the total number of shots required to reach precision $\varepsilon$ remains polynomial in $1/\varepsilon$.

For a many-body Hamiltonian, these hypotheses must be read at a chosen resolution.
The full measure $\mu_\alpha$ of a generic $H$ has $O(D^2)$ atoms
whose minimum gap and minimum weight shrink exponentially with $N$
(the exact values for the model of Sec.~\ref{sec:numerics} are listed below the theorem),
so no estimator can resolve all of them from a realistic record.
Given a target resolution $q$ and a weight floor $p_{\min}$,
let $\mu^{(q)}_\alpha$ denote a resolved measure whose retained atoms have spacing at least $q$ and weights at least $p_{\min}$.
The theorem below is exact when the sampled signal is generated by such a resolved measure.
If $\mu^{(q)}_\alpha$ is instead used as a coarse-grained approximation to a denser physical measure $\mu_\alpha$,
the resulting deterministic time-domain approximation error must be included in the signal-error budget.
The theorem bounds the number of shots needed to recover the resolved measure.
Throughout this subsection and Appendix~\ref{app:sample_complexity},
frequencies are expressed in a fixed reference unit,
so $\varepsilon$ is a dimensionless numerical tolerance.
This is stated in the following theorem.

\begin{theorem}[End-to-end sample complexity]
\label{thm:sample_complexity}
Let the resolved measure $\mu^{(q)}_\alpha$ have $r$ positive Bohr frequencies, minimum spacing $q$,
bandwidth $\omega_{\max}$, and minimum weight $p_{\min} = \min_\nu p_\nu$.
There is an estimator, using the sampled probe autocorrelation alone and no property of $H$,
that recovers all of its frequencies and weights to precision $\varepsilon$
with probability at least $1-\eta$, using a total number of shots
\begin{equation}
N_{\mathrm{tot}} = 
\mathrm{poly}\left( r,\ \frac{\omega_{\max}}{q},\ \frac{1}{p_{\min}},\ \frac{1}{\varepsilon},\ \log\left(\frac{1}{\eta}\right) \right),
\label{eq:sample_complexity}
\end{equation}
which is polynomial in $1/\varepsilon$.
\end{theorem}

The theorem is proved constructively.
Such an estimator is exhibited in Appendix~\ref{app:explicit},
where it is built out of the measured samples by a filtering step,
a quotient of two filtered sums, and a linear solve.
It is stated in this form because the constant in Eq.~\eqref{eq:explicit_exponents}
below belongs to that estimator.
The matrix pencil and band-limited non-negative fit actually
used in Sec.~\ref{sec:numerics} are close relatives of it,
and obey the same polynomial bound
by the standard route of spectral estimation,
which we record in Remark~\ref{rem:pencil_relation} (Appendix~\ref{app:explicit}).
That route does not produce explicit constants.
Figure~\ref{fig:estimator} compares the two implementations on the same data.

The polynomial can be made fully explicit.
A self-contained analysis (Appendix~\ref{app:explicit}),
which replaces the matrix-pencil perturbation chain by a Fej\'er-kernel localization
and uses nothing beyond Hoeffding's inequality and elementary estimates,
gives
\begin{equation}
N_{\mathrm{tot}}
\le \frac{1.8 \times 10^{5} \, \Gamma^{3} \, \omega_{\max}^{3}}{q^{3} \, p_{\min}^{2} \, \varepsilon^{2}} \ln\left( \frac{2(N_t + 2)}{\eta} \right), \qquad
\Gamma = e \, m_F + 4 ( 1 + \ln(s) ) + 1, \quad
m_F = \left\lceil \frac{1}{2} \ln\left( \frac{1}{\delta} \right) \right\rceil,
\label{eq:explicit_exponents}
\end{equation}
where $s$ is the number of atoms of $\mu^{(q)}_\alpha$,
$N_t = O(\Gamma\, \omega_{\max}/q)$ is the number of sampled times,
and $\delta \propto q \, p_{\min} \, \varepsilon / (\Gamma \, \omega_{\max})$
is the per-sample signal accuracy~\footnote{
The mutual dependence of $\Gamma$ and $\delta$ forms a simple self-consistency equation,
which can be solved iteratively.}.
The degrees are therefore $\varepsilon^{-2}$, $q^{-3}$, $p_{\min}^{-2}$, and $\omega_{\max}^{3}$,
with only polylogarithmic dependence on $s$ and $1/\eta$.
The exponent of $1/q$ splits.
Recovering the \emph{frequencies} alone costs only $q^{-1}$
(Remark~\ref{rem:expl_freqonly} in Appendix~\ref{app:explicit}),
and the extra $q^{-2}$ is the cost associated with estimating the weights,
incurred because a node misplaced by $\varepsilon_\omega$
accumulates a phase error $T\varepsilon_\omega$ across the record
that an amplitude fit cannot distinguish from a change of weight.
Since $\tauop(m)$ is a smooth functional of the measure that does not require
resolving individual Bohr frequencies, it sits in the cheaper regime.
This accounts for the accuracy of a few percent obtained for $\tauop$ in
Sec.~\ref{sec:numerics_robust} under conditions where the individual higher
$\betaop_n$ are less accurate.

\paragraph{Practice with a finite observation window.}
A finite record is naturally fit by a resolved approximation of the kind described above.
With a window $T = K\Delta t$, frequencies closer than $\lesssim 2\pi/T$ cannot be resolved,
and the estimator returns a truncated measure $\hat\mu_\alpha$ made of the dominant $r' \le s$ atoms,
that is, an estimate of $\mu^{(2\pi/T)}_\alpha$.
Strictly, this replacement is an approximation to the full signal,
and its induced time-domain discrepancy must remain
within the error budget of Theorem~\ref{thm:sample_complexity}.
Because the $m$-step Jacobi matrix depends on $\mu_\alpha$ only
through its first $2m$ moments (see Eq.~\eqref{eq:moment_match} in Appendix~\ref{app:oracle_free_proof}),
the leading coefficients $\betaop_n$ and the transport time $\tauop(m)$
are accurate when the resolved approximation preserves those low moments,
without the individual atoms ever being resolved.
This is the regime in which the numerics of Sec.~\ref{sec:numerics} operate.
There the window resolution $2\pi/T_C = 0.039$
exceeds the median gap $7.4\times10^{-3}$ of the exact measure,
so the record determines the coarse-grained measure and the low moments,
not the individual atoms,
and it is the leading coefficients and $\tauop$ that the demonstration tests.

The proof, with the estimator constructed explicitly, is given in Appendix~\ref{app:sample_complexity}.
In addition, for a single frequency of given weight the shot count
$1/\varepsilon^2$ is information-theoretically optimal at fixed evolution time
(Proposition~\ref{prop:optimal} in Appendix~\ref{app:optimal_scaling}).
Since recovering the whole measure includes recovering any one of its atoms,
$\varepsilon^{-2}$ lower-bounds the joint problem as well,
and Theorem~\ref{thm:sample_complexity} achieves it.
We do not claim optimality in the remaining parameters.
Exact frequency degeneracy $q \to 0$ is non-generic
(Remark~\ref{remark:smoothing} in Appendix~\ref{app:optimal_scaling}),
but even for a generic spectrum the bound can be very large.
For a many-body spectrum $q$ is exponentially small in $N$,
as the paragraph below quantifies.
This guarantee rests on the same logical footing as the end-to-end guarantee
established by parametric quantum probe tomography~\cite{chen2025probe}.
The difference is that the target here is a model-independent spectral measure,
and the access is a single infinite-temperature measurement rather than repeated Gibbs-state preparation.
The two guarantees are complementary.
The QPT method guarantees point identification of parameters within a known family,
and this work guarantees recovery of the model-independent measure.

\paragraph{Efficiency and its limitations.}
Three statements delimit the scope of Theorem~\ref{thm:sample_complexity}.
(i) The full measure $\mu_\alpha$ of a generic many-body $H$ has $O(D^2) = O(4^N)$ atoms,
and its minimum gap and minimum weight shrink exponentially with $N$.
For the XXZ model of Sec.~\ref{sec:numerics} at $g_x = 0.2$ the exact values
are shown in Table~\ref{tab:spectrum_scaling}.
Learning $\mu_\alpha$ at full resolution is consequently exponential
in $N$ for \emph{any} protocol.
This is a property of the target, not of the access model.
(ii) The protocol learns efficiently the measure at a chosen resolution.
At fixed $(q, p_{\min})$, the bound~\eqref{eq:sample_complexity} is polynomial
in $1/\varepsilon$ with the optimal exponent $2$.
(iii) The leading Jacobi coefficients and $\tauop(m)$ depend only on the low
moments of $\mu_\alpha$ and are therefore determined by the coarse-grained measure.
These are the quantities the numerical demonstration recovers.
The demonstration succeeds within a window whose resolution $2\pi/T_C = 0.039$
distinguishes only $\sim450$ frequency bins against the $888$ atoms at $N = 6$.
The guarantee of quantum probe tomography~\cite{chen2025probe} is polynomial in
the number of parameters of its family and in the system size.
The present guarantee is polynomial in $1/\varepsilon$ at fixed spectral resolution
and exponential in $N$ at full resolution.
The two are guarantees about different targets and are not directly comparable
along the system-size axis.

\begin{table}[tbp]
\centering
\caption{Parameters and spectral properties of the exact measure
  for the XXZ spin ring model ($J = 1$, $\Delta = 0.2$, $g_x = 0.2$)
  across various system sizes $N$.
  The table lists the number of atoms in the measure ($s$),
  the spectral bandwidth ($\omega_{\max}$),
  the minimum frequency spacing ($q_{\min}$), and
  the minimum spectral weight ($p_{\min}$).}
\label{tab:spectrum_scaling}
\begin{tabular}{lcccc}
\toprule
$N$ & atoms $s$ & $\omega_{\max}$ & $q_{\min}$           & $p_{\min}$            \\
\midrule
4   &    66     &  5.19           & $1.2 \times 10^{-2}$ & $1.2 \times 10^{ -6}$ \\
6   &   888     &  8.86           & $6.5 \times 10^{-6}$ & $3.4 \times 10^{-11}$ \\
8   & 11832     & 10.22           & $5.0 \times 10^{-7}$ & $1.1 \times 10^{-12}$ \\
\bottomrule
\end{tabular}
\end{table}

\subsection{Systematic deviation from finite temperature}
\label{sec:finite_temp}

Theorem~\ref{thm:oracle_free} crucially relied on the assumption that
the unobserved part $\mathcal{H}_b$ is at infinite temperature (maximally mixed).
On a real sensing platform the temperature is high but finite,
and $\mathcal{H}_b$ is slightly biased toward lower energies
($\rho_{\mathrm{bath}} \propto e^{-\beta H_{\mathrm{bath}}}$, with small $\beta > 0$).
We write $H = H_p \otimes \id_b + \id_p \otimes H_b + H_{\mathrm{int}}$
and take the bath state to be the Gibbs state of $H_b$.
At $\beta = 0$, the result is independent of this (otherwise ambiguous) split.

A finite temperature does not affect the Bohr frequencies.
The measured signal remains a sum of exponentials $e^{i\omega_{jk}t}$ and
only the weight of each component changes.
At $\beta = 0$ that weight is $|M^\alpha_{jk}|^2/D$,
whereas for $\beta > 0$ it acquires an $O(\beta)$ term built
from matrix elements of $\sigma^\alpha_p \otimes H_b$
(see Eq.~\eqref{eq:ftc_weights} in Appendix~\ref{app:finite_temp}).
Those matrix elements are not determined by $\mu_\alpha$,
so the measured signal can no longer be identified with the characteristic function of the
Hilbert--Schmidt spectral measure, and Theorem~\ref{thm:oracle_free} is no longer exact.
The obstruction is an unknown, Hamiltonian-dependent reweighting of the spectrum rather than a loss of spectral information. Correcting this reweighting requires knowledge of $\rho_b(\beta)$ and hence of $H_b$, which is itself part of the target information.
This is a limitation of the access model rather than an information-theoretic obstruction,
and it is removed by tunable $\beta$~\cite{chen2025probe}, by independent knowledge of $H_b$,
or by the bitstring randomization of Sec.~\ref{sec:implementation},
which realizes $\beta = 0$ exactly.

The high-temperature expansion turns this into a quantitative statement.
To first order, the correction measure splits into an even and an odd part,
and a parity selection rule (Appendix~\ref{app:finite_temp})
sends the even part into the hopping coefficients $\betaop_n$
and the odd part into the diagonal coefficients $\alphaop_n$.
The even-measure projection performed by the non-negative least-squares step
therefore removes the odd-sector artifact and restores $\alphaop_n=0$ exactly,
but leaves a systematic bias in $\betaop_n$, bounded by
\begin{equation}
\bigl| \tbetaop_n(\beta)-\betaop_n \bigr| \le 3\beta\,\Lambda\,\varrho_\star + O(\beta^2),
\qquad
\varrho_\star \le 2\Delta H_b \sqrt{\frac{g_{\max}}{D\,p_{\mathrm{cut}}}},
\label{eq:jacobi_finite_temp_main}
\end{equation}
with $\Lambda$ the Bohr bandwidth,
$\Delta H_b$ the spectral half-spread of the bath,
$p_{\mathrm{cut}}$ the smallest retained weight,
and $g_{\max}$ the maximum multiplicity of ordered level
pairs pooled into one retained Bohr-frequency atom.
Propagated to the transport time the bias is independent of the Krylov depth $m$
(see Eq.~\eqref{eq:tau_error} in Appendix~\ref{app:tau_error}).
The odd sector that the projection discards is actually informative.
Fitting it with an additional sine block
allows it to serve as an in-situ thermometer for the bath
(see Eq.~\eqref{eq:ftc_thermometer} in Appendix~\ref{app:thermometry}),
although it vanishes identically for systems with an antiunitary symmetry,
the XXZ benchmark used here among them.

For the $N = 6$ benchmark an exact finite-$\beta$ calculation gives
$\max_n|\Delta\betaop_n| \approx 30\,\beta J^2$ and $|\Delta\tauop|/\tauop \approx 1.1\,\beta J$,
so at $\beta J=10^{-4}$ the transport-time bias is $\approx10^{-4}$.
At room temperature, appropriate for an NV-center sensor, the nuclear dipolar bath coupling
($\sim$1~kHz) compared with $k_{\mathrm B}T/h \approx 6\,$THz
(where $k_{\mathrm B}$ is the Boltzmann constant and $h$ is the Planck constant)
gives $\beta\|H_b\| \sim 10^{-10}$.
For a donor-spin sensor operated at its cryogenic readout temperature ($\sim$0.3--3~K),
$k_{\mathrm B}T/h$ is $6$--$60\,$GHz instead, so $\beta\|H_b\| \sim 10^{-7}$--$10^{-8}$.
The finite-temperature systematic error is orders of magnitude
below the finite-window and sampling errors of Table~\ref{tab:finite_temp} (Sec.~\ref{sec:numerics_robust}).
Appendix~\ref{app:finite_temp} gives the derivation.

\subsection{Structure of the error budget}
\label{sec:error_budget}

Four sources of deviation propagate through the pipeline.
They are the finite window ($\sim2\pi/T$),
shot noise ($\sim1/\sqrt{M}$),
the model-dependent finite-bitstring-sampling fluctuation
(scaling as $n_b^{-1/2}$ for independent samples),
and finite temperature ($\propto\beta$).
Stage by stage they are listed with numbers
in Table~\ref{tab:finite_temp} (Sec.~\ref{sec:numerics_robust}).
At $N = 6$ the window and bath-sampling floors are comparable and dominant,
then shot noise, with finite temperature smallest by far.
We do not infer a universal system-size scaling
for the computational-basis bath-sampling term from these data.
These errors can be reduced by increasing $T$ or $M$, lowering $\beta$, or applying the even projection.

\section{Numerical Verification}
\label{sec:numerics}

We verify the results on the $N = 6$ XXZ ring with a transverse field,
\begin{equation}
H = \sum_{j=1}^{N} \bigl[ J (X_j X_{j+1} + Y_j Y_{j+1} + \Delta Z_j Z_{j+1}) + g_x X_j \bigr], \qquad
J = 1, \quad \Delta = 0.2,
\label{eq:model}
\end{equation}
with periodic boundary conditions and $g_x \in \{0, 0.2, 0.4\}$.
At $g_x = 0$, $[H, \hat{N}] = 0$.
For $g_x > 0$, the number-conservation symmetry is broken.
All calculations use exact diagonalization with full reorthogonalization.
The Lanczos parameters are $t_{\max} = 80$, $n_{\mathrm{times}} = 4001$ equally spaced steps,
$m_{\mathrm{Lanczos}} = 20$, and $n_\mu = 4$ initial states,
and the oracle-free comparison uses $m = 10$ with the window $T_C = 2 t_{\max} = 160$
at the same sampling rate.
At $g_x = 0$, only $s_z(t)$ is nonzero because of number conservation. For $g_x > 0$, all three Bloch components are occupied. The survival amplitude also decays faster as $g_x$ increases.

\subsection{Verification of the Jacobi structure, transport formula, and transport time scale}
\label{sec:numerics_structure}

Figure~\ref{fig:jacobi_struct} confirms Theorem~\ref{thm:jacobi}.

For all $g_x \in \{0,0.2,0.4\}$ and $m=20$, the heatmap of $|\Jop|$ is exactly tridiagonal, with off-tridiagonal entries at machine precision. The diagonal $\alphaop_n$ is also zero at machine precision. The Ritz values approximate the true Liouvillian eigenvalues, with the dominant modes captured first.

Figure~\ref{fig:transport} verifies the transport formula~\eqref{eq:transport} directly in operator space,
using the round-trip truncation error of the operator autocorrelation itself,
$\epsilon^{\mathrm{op}}_m(T) = \max_{s \le T} |C_\alpha(s) - C_\alpha^{(m)}(s)|$
with $C_\alpha^{(m)}(t) = (Q_0|e^{i\Jop t}|Q_0)$,
i.e., Eq.~\eqref{eq:duhamel} (Appendix~\ref{app:jacobi}) evaluated directly.
The break time is affine in $\tauop(m)$,
with slopes $(a, R^2) = (1.00, 0.997), (0.96, 0.998), (0.96, 0.999)$
and intercepts $b \approx 0$ for $g_x = 0, 0.2, 0.4$,
consistent with the round-trip velocity-$2$ cancellation of Theorem~\ref{thm:LR}
and matching the state Krylov result of Ref.~\cite{matsuura2026state} to the same precision.
Figure~\ref{fig:tau_accuracy} shows that the oracle-free $\tauop(m)$ tracks the oracle across $g_x$,
with final relative errors at $m = 10$ of $0.9\%, 3.2\%, 6.6\%$.

\begin{figure}[tbp]
\centering
\includegraphics[width=\linewidth]{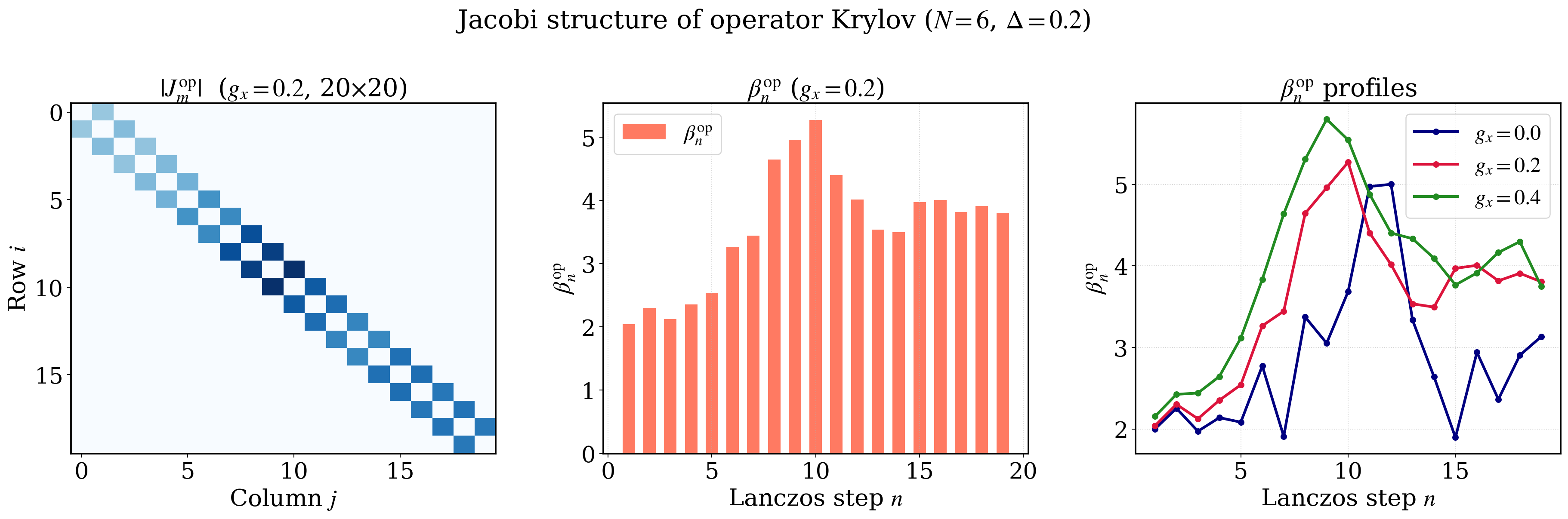}
\caption{Oracle construction of the operator Krylov Jacobi structure ($N = 6$, $\Delta = 0.2$).
  Left, heatmap of $|\Jop|$ ($20 \times 20$) for $g_x = 0.2$.
  The exactly tridiagonal pattern confirms Theorem~\ref{thm:jacobi},
  with off-tridiagonal entries at machine precision ($\lesssim 10^{-15}$).
  Center, subdiagonal coefficients $\betaop_n$ versus Lanczos step $n$ for $g_x = 0.2$
  (the diagonal $\alphaop_n$ is omitted since it is zero at machine precision).
  Right, $\betaop_n$ profiles for $g_x \in \{0,0.2,0.4\}$.
  Larger $g_x$ gives larger $\betaop_n$ for $n \lesssim 10$,
  reflecting faster operator spreading under stronger symmetry breaking.
  Beyond the peak the ordering is no longer monotonic in $g_x$.
  The peak near $n \approx 10$ reflects operator spreading before finite-size saturation.}
\label{fig:jacobi_struct}
\end{figure}

\begin{figure}[tbp]
\centering
\includegraphics[width=\linewidth]{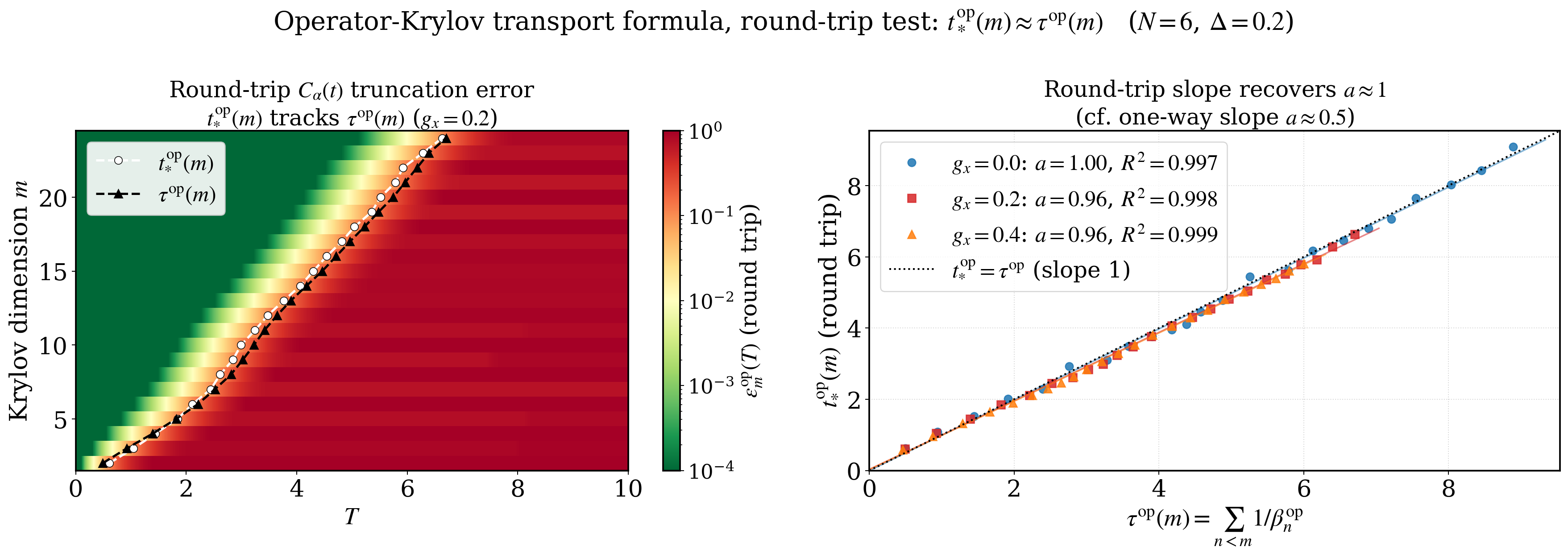}
\caption{Direct round-trip verification of the operator transport formula~\eqref{eq:transport}
  ($N = 6$, $\Delta = 0.2$, $\epsilon^{\mathrm{op}}_m(T) = \max_{s \le T} |C_\alpha(s) - C_\alpha^{(m)}(s)|$,
  the truncation error of the operator autocorrelation $C_\alpha(t) = (Q_0|e^{i\Jop t}|Q_0)$ itself,
  i.e., the state-independent return element bounded by Eq.~\eqref{eq:duhamel}).
  Left, heatmap of $\epsilon^{\mathrm{op}}_m(T)$ for $g_x = 0.2$
  with the break time $t_{\ast}^{\mathrm{op}}(m)$ (white) at threshold $0.1$,
  tracking $\tauop(m)$ (black) closely.
  Right, $t_{\ast}^{\mathrm{op}}(m)$ versus $\tauop(m)$
  with the affine fit $t_{\ast}^{\mathrm{op}} = a \,\tauop + b$
  for $g_x \in \{0, 0.2, 0.4\}$,
  yielding $(a, R^2) = (1.00, 0.997), (0.96, 0.998), (0.96, 0.999)$,
  respectively, with the intercept $b \approx 0$ in all cases.
  The round-trip slope recovers $a \approx 1$,
  matching the state Krylov result of Ref.~\cite{matsuura2026state}.}
\label{fig:transport}
\end{figure}

\begin{figure}[tbp]
\centering
\includegraphics[width=0.6\linewidth]{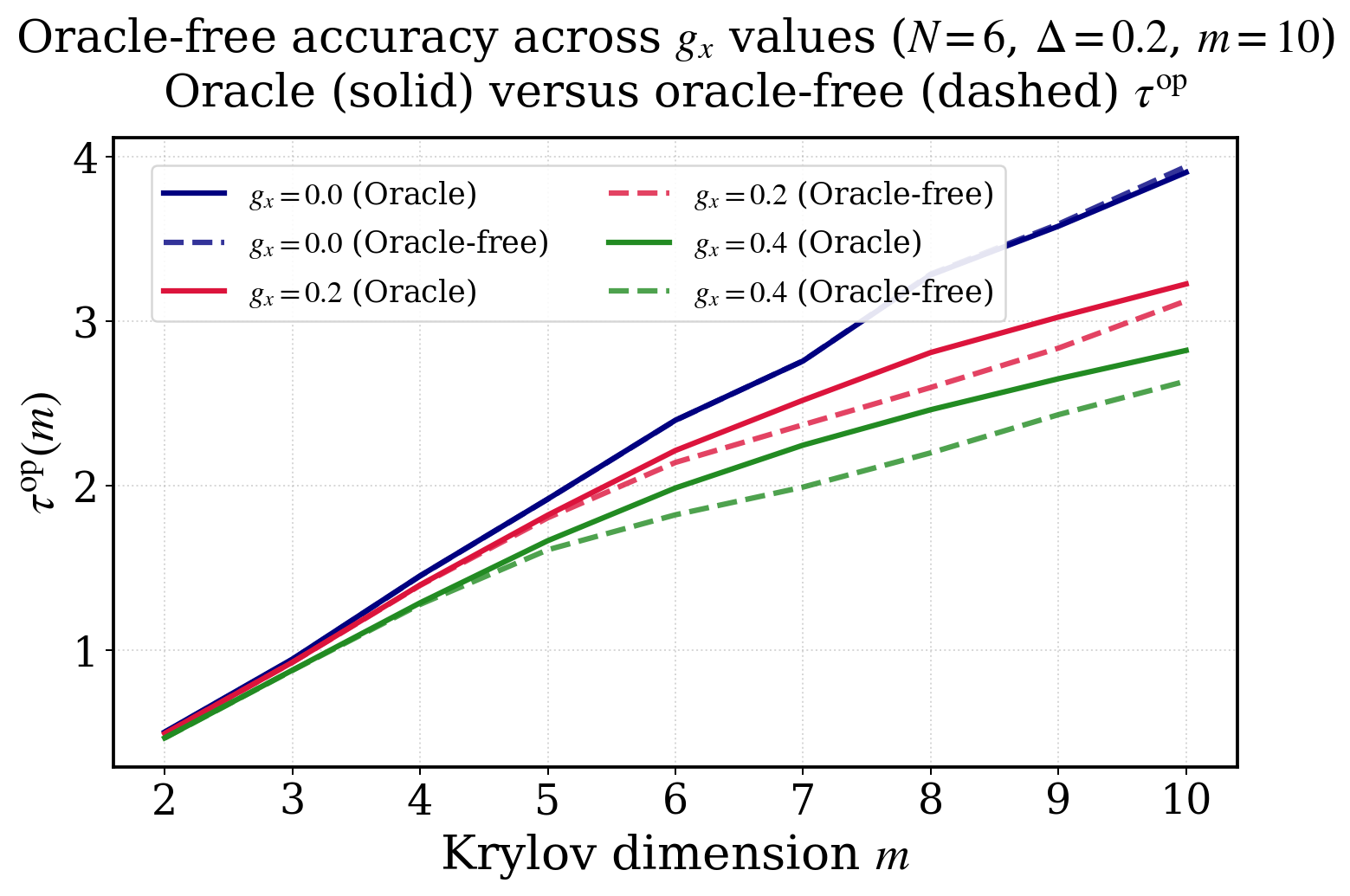}
\caption{Accuracy of the oracle-free $\tauop$ across $g_x$
  ($N = 6$, $\Delta = 0.2$, $m = 10$, $T_C = 160$).
  Oracle (solid) versus oracle-free (dashed) $\tauop(m)$ for $g_x \in \{0, 0.2, 0.4\}$.
  The curves closely track each other over the whole Krylov depth shown.
  Final relative errors at $m = 10$ are
  $0.9\%$ ($g_x = 0$), $3.2\%$ ($g_x = 0.2$), and $6.6\%$ ($g_x = 0.4$).}
\label{fig:tau_accuracy}
\end{figure}

\subsection{Numerical verification of Theorem~\ref{thm:oracle_free}}
\label{sec:numerics_thm3}

Figure~\ref{fig:oracle_free} verifies Theorem~\ref{thm:oracle_free} numerically.
For the model above (Eq.~\eqref{eq:model}),
we construct the same operator Krylov Jacobi matrix $\Jop$ in three ways and compare them.
(A) The oracle is the ordinary operator Lanczos that applies the commutator $\mathcal{L} = [H,\cdot\,]$
directly to the seed $\sigma^\alpha_p$ in the $D^2$-dimensional Liouville space,
assuming full knowledge of $H$.
It cannot be used in a restricted-access experiment but gives the numerical ground truth.
(B) Exact diagonalization of $H$ builds the \emph{exact} spectral measure $\mu_\alpha$
from the true Bohr frequencies $\omega_{jk} = E_j-E_k$
and weights $|\langle E_j|\sigma^\alpha_p|E_k\rangle|^2/(D\|\sigma^\alpha_p\|^2_{\mathrm{HS}})$,
and applies the Stieltjes/Lanczos procedure to it.
(B) is not built from a time series and is therefore not experimentally oracle-free. It is included to verify the theorem in two separate stages.
(C) is the actual protocol.
It builds $C_\alpha(t)$ from the two-state probe measurement of Eq.~\eqref{eq:two_state},
obtains the estimated measure $\hat\mu_\alpha$ by matrix pencil with band-limited non-negative least squares,
and applies the Stieltjes procedure to it.
No matrix element, eigenvector, or commutator of $H$ is used here.

The left panel shows $\betaop_n$ from the three routes,
and the two comparisons test the two claims of the theorem separately.
(A) and (B) agree to machine precision, $|\Delta\betaop_n| \lesssim 10^{-13}$.
This is claim (II), the stage $\mu_\alpha \to \Jop$, and it must hold because by
Eq.~\eqref{eq:inner_prod_main} $\Jop$ depends on $H$ only through $\mu_\alpha$.
(C) tracks them for the leading coefficients and deviates at higher $n$.
This is claim (I), the stage $C_\alpha \to \mu_\alpha$, degraded by the finite record,
since frequencies closer than $\lesssim 2\pi/T_C$ are unresolved and $\hat\mu_\alpha$
is a finite-resolution approximation of $\mu_\alpha$ rather than $\mu_\alpha$ itself.
Quantitatively, at $g_x = 0.2$ route (C) reproduces $\betaop_{1,2,3}$ to $\lesssim 0.5\%$
and $\betaop_4$ to $\approx 3\%$, while $\betaop_{5\text{--}9}$ deviate by $10$--$34\%$.
The deviation shrinks as the window and the shot count grow,
and grows with $g_x$ because faster operator spreading (Fig.~\ref{fig:jacobi_struct})
demands higher coefficients within the same window.

The center panel shows the cumulative quantity $\tauop(m) = \sum_{n<m} 1/\betaop_n$.
Reading $1/\betaop_n$ as the local time to advance one step along the Krylov chain,
$\tauop(m)$ is the cumulative time for the operator to reach Krylov depth $m$,
that is, a data-driven indicator of the effective observation window.
Here too the oracle and the oracle-free agree well,
showing that this transport time scale can be estimated from measured data alone.

The right panel shows the weighted spectral-capture error
$\sum_\nu p_\nu \min_i | \hat\omega_i - \omega_\nu | \, / \sum_\nu p_\nu$,
where $\hat\omega_i$ are the Ritz values of the $m\times m$ Jacobi matrix.
Equivalently, these are the Gauss-quadrature nodes associated with $\mu_\alpha$ at Krylov order $m$.
They need not coincide with individual exact atoms in general.
In this benchmark, each individual Ritz value sits very close to a genuine Bohr frequency
(mean nearest-atom distance $4\times 10^{-3}$--$2\times 10^{-2}$ for the oracle
and $2 \times 10^{-3}$--$4 \times 10^{-2}$ for the oracle-free construction, over $m=2$--$10$),
and the weighted capture error decreases with $m$ from $\approx 1.1$ to $\approx 0.3$--$0.4$,
showing that the dominant spectral weight is captured first.
That the oracle-free (red) and oracle (blue) curves are comparable means
that a $\Jop$ built from measured data alone approximates
the dominant Liouvillian transitions as accurately as the oracle.

Together, (A) and (B) confirm that the Jacobi matrix built from the exact
spectral measure matches the oracle result to machine precision.
(B) and (C) confirm that the measure estimated from the measured
autocorrelation closely approximates the exact measure.
Combined, these results verify the construction of
Theorem~\ref{thm:oracle_free}, including at the symmetric point $g_x=0$.
Since $\mu_\alpha$ has support in every symmetry sector,
no dark sector obstructs the construction of $\Jop$.

\begin{figure}[tbp]
\centering
\includegraphics[width=\linewidth]{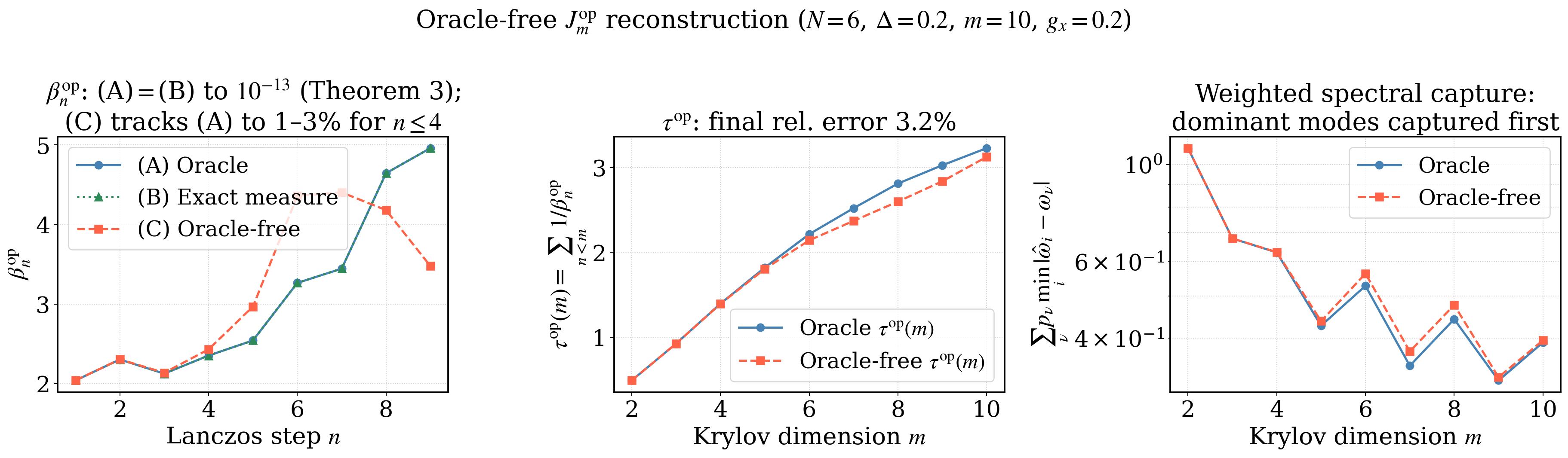}
\caption{Oracle versus oracle-free construction of $\Jop$
  from the infinite-temperature autocorrelation
  ($N = 6$, $\Delta = 0.2$, $m = 10$, $g_x = 0.2$).
  Left, subdiagonal coefficients $\betaop_n$ obtained from the oracle (blue circles),
  from Lanczos applied to the exact pooled measure $(\omega_\nu,p_\nu)$ (green triangles),
  and from the fully oracle-free procedure (red squares). In the oracle-free procedure,
  the measured autocorrelation $C_\alpha(t)$ is first used
  to estimate the spectrum by band-limited NNLS,
  after which the Stieltjes procedure is used to construct the Jacobi coefficients.
  The results from the oracle and from Lanczos on the exact pooled measure coincide to machine precision,
  confirming Theorem~\ref{thm:oracle_free}.
  The fully oracle-free result agrees with them to $\lesssim 0.5\%$ for $n \le 3$
  and to approximately $3\%$ for $n = 4$,
  with larger deviations at higher $n$ due to finite-window estimation error.
  Center, cumulative $\tauop(m)$ obtained from the oracle and oracle-free constructions.
  The two results closely track each other, with a final relative error of $3.2\%$.
  Right, weighted spectral-capture error
  $\sum_\nu p_\nu \min_i |\hat\omega_i - \omega_\nu|$ as a function of $m$
  for the oracle (blue) and oracle-free (red) constructions.
  Both errors decrease with $m$ at a similar rate.
  The dominant Liouvillian transitions are captured first,
  and the oracle-free spectral estimate achieves a quality comparable to that of the oracle.
  The oracle-free profiles are reconstructed from measured probe data alone.}
\label{fig:oracle_free}
\end{figure}

Figure~\ref{fig:estimator} repeats
this comparison with the estimator that Theorem~\ref{thm:oracle_free} actually assumes,
namely a genuine matrix pencil (Hankel generalized eigenvalue problem) run over the whole record,
alongside the band-limited non-negative fit used above.
The resolution of a pencil is set by the time span of its Hankel matrix,
so the record is decimated to just above the Nyquist rate and the Hankel blocks are chosen to span the full window.
On noiseless data the two estimators agree with the oracle at the same level.
The relative $\tauop(m{=}10)$ error is $0.9, 3.2, 6.6\%$ for the grid fit and
$9.4, 5.0, 5.7\%$ for the pencil at $g_x = 0, 0.2, 0.4$,
with the grid fit clearly better on the leading coefficients
($\betaop_{1,2,3}$ to $0.2$--$0.8\%$ versus $5$--$14\%$),
because positivity constrains exactly the low moments that fix them.
The correspondence between Theorem~\ref{thm:oracle_free} and the numerics is
therefore a property of the construction, not of one implementation.
Under shot noise, the difference becomes qualitative.
At $M = 10^2$--$10^4$ the unconstrained pencil loses the leading coefficients
entirely (relative $\tauop$ error $\approx 80\%$),
while the constrained fit remains at the few-percent level of Fig.~\ref{fig:robust}.
Non-negativity is the regularization that makes the reconstruction usable at
realistic shot counts, and the constrained estimator is used for the production figures throughout.

\begin{figure}[tbp]
\centering
\includegraphics[width=\linewidth]{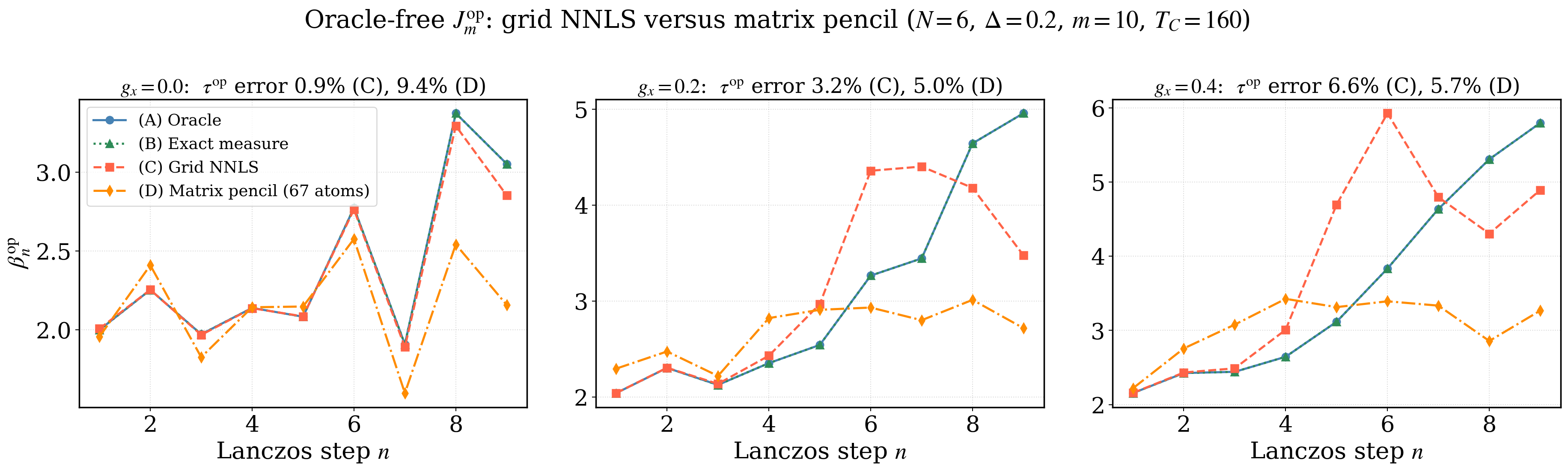}
\caption{Oracle-free $\Jop$, grid NNLS versus genuine matrix pencil
  ($N = 6$, $\Delta = 0.2$, $m = 10$, $T_C = 160$, noiseless data).
  Each panel shows the subdiagonal coefficients $\betaop_n$
  at one value of $g_x$ from
  (A) the oracle,
  (B) Lanczos on the exact measure,
  (C) the band-limited non-negative fit, and
  (D) the full-record matrix pencil.
  The two data-driven estimators agree with the oracle at the same few-percent level in $\tauop(m{=}10)$,
  with (C) more accurate on the leading coefficients
  because positivity constrains the low moments.
  Under shot noise (not shown) the pencil loses the leading coefficients
  ($\tauop$ error $\approx 80\%$ at $M = 10^2$--$10^4$)
  while (C) stays at the few-percent level (see text).}
\label{fig:estimator}
\end{figure}

\subsection{Robustness to finite data and error budget}
\label{sec:numerics_robust}

Figure~\ref{fig:robust} quantifies the robustness to finite bath sampling ($n_b$) and finite shots ($M$).
All cells are averaged over $n_{\mathrm{rep}} = 6$ independent noise
realizations and quoted as mean $\pm$ standard error.
The relative $\tauop$ error falls from $12$--$14\%$ at $M = 10^2$ to
$3$--$5\%$ at $M = 10^4$.
The $n_b$ dependence is significant and monotone at $M = 10^2$.
The mean Jacobi-coefficient error is
$0.795 \pm 0.010$, $0.763 \pm 0.007$, $0.717 \pm 0.012$ for $n_b = 4, 16, 64$,
with non-overlapping error bars,
consistent with the $n_b^{-1/2}$ suppression expected from averaging independent bitstring samples.
No system-size scaling is inferred from these $N=6$ data.
The individual coefficients $\betaop_n$ are noisier than this.
At $M = 10^4$ the mean relative error $\langle|\Delta\betaop_n|/\betaop_n\rangle$ is
$\approx 14$--$20\%$ across $n_b \in \{4, 16, 64\}$,
substantially larger than the $\tauop$ error
because the sum defining $\tauop(m)$ averages down the per-coefficient noise (Table~\ref{tab:finite_temp}).
This floor decreases by only $12\%$ between $M = 10^2$ and $M = 10^4$ (Fig.~\ref{fig:robust}, left),
so it is set by the finite window and the bath-sampling fluctuation rather than
by shot noise.
A pure shot-noise error would fall by a factor of $10$ over that range.
A shot-noise-free column ($M = \infty$, dotted in Fig.~\ref{fig:robust})
is included.
Its error is slightly larger and much more variable
($0.86 \pm 0.17$, $0.94 \pm 0.18$, $0.66 \pm 0.14$ for $n_b = 4, 16, 64$).
The reason is that the retention threshold of the estimator adapts to the noise floor.
Removing shot noise therefore widens the retained band
and admits bitstring-ensemble sampling fluctuation as signal.
That column therefore bounds the window-plus-bath floor from above rather than isolating it.

\begin{figure}[tbp]
\centering
\includegraphics[width=\linewidth]{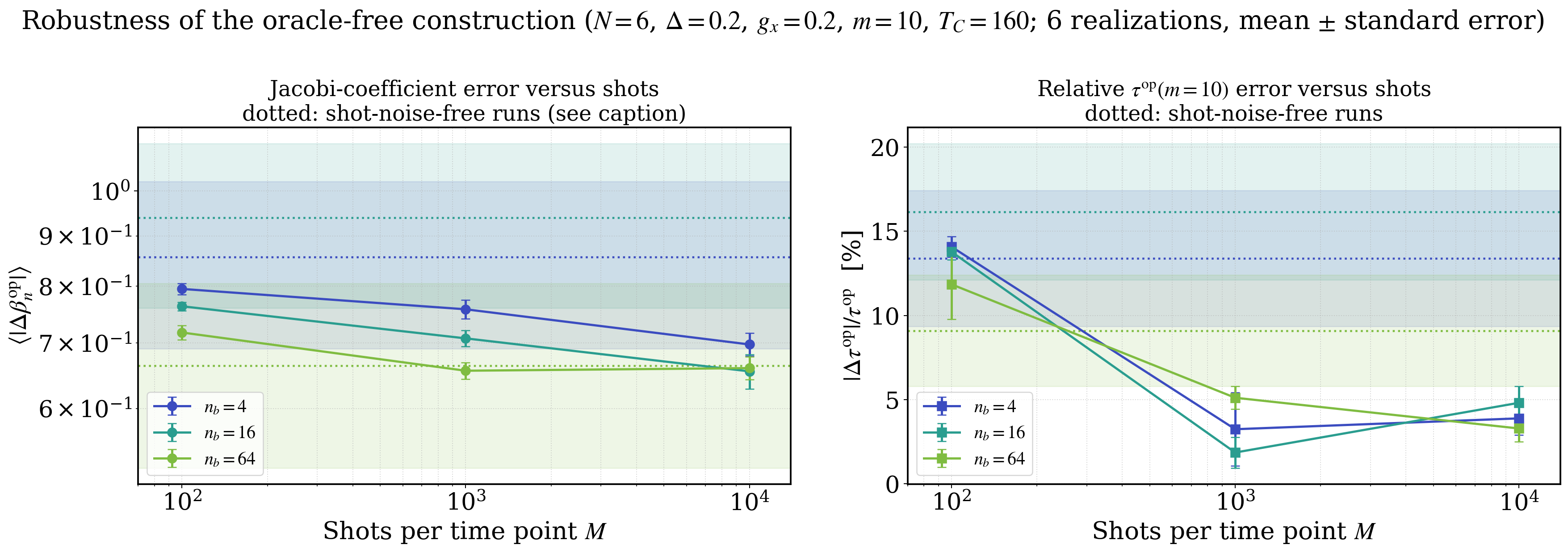}
\caption{Robustness of the oracle-free construction
  ($N = 6$, $\Delta = 0.2$, $g_x = 0.2$, $m = 10$, $T_C = 160$,
  $n_{\mathrm{rep}} = 6$ independent noise realizations, mean $\pm$ standard error).
  Left, mean Jacobi-coefficient error $\langle|\Delta\betaop_n|\rangle$
  versus shots per time point $M$,
  for several random-bitstring bath sample counts $n_b$.
  Dotted lines and shaded bands show the shot-noise-free runs ($M = \infty$).
  The plot shows absolute error, and the corresponding mean relative error
  $\langle|\Delta\betaop_n|/\betaop_n\rangle$ is approximately $14$--$20\%$.
  The ordering in $n_b$ is monotone at $M = 10^2$ with non-overlapping error
  bars, and the weak $M$ dependence (a $12\%$ decrease over two decades,
  against the tenfold decrease pure shot noise would give)
  shows that the $\betaop_n$ error floor is set by the finite window and bath
  sampling, not by shot noise.
  Right, relative error of $\tauop(m = 10)$ under the same conditions,
  falling from $12$--$14\%$ at $M = 10^2$ to $3$--$5\%$ at $M = 10^4$.
  The shot-noise-free runs are not lower than the finite-$M$ ones because the
  retention threshold of the estimator adapts to the noise floor (see text),
  and they bound the window-plus-bath floor from above.
  The non-negative least-squares constraint suppresses spurious noise atoms by construction.}
\label{fig:robust}
\end{figure}

We also estimate the concrete size of the systematic error from finite temperature.
The measured slopes of the exact finite-$\beta$ computation (Appendix~\ref{app:finite_temp})
for the $N = 6$ ring at $g_x = 0.2$ are
$\max_n|\Delta\betaop_n| \approx 30\,\beta J^2$ and
$|\Delta\tauop|/\tauop \approx 1.1\,\beta J$.
At $\beta J = 10^{-4}$
($k_{\mathrm B}T = 10^4J$, well within the realistic range for NV-center and donor-spin sensing platforms)
this gives a relative $\tauop$ error of $\approx 10^{-4}$ and,
propagated through the pole shifts,
a relative parameter error $\lesssim 10^{-3}$,
subdominant to the $0.1$--$2.7\%$ error from partial spectral coverage.

On NV-center and donor-spin sensing platforms,
the finite-temperature error is negligible across the whole pipeline.
On a programmable device where the bath is initialized near a pure state,
the bitstring-randomization route (Sec.~\ref{sec:implementation})
realizes the infinite-temperature condition exactly,
so the relevant error is instead the bitstring-sampling variance,
which is already quantified in Fig.~\ref{fig:robust}.

Finally, Table~\ref{tab:finite_temp} shows the error budget
summarizing the error at each stage in numbers
($N = 6$ XXZ ring, $g_x = 0.2$, $T_C = 160/J$, $M = 10^4$).
Among the three error sources, the dominant one at each stage is the finite window,
then shot noise, with finite temperature the smallest.
On a programmable device where the bath is initialized near a pure state,
the bitstring randomization (Sec.~\ref{sec:implementation})
realizes the infinite-temperature condition exactly,
so the finite-temperature column is not needed and the bitstring-sampling variance acts instead,
quantified in Fig.~\ref{fig:robust}.

\begin{table}[tbp]
\centering
\caption{Error budget for each stage of the oracle-free pipeline
  for the $N = 6$ XXZ ring ($g_x = 0.2$, $T_C = 160/J$, $M = 10^4$ shots, $n_b \le 64$ bitstrings).
  All entries are relative errors.
  The Finite window column refers to the noiseless,
  exact-bath pipeline, and the Bath sampling column
  to the empirical finite-bitstring-sampling fluctuation at $N = 6$,
  which at this system size is not separable from the window error
  in Fig.~\ref{fig:robust} (their combined effect is quoted in parentheses).
  The finite-temperature column is evaluated at $\beta J = 10^{-4}$
  from the measured slopes of Appendix~\ref{app:finite_temp}
  (e.g., $|\Delta\tauop|/\tauop \approx 1.1\,\beta J$).
  The dominant source at each stage is shown in bold.}
\label{tab:finite_temp}
\begin{tabular}{lllll}
\toprule
Stage & Finite-$T$          & Finite window & Shot noise & Bath sampling \\
      & ($\beta J=10^{-4}$) &               & ($M=10^4$) &               \\
\midrule
Measure $\mu_\alpha$           & $\lesssim 0.01\%$ & $\mathbf{\sim 4\%}$ (res.\ $2\pi/T_C$)  & $\sim 1\%$                         & $\sim 2\%$ \\
Jacobi coeff.\ $\betaop_n$     & $\lesssim 0.2\%$  & \multicolumn{3}{l}{$\mathbf{14\text{--}20\%}$ (window $+$ bath, $M$-independent)}         \\
Transport time $\tauop(m=10)$  & $\sim     0.01\%$ & $\mathbf{\sim 3.2\%}$                   & \multicolumn{2}{l}{$3$--$5\%$ (shot $+$ bath)}  \\
Pole position $\hat\omega_\nu$ & $\lesssim 0.1\%$  & $\mathbf{\sim 4\%}$                     & $\sim 1\%$                         & $\sim 2\%$ \\
\bottomrule
\end{tabular}
\end{table}

\section{Discussion}
\label{sec:discussion}

In Quantum Krylov Learning (QKL),
the operator Krylov Jacobi matrix $\Jop$ and
the state Krylov matrix $J_m$ of Ref.~\cite{matsuura2026state}
share the same mathematical structure.
Both are Jacobi matrices obeying a Lieb--Robinson bound with $v_{\mathrm{LR}} = 2$
and a transport formula, but the two are built for opposite purposes.
The matrix $J_m$ is an efficient signal predictor that needs oracle access
and is restricted to $[H, \hat{N}] = 0$.
The matrix $\Jop$ is a learned, oracle-free spectral object valid for any Hamiltonian.
However, $\Jop$ is not itself an efficient predictor.
Expanding $\rho_0$ in the Krylov basis of $\sigma^\alpha_p$
takes a number of steps of order the number of distinct Bohr frequencies, which can reach $D^2$.
We therefore use $\Jop$ for spectral learning and, where a signal has to be predicted, $J_m$.

The construction removes the central restriction of the number-conserving survival-amplitude approach,
since $C_\alpha(t)$ is well defined and probe-measurable for every $H$
and no stationary reference state is needed (Sec.~\ref{sec:intro}, Appendix~\ref{app:survival}).
The rest of this section places the result among related methods (Sec.~\ref{sec:disc_comparison}),
states what it does and does not determine (Sec.~\ref{sec:disc_scope}),
and collects what remains open (Sec.~\ref{sec:disc_open}).

\subsection{Comparison with related methods}
\label{sec:disc_comparison}

The setting of QKL is closest to the earlier time-trace identification series
based on the eigensystem realization algorithm (ERA)~\cite{sone2017,zhang2014}
and to quantum probe tomography (QPT)~\cite{chen2025probe},
which learns a many-body Hamiltonian from a single-site probe.
Table~\ref{tab:comparison} places the three side by side.

\begin{table}[tbp]
\centering
\small
\caption{The three approaches to single-probe Hamiltonian learning.
  Eigensystem realization algorithm (ERA)~\cite{sone2017,zhang2014} and 
  quantum probe tomography (QPT)~\cite{chen2025probe}
  return a parameter vector inside an assumed model class,
  whereas Quantum Krylov Learning (QKL) returns the spectral measure, model-independently.}
\label{tab:comparison}
\renewcommand{\arraystretch}{1.25}
\begin{tabular}{@{}>{\raggedright\arraybackslash}p{19mm}>{\raggedright\arraybackslash}p{34mm}>{\raggedright\arraybackslash}p{38mm}>{\raggedright\arraybackslash}p{50mm}@{}}
\toprule
                                       &
\textbf{ERA}~\cite{sone2017,zhang2014} &
\textbf{QPT}~\cite{chen2025probe}      &
\textbf{QKL (this work)}               \\
\midrule
Access model &
Probe pure, all other qubits maximally mixed &
Gibbs states $\rho \propto e^{-\beta H}$ at tunable $\beta$, interleaved with probe operations &
Single fixed $\beta \approx 0$, with the bath $\id_{2\ldots N}/2^{N-1}$ either natural or bit-flip randomized, and only probe preparation, free evolution, and single-probe measurement \\
\addlinespace
Learned object &
Hamiltonian parameters &
Parameter vector in a structured family &
Spectral measure $\mu_\alpha$, equivalently $\Jop$ \\
\addlinespace
Data domain &
Time traces through a transfer-function intermediate &
Probe observables at several times \emph{and} temperatures, with polynomial relations &
Full $C_\alpha(t)$ over a window $T_C$, with Bohr atoms resolved as decay-free exponentials at a single $\beta$ \\
\addlinespace
Identifiability &
Instance by instance &
Generic, by algebraic geometry and a smoothing analysis &
Spectral measure $\mu_\alpha$ is identifiable, but $H$ itself remains non-unique (Sec.~\ref{sec:disc_scope}) \\
\addlinespace
Guarantee &
No end-to-end bound &
Cost of $\mathrm{poly}(1/\varepsilon)$ queries and polylogarithmic post-processing &
Cost of $\mathrm{poly}(1/\varepsilon)$ shots at fixed spectral resolution, with explicit degrees
(Theorem~\ref{thm:sample_complexity}, Eq.~\eqref{eq:explicit_exponents})
and optimal $\varepsilon^{-2}$ shot scaling at fixed evolution time
(Proposition~\ref{prop:optimal}) \\
\addlinespace
Output &
Individual couplings &
Individual couplings &
Operator-growth rate $\{\betaop_n\}$, Ritz values $\approx$ dominant Liouvillian transitions, and a data-driven window criterion $\tauop(m)$ \\
\bottomrule
\end{tabular}
\end{table}

Three points are not captured in the table.
First, the initial state of the ERA series,
with the probe pure and all other qubits maximally mixed, is \emph{exactly} our $\rho_\pm$.
The three approaches differ in the treatment of the resulting time series,
not in the preparation of the system.
Second, the structural symmetries that QPT identifies as unavoidable in its access model,
for instance the reflection about the probe site on a translation-invariant lattice,
are special cases of our probe-commutant gauge.
A reflection that fixes the probe commutes with $\sigma^\alpha_p$.
Third, the ERA route has two limitations noted in the literature~\cite{chen2025probe}.
They are a numerically fragile transfer-function intermediate and identifiability that is checked instance by instance.
Both are resolved here without changing the physical setting.
The transfer function is replaced by the spectral measure,
which Theorem~\ref{thm:oracle_free} determines exactly and
which is estimated under a positivity constraint.
Identifiability is organized by the probe spectral measure and its explicit non-identifiabilities.

The two guarantees address different targets.
QPT identifies Hamiltonian parameters within a known structured family.
QKL instead reconstructs a model-independent spectral measure at a fixed resolution.
Theorem~\ref{thm:sample_complexity} provides an end-to-end guarantee for this resolved spectral reconstruction.
Because the targets differ, their system-size dependence should not be compared directly.

\subsection{Scope of the reconstruction}
\label{sec:disc_scope}

Two limits bound what probe-only data can reveal about $H$,
and they are of different kinds.
The first is exact and information-theoretic.
The autocorrelation is invariant under $H \mapsto UHU^\dagger + c\,\id$
for any $U$ commuting with $\sigma^\alpha_p$.
This is one source of non-identifiability,
but not the only one (for example $H \mapsto -H$ also leaves $C_\alpha$ unchanged).
The exact non-parametric object determined here is $\mu_\alpha$, equivalently $\Jop$.
Recovering individual couplings, therefore,
requires an additional structural assumption, a model form for $H$,
and transitions with $\langle E_j|\sigma^\alpha_p|E_k\rangle = 0$
remain invisible in this channel at any temperature.
Neither limit comes from the estimator,
and both persist however long one measures.

The second limit is quantitative and concerns resolution.
The full measure of a generic many-body $H$ has $O(D^2)$ atoms whose minimum gap
and minimum weight shrink exponentially with $N$ (Table~\ref{tab:spectrum_scaling}),
so resolving the spectrum atom by atom is exponentially hard for any protocol.
Theorem~\ref{thm:sample_complexity} is accordingly a statement
about a resolved measure satisfying the stated gap and weight assumptions.
When such a measure is used as a coarse-grained proxy for a denser physical measure,
the induced time-domain approximation error must also be controlled.
The leading Jacobi coefficients and $\tauop(m)$ depend on $\mu_\alpha$ only through its low moments,
so they are stable when the resolved approximation preserves those low moments,
and they are recovered accurately in a window
that resolves only a few hundred of the $888$ atoms at $N = 6$
(Sec.~\ref{sec:sample_complexity_stat}).

Two further conditions are demanded of the platform rather than of $H$.
First, the bath must be at, or near, infinite temperature.
At $\beta > 0$ the measured signal carries an unknown,
$H$-dependent reweighting of $\mu_\alpha$,
negligible on NV-center and donor-spin sensing platforms and exactly absent under bitstring randomization,
but a structural limitation for a genuinely thermal state unless $\beta$ can be tuned
(Sec.~\ref{sec:finite_temp}).
Second, the dynamics must stay closed over the observation window.
The reduction of Arnoldi to Lanczos, hence the Jacobi rather than Hessenberg structure,
rests on Hermiticity of the Liouvillian, which in the laboratory means $T_2 \gtrsim T_C$
(Sec.~\ref{sec:implementation}).

\subsection{Open directions}
\label{sec:disc_open}

An open direction is robustness to model mismatch.
Figure~\ref{fig:robust} tests the pipeline against shot noise and finite bath sampling,
but not against a true Hamiltonian containing interactions outside an assumed class.
Ref.~\cite{ghosh2025local} quantifies exactly this for recursive boundary-probe estimation.
Quantifying how the estimated $\hat\mu_\alpha$ responds to such mismatch
is a well-posed question we leave open.

The extension to open systems requires more than a change of estimator.
Under Lindblad dynamics, the poles of $s_\alpha(t)$ move from $i\omega_\nu$
with real $\omega_\nu$ to $-\gamma_k + i\omega_k$ with decay rates $\gamma_k > 0$.
The matrix pencil handles damped exponentials without difficulty,
but the Krylov structure itself changes.
A non-Hermitian Lindbladian requires Arnoldi rather than Lanczos and
returns a Hessenberg rather than a Jacobi matrix~\cite{bhattacharya2022open,liu2023krylov},
so the orthogonal-polynomial framework used throughout this paper,
and with it the identification of the signal with a positive measure,
does not carry over directly.

Two further theoretical directions are more incremental.
They are sharpening the spectral-gap dependence of the sample complexity
and using multi-site probes to enlarge the fraction of the Liouvillian spectrum resolved.
On the experimental side,
a full hardware demonstration is the next logical step.
As discussed in Sec.~\ref{sec:intro}, hole-spin qubits in Si and Ge nanostructures,
with their intrinsically anisotropic exchange coupling,
are a promising candidate platform for such a demonstration.

Finally, two connections go beyond Hamiltonian learning.
First, the coefficients $\betaop_n$ are formally the growth rate of Krylov complexity
in the sense of the universal operator-growth hypothesis~\cite{parker2019,barbon2019,Rabinovici2021,Rabinovici2025review,nandy2025krylov},
which for chaotic systems predicts asymptotic linear growth $\betaop_n \sim n$~\cite{rabinovici2022,dymarsky2021}.
Our $N = 6$ data show $\betaop_n$ peaking near $n \approx 10$ and then falling,
which we interpret as finite-size saturation rather than a test of the hypothesis.
Establishing the asymptotic regime requires $N \gg 6$, beyond exact diagonalization.
We note this connection because Theorem~\ref{thm:oracle_free} makes $\{\betaop_n\}$
measurable from probe-only data, which in principle opens operator growth to experiment.
Second, Ref.~\cite{matsuura2026hp} uses the present construction
as the spectral-learning stage of a coherent state-transfer protocol
built on the Hayden--Preskill--Yoshida--Kitaev recovery circuit,
where the same gauge freedom decides which part of an absorbed state
can be recovered without further structural assumptions.

\section{Conclusion}
\label{sec:conclusion}

We have introduced \emph{Quantum Krylov Learning (QKL)}
as a framework for Hamiltonian learning under restricted access.
The QKL framework uses a single-temperature quantum-probe autocorrelation,
organized through its operator Krylov structure,
to construct model-independent spectral data.
Its spectral-learning core is organized around the operator Krylov Jacobi matrix $\Jop$,
constructed from the infinite-temperature autocorrelation of a single probe observable.
The matrix $\Jop$ is reconstructed for any Hamiltonian, without a symmetry assumption,
exactly in the noiseless limit (Theorem~\ref{thm:oracle_free}) and,
in the finite-shot regime, with a $\mathrm{poly}(1/\varepsilon)$
end-to-end sample-complexity guarantee at fixed spectral resolution,
with explicit polynomial degrees
(Theorem~\ref{thm:sample_complexity}, Eq.~\eqref{eq:explicit_exponents}).
It obeys a Lieb--Robinson bound and a transport formula,
and the transport formula gives a data-driven criterion for the observation window.
The infinite-temperature requirement is met naturally on NV-center and donor-spin platforms,
and the error from finite temperature is controlled and subdominant to finite-window
and shot-noise effects (Sec.~\ref{sec:finite_temp}).
On the $N = 6$ XXZ benchmark, the fully data-driven pipeline
reproduces the leading oracle Jacobi coefficients to within a few percent
(the higher coefficients within the present window to $10$--$34\%$)
and the window scale $\tauop$ to $0.9$--$6.6\%$.
Together with the state Krylov framework of Ref.~\cite{matsuura2026state},
QKL provides a unified, oracle-free Krylov approach to Hamiltonian learning under restricted access,
realizable with free time evolution and single-probe measurement alone.

\appendix

\section{Non-Measurability of the Survival Amplitude under Probe-Only Access}
\label{app:survival}

Here we show the two special properties that make the BMN survival-amplitude measurement possible,
and see where they fail for a generic Hamiltonian.
Let
\begin{equation}
|\mathbf{0}  \rangle = |0\rangle_p \otimes |\mathbf{0}\rangle_b, \qquad
|\mathbf{1}_1\rangle = |1\rangle_p \otimes |\mathbf{0}\rangle_b,
\end{equation}
with the first factor being the accessible boundary probe.

\paragraph{Local population as a global return probability.}
If $[H, \hat{N}] = 0$ and the initial state contains one excitation,
the dynamics stays in the one-excitation sector,
\begin{equation}
e^{-iHt}|\mathbf{1}_1\rangle = \sum_j a_j(t)|\mathbf{1}_j\rangle, \qquad
\mathcal{A}(t) = \langle\mathbf{1}_1|e^{-iHt}|\mathbf{1}_1\rangle = a_1(t),
\end{equation}
where $|\mathbf{1}_j\rangle$ is a single-excited state at the site $j$.
Because there is exactly one excitation,
finding the probe excited means that the excitation has returned to the original site~1.
Hence
\begin{equation}
P_p(t) = \langle n_p(t)\rangle = |a_1(t)|^2 = |\mathcal{A}(t)|^2.
\label{eq:app_probe_population_equals_survival}
\end{equation}
This is the first special ingredient of the BMN setting.
A local probe-population measurement equals the return probability of a particular global many-body state.

\paragraph{The vacuum as an internal phase reference.}
The modulus $|\mathcal{A}(t)|$ is not enough to determine the complex amplitude.
In the number-conserving case, however, the zero-excitation sector is one-dimensional,
so
\begin{equation}
H|\mathbf{0}\rangle = E_0|\mathbf{0}\rangle.
\end{equation}
A probe-only rotation prepares
$(|\mathbf{0}\rangle + |\mathbf{1}_1\rangle) / \sqrt{2}$,
and after free evolution the transverse probe coherence contains
\begin{equation}
e^{i E_0 t} \langle \mathbf{1}_1 | e^{-iHt} | \mathbf{1}_1 \rangle = e^{i E_0 t} \mathcal{A}(t).
\label{eq:app_bmn_cross_term}
\end{equation}
Measurements in two transverse probe bases therefore determine the real and imaginary parts of $\mathcal{A}(t)$
relative to the stationary vacuum. This is the second special ingredient.

\paragraph{Generic Hamiltonian.}
When number-nonconserving terms are present,
the evolved state is no longer restricted to the one-excitation sector.
A probe-population measurement instead gives
\begin{equation}
P_p(t) = \sum_b \left| \langle 1, b | e^{-iHt} | 1, \mathbf{0}_b \rangle \right|^2,
\label{eq:app_general_probe_population}
\end{equation}
whereas the many-body survival probability is only the single contribution
with $b = \mathbf{0}_b$,
\begin{equation}
|\mathcal{A}(t)|^2 = \left| \langle 1, \mathbf{0}_b | e^{-iHt} | 1, \mathbf{0}_b \rangle\right|^2.
\end{equation}
Thus probe population no longer determines even the modulus of the global survival amplitude.
Independently, the vacuum is not generically an eigenstate,
so the probe-only experiment has no guaranteed stationary internal branch against which to read the phase of $\mathcal{A}(t)$.
A controlled-$e^{-iHt}$ Hadamard test, or an independently known probe-preparable eigenstate,
could supply such a reference, but neither is assumed in the restricted-access model of this work.

These are the only points needed from survival-amplitude spectroscopy here.
The main construction therefore adopts the probe autocorrelation $C_\alpha(t)$
defined in Eq.~\eqref{eq:two_state}.
Its probe-only measurement and its identification with the operator Krylov spectral measure
are developed in Secs.~\ref{sec:setup} and~\ref{sec:oracle_free}.

\section{Jacobi Structure, Lieb--Robinson Bound, and Transport Formula}
\label{app:jacobi}

This appendix provides the formal proofs for the structural and dynamical properties
of the operator Krylov chain presented in Sec.~\ref{sec:jacobi}.
We first prove the exact Jacobi structure and the unconditional vanishing
of the diagonal coefficients (Theorem~\ref{thm:jacobi}),
and then derive the operator Lieb--Robinson bound (Theorem~\ref{thm:LR})
and the associated transport formula.

\subsection{Proof of Theorem~\ref{thm:jacobi}}

Tridiagonality is immediate. Arnoldi applied to the Hermitian generator
$\mathcal{L}$ (Eq.~\eqref{eq:L_hermitian}) is the symmetric Lanczos process,
and its matrix is Jacobi~\cite{saad2003}.
It remains to prove that $\alphaop_n = 0$ holds unconditionally.

\paragraph{Trace proof.}
Write $Q_0 = \sigma^\alpha_p/\|\sigma^\alpha_p\|_{\mathrm{HS}}$
with $\sigma^\alpha_p = (\sigma^\alpha_p)^\dagger$.
Since $\mathcal{L} = [H,\cdot\,]$,
the adjoint parity of the Krylov vectors alternates,
$Q_n^\dagger = (-1)^n Q_n$.
This is proved by induction.
If $Q_n^\dagger = (-1)^n Q_n$ and $Q_{n-1}^\dagger = (-1)^{n-1} Q_{n-1}$,
then $\mathcal{L}Q_n = [H,Q_n]$ has adjoint parity $(-1)^{n+1}$,
and once $\alphaop_n = 0$ is established at step $n$,
the residual defining $Q_{n+1}$ inherits the same parity.
The base case is $Q_0^\dagger = Q_0$.
Given this parity, the diagonal coefficient is
\begin{align}
\alphaop_n
&= (Q_n|\mathcal{L}Q_n)
 = \frac{1}{D} \tr \bigl( Q_n^\dagger[H,Q_n] \bigr)
 = \frac{(-1)^n}{D} \tr \bigl( Q_n[H,Q_n] \bigr) \notag \\
&= \frac{(-1)^n}{D} \bigl( \tr (HQ_n^2) - \tr (Q_n^2H) \bigr)
 = 0
\label{eq:alpha_trace}
\end{align}
by cyclicity of the trace.
At the first step, $\alphaop_0 = (Q_0|[H,Q_0]) = 0$ holds by the same identity,
so the induction and the vanishing of $\alphaop_n$ hold together at every step.

\paragraph{Spectral-measure proof.}
In the energy eigenoperator basis,
$\mathcal{L}$ is diagonal with the Bohr frequencies $\omega_{jk}=E_j-E_k$,
and $C_\alpha(t) = (\sigma^\alpha_p|e^{i\mathcal{L}t}|\sigma^\alpha_p)$
is the characteristic function of the spectral measure $\mu_\alpha$,
which is even because $\sigma^\alpha_p$ is Hermitian.
The Lanczos iterates $Q_n$ correspond to the orthonormal polynomials $\pi_n(\omega)$ of $\mu_\alpha$,
and evenness of $\mu_\alpha$ forces a definite parity $\pi_n(-\omega) = (-1)^n\pi_n(\omega)$~\cite{gautschi2004}.
Hence $\omega\,\pi_n(\omega)^2$ is odd and $\alphaop_n = \int\omega\,\pi_n(\omega)^2 d\mu_\alpha = 0$.
The two proofs are the same statement in two representations.
Positivity $\betaop_n>0$ follows as in the standard Lanczos process. \qed

\subsection{Lieb--Robinson bound and transport formula}

The Lieb--Robinson bound of Theorem~\ref{thm:LR} is proved exactly
as its state Krylov counterpart in Ref.~\cite{matsuura2026state},
applied to $\Jop$ with hopping coefficients $\betaop_n$ in place of $\beta_n$.
Conjugate $\Jop$ by the weighted shift $e^{\lambda\rho^{\mathrm{op}}(\cdot,0)}$,
apply a finite Schur test to the conjugated generator,
and integrate the resulting differential inequality by Gr\"onwall's lemma.
The velocity $v_{\mathrm{LR}}^{\mathrm{op}} = 2$ follows from
$\lim_{\lambda \to 0} \betaop_j\sinh(\lambda/\betaop_j)/\lambda = 1$
for $\betaop_j > 0$.

For the transport formula, decompose the full generator as
$i\mathcal{L} = i(\Jop\oplus\mathcal{L}_\perp)+iV_m$, where the boundary term
$V_m = \betaop_{m} \bigl( |Q_{m})(Q_{m-1}| + \mathrm{h.c.} \bigr)$
couples the last retained site $m-1$ to the complement.
Note that $\Kop_m = \mathrm{span}\{Q_0,\dots,Q_{m-1}\}$
carries the hoppings $\betaop_1,\dots,\betaop_{m-1}$,
so the first neglected coupling is $\betaop_m$.
Duhamel's principle applied to the return element
$C_\alpha(t)-C^{(m)}_\alpha(t) = (Q_0|(e^{i\mathcal{L}t}-e^{i\Jop t}P_m)|Q_0)$
(with $P_m$ the projector onto $\mathcal{K}_m$) gives
\begin{equation}
\epsilon^{\mathrm{op}}_m(T) \leq
\betaop_{m} \int_0^T
\bigl| (Q_0|e^{i\mathcal{L}(T-s)}|Q_{m}) \bigr| \,
\bigl| (Q_{m-1}|e^{i\Jop s}|Q_0) \bigr| \, ds,
\label{eq:duhamel}
\end{equation}
a product of \emph{two} propagation amplitudes,
namely one from site $0$ out to site $m-1$ (time $s$)
and one from site $m$ back to site $0$ (time $T-s$).
Inserting the Lieb--Robinson bound~\eqref{eq:LR} into each factor shows that the integrand is
exponentially small unless both $s \gtrsim \rho^{\mathrm{op}}(m-1,0)/2$
and $T-s$ $\gtrsim \rho^{\mathrm{op}}(m,0)/2 > \rho^{\mathrm{op}}(m-1,0)/2$,
i.e., unless $T$ $\gtrsim 2\rho^{\mathrm{op}}(m-1,0)/v_{\mathrm{LR}}^{\mathrm{op}} = \tauop(m)$,
where $\rho^{\mathrm{op}}(m-1,0) = \sum_{j=1}^{m-1} 1/\betaop_j = \tauop(m)$
exactly matches the convention of Eq.~\eqref{eq:transport}.
The return signal has to traverse the chain \emph{twice} at velocity $2$.
This is the origin of the slope $a \approx 1$ of the round-trip transport formula~\eqref{eq:transport}.
A one-way matrix element $(Q_{m-1}|e^{i\Jop t}|Q_0)$
crosses the distance once and becomes appreciable already at $t \approx \tauop(m)/2$,
consistent with the slope $a \approx 0.52$ observed when the break time
is defined through the one-way state-signal error.
Fitting the resulting break time $t_{\ast}^{\mathrm{op}}(m)$
against $\tauop(m)$ gives slope $a \approx 0.95$--$1.00$
and a negligibly small intercept $b \approx 0$
($R^2 \approx 0.997$--$0.999$ across $g_x \in \{0, 0.2, 0.4\}$,
see Fig.~\ref{fig:transport} in Sec.~\ref{sec:numerics_structure}),
confirming the round-trip cancellation directly in operator space
to the same precision as the state Krylov result of Ref.~\cite{matsuura2026state}.

\section{Determination of the Jacobi Matrix from the Spectral Measure}
\label{app:measure_jacobi}

This appendix shows that the spectral measure $\mu_\alpha$
determines the operator Krylov Jacobi matrix $\Jop$ uniquely
(using no other information about $H$).
This is the basis of the claim in Sec.~\ref{sec:jacobi}
that the coefficients are a function of $\mu_\alpha$ alone,
and it is also the justification of the second stage of the two-stage construction
in Sec.~\ref{sec:oracle_free_construction} (the step from measure to Jacobi matrix).

The Lanczos iteration started from
$Q_0 = \sigma^\alpha_p/\|\sigma^\alpha_p\|_{\mathrm{HS}}$ has the structure
$Q_n \propto \sum_{j,k}r_n(\omega_{jk})M^\alpha_{jk}|E_j\rangle\langle E_k|$
with real polynomials $r_n$.
Hence every inner product entering the recursion is
\begin{equation}
(Q_n|f(\mathcal{L})|Q_{n'}) = \int r_n(\omega)\,f(\omega)\,r_{n'}(\omega)\,d\mu_\alpha(\omega),
\label{eq:inner_prod}
\end{equation}
an integral against $\mu_\alpha$ alone.
The phases of $M^\alpha_{jk}$, the individual magnitudes within a degenerate frequency class,
and the eigenvectors never appear.
Therefore the recursion is the Stieltjes procedure for the orthonormal polynomials $\{\pi_n\}$ of $\mu_\alpha$,
and the three-term recurrence
\begin{equation}
\omega\,\pi_n(\omega) = \betaop_{n+1}\pi_{n+1}(\omega) + \alphaop_n\pi_n(\omega) + \betaop_n\pi_{n-1}(\omega), \qquad
\int\pi_n\pi_{n'}\,d\mu_\alpha = \delta_{nn'},
\label{eq:three_term}
\end{equation}
is unique by Favard's theorem~\cite{chihara1978,Simon1998,gautschi2004}.
In particular, Lanczos applied to $(\diag(\omega_\nu),(\sqrt{p_\nu})_\nu)$
reproduces $\Jop$ exactly using only the atoms and weights of $\mu_\alpha$.

\section{Recursion Method and Continued-Fraction Representation}
\label{app:recursion}

In this appendix,
we describe the recursion method touched on in Sec.~\ref{sec:oracle_free} in more detail,
and explain how it gives the autocorrelation $C_\alpha(t)$
as a continued fraction and how the construction of this paper is its inverse map.

The recursion method of condensed-matter physics~\cite{viswanath1994,parker2019}
starts from a given Hamiltonian $H$ and a seed operator $\sigma^\alpha_p$ and
generates the Jacobi coefficients $\{\alphaop_n,\betaop_n\}$
(with $\alphaop_n = 0$ here) by the Lanczos / Gram--Schmidt procedure
for the Liouvillian $\mathcal{L} = [H,\cdot\,]$.
These coefficients connect to the autocorrelation through the resolvent (Green function),
\begin{equation}
G(z)
= \bigl( \sigma^\alpha_p \, \big| \, (z-\mathcal{L})^{-1} \, \big| \, \sigma^\alpha_p \bigr)
= \int \frac{d\mu_\alpha(\omega)}{z-\omega}, \qquad \mathrm{Im}\,z
\neq 0,
\label{eq:resolvent}
\end{equation}
which is the Stieltjes transform of the spectral measure $\mu_\alpha$.
The function $G(z)$ corresponds to the one-sided Laplace transform of the autocorrelation
$C_\alpha(t) = \int e^{i\omega t}d\mu_\alpha(\omega)$.

The key point of the recursion method is that $G(z)$ is represented exactly as
a continued fraction through the Jacobi coefficients,
\begin{equation}
G(z)
= \cfrac{1}{\,z-\alphaop_0
- \cfrac{(\betaop_1)^2}{\,z-\alphaop_1
- \cfrac{(\betaop_2)^2}{\,z-\alphaop_2-\cdots}}} \, .
\label{eq:continued_fraction}
\end{equation}
In this paper the seed $\sigma^\alpha_p$ is Hermitian and the measure $\mu_\alpha$ is even,
so all diagonal coefficients vanish, $\alphaop_n = 0$,
and the continued fraction reduces to
\begin{equation}
G(z) = \cfrac{1}{\,z-\cfrac{(\betaop_1)^2}{\,z-\cfrac{(\betaop_2)^2}{\,z-\cdots}}} \, .
\label{eq:continued_fraction_even}
\end{equation}
Truncating at step $m$ gives a finite continued fraction,
which equals the resolvent $(\sigma^\alpha_p|(z-\Jop)^{-1}|\sigma^\alpha_p)$
of the finite Jacobi matrix $\Jop$, that is, the $(0,0)$ matrix element of $\Jop$.
The $n$-th level of the continued fraction therefore corresponds directly
to the hopping $\betaop_n$ at site $n$ of the Krylov chain.

The ordinary recursion method uses this relation in the \emph{forward} direction, starting from a known $H$.
It first obtains the Jacobi coefficients $\{\betaop_n\}$,
then builds the resolvent $G(z)$ as a continued fraction,
and from that recovers the autocorrelation $C_\alpha(t)$.
This efficiently resums the autocorrelation and spectral functions from a few Jacobi coefficients.
The higher the continued fraction is truncated,
the better $C_\alpha(t)$ is reproduced on the low-frequency (long-time) side.

The oracle-free construction of this paper runs this relation in the \emph{reverse} direction.
Starting from the measured $C_\alpha(t)$ of an unknown $H$,
it recovers the spectral measure $\mu_\alpha$ as its characteristic function (matrix pencil / Prony),
and obtains the Jacobi coefficients $\{\betaop_n\}$, hence $\Jop$,
by the Stieltjes procedure for $\mu_\alpha$
(equivalently, reading off the coefficients of the continued-fraction expansion~\eqref{eq:continued_fraction_even}).
The essential difference is that the forward recursion method requires $H$ as input,
whereas the reverse construction here takes only the measurable $C_\alpha(t)$ as input.
That the low levels of the continued fraction (the low-order $\betaop_n$) are determined
stably within a finite observation window is consistent with the claim in Sec.~\ref{sec:oracle_free}
that $\Jop$ depends only on the low-order moments of $\mu_\alpha$.

\section{Proof of Theorem~\ref{thm:oracle_free}}
\label{app:oracle_free_proof}

Theorem~\ref{thm:oracle_free} consists of three independent claims.
(I) the measure recovered from the measured time series equals the true spectral measure,
$\hat\mu_\alpha = \mu_\alpha$,
(II) the Stieltjes/Lanczos procedure applied to
that measure reproduces the oracle Jacobi matrix $\Jop$ exactly, $\hat\Jop = \Jop$, and
(III) the whole construction can be carried out with probe-only operations and uses no property of $H$.
We prove each claim in turn. For notation, write the spectral measure as
$\mu_\alpha = \sum_{\nu=1}^{s}p_\nu\delta_{\omega_\nu}$,
its support size as $s = |\mathrm{supp}\,\mu_\alpha|$
(since $\mu_\alpha$ is even the frequencies consist of $0$ and positive/negative pairs),
and the measured samples as $\hat c_k := C_\alpha(t_k)$, $t_k = k\Delta t$.

\paragraph{Claim (I), exact recovery of the measure.}
From Eq.~\eqref{eq:C_measure}, the samples are a sum of $s$ complex exponentials,
\begin{equation}
\hat c_k = \sum_{\nu=1}^{s} p_\nu\, z_\nu^{\,k}, \qquad
z_\nu := e^{i\omega_\nu \Delta t}.
\label{eq:prony_form}
\end{equation}
The Nyquist condition (a) $\Delta t<\pi/\omega_{\max}$ ensures $\omega_\nu\Delta t\in(-\pi,\pi)$,
so the map $\omega\mapsto e^{i\omega\Delta t}$ is injective
and the nodes $z_\nu$ are distinct points on the unit circle.
From $2s$ samples, forming the two Hankel matrices
\begin{equation}
H_0 = [\hat c_{i+j}]_{i,j=0}^{s-1}, \qquad
H_1 = [\hat c_{i+j+1}]_{i,j=0}^{s-1}
\label{eq:hankel_pencil}
\end{equation}
and substituting Eq.~\eqref{eq:prony_form} factorizes them as
$H_0 = V\,\mathrm{diag}(p_\nu)\,V^\mathrm{T}$ and
$H_1 = V\,\mathrm{diag}(p_\nu)\,\mathrm{diag}(z_\nu)\,V^\mathrm{T}$, where
$V = [z_\nu^{\,i}]_{i,\nu}$ is a Vandermonde matrix.
Since the nodes are distinct and the weights nonzero, $V$ is nonsingular,
hence so is $H_0$, and the matrix pencil
$H_1-z\,H_0 = V\,\mathrm{diag}(p_\nu)\,\mathrm{diag}(z_\nu-z)\,V^\mathrm{T}$
degenerates exactly at $z\in\{z_\nu\}$.
Therefore the eigenvalues of the generalized eigenvalue problem
\begin{equation}
H_1\, u = z\, H_0\, u
\label{eq:gen_eig}
\end{equation}
give precisely the node set $\{z_\nu\}$,
and $\omega_\nu=\arg(z_\nu)/\Delta t$ is determined uniquely (Prony's uniqueness theorem~\cite{hua1990}).
The weights are then determined uniquely by solving the Vandermonde linear system with the obtained nodes,
\begin{equation}
\sum_{\nu=1}^{s} p_\nu\, z_\nu^{\,k} = \hat c_k \quad (k = 0, \dots, 2s-1).
\label{eq:vandermonde_weights}
\end{equation}
Condition (b), $K+1\ge 2s$, guarantees the required $2s$ samples,
and the noiseless condition (c) introduces no error.
Hence
$\hat\mu_\alpha = \sum_\nu p_\nu\delta_{\omega_\nu} = \mu_\alpha$ exactly. \qed

\paragraph{Claim (II), exact reproduction of the Jacobi matrix.}
As shown in Appendix~\ref{app:measure_jacobi},
every inner product entering the Lanczos recursion reduces to an integral against the measure $\mu_\alpha$,
and the coefficients $\{\alphaop_n,\betaop_n\}$ are determined uniquely
by the three-term recurrence~\eqref{eq:three_term} of the orthogonal polynomials of $\mu_\alpha$.
Equivalently, these coefficients are functions of the moments of the measure,
\begin{equation}
m_\ell
= \int \omega^\ell\, d\mu_\alpha(\omega)
= \sum_\nu p_\nu\,\omega_\nu^{\,\ell}
= (-i)^\ell \frac{d^\ell}{dt^\ell}C_\alpha(t)\Big|_{t=0} \, ,
\label{eq:moment_match}
\end{equation}
and the Jacobi matrix up to step $m$ depends only on the first $2m$ moments $m_0,\dots,m_{2m-1}$
(Cholesky factorization of the moment Hankel matrix, equivalently the Stieltjes procedure).
Oracle Lanczos generates the same coefficients
by applying $\mathcal{L}=[H,\cdot\,]$ repeatedly to the seed $\sigma^\alpha_p$,
and its result also depends only on the same moments~\eqref{eq:moment_match} of the same measure $\mu_\alpha$.
By claim (I) $\hat\mu_\alpha = \mu_\alpha$, so the moments match completely,
and the Stieltjes procedure applied to $\hat\mu_\alpha$ gives exactly the same
coefficients as oracle Lanczos, hence $\hat\Jop = \Jop$ for all $m \le s$. \qed

\paragraph{Claim (III), probe-only measurability.}
The input to claims (I) and (II) is only the samples of the autocorrelation $C_\alpha(t)$.
That $C_\alpha(t)$ itself can be measured with the probe alone follows from
the unobserved part $\mathcal{H}_b$ being at infinite temperature.
The difference of the prepared states
$\rho_\pm = |\pm_\alpha\rangle\langle\pm_\alpha|_p\otimes\id_{2\ldots N}/2^{N-1}$ is
\begin{equation}
\rho_{+} - \rho_{-} = \sigma^\alpha_p \otimes \frac{\id_{2\ldots N}}{2^{N-1}},
\label{eq:rho_diff}
\end{equation}
proportional to the identity on the $\mathcal{H}_b$ part.
Hence, the difference of the two probe signals is
\begin{equation}
  \frac{1}{2} \left[ s_\alpha(t|\rho_+) - s_\alpha(t|\rho_-) \right]
= \frac{1}{2} \tr \left[ (\rho_+-\rho_-)\,\sigma^\alpha_p(t) \right]
= \frac{1}{2^{N}} \tr \left[ \sigma^\alpha_p\,\sigma^\alpha_p(t) \right]
= \frac{1}{D} \tr \left[ \sigma^\alpha_p(t)\sigma^\alpha_p \right]
= C_\alpha(t),
\label{eq:protocol_identity}
\end{equation}
where we factored out the normalization $1/2^{N-1}$
associated with the identity $\id_{2\ldots N}$
to obtain the overall prefactor $1/2^N = 1/D$.
Each $s_\alpha(t|\rho_\pm) = \tr [\rho_\pm\sigma^\alpha_p(t)]$ is obtained just by preparing $\rho_\pm$,
evolving freely under the unknown $H$, and measuring $\sigma^\alpha$ on the probe.
No matrix element, eigenvector, or commutator $[H,\cdot\,]$ of $H$ is used,
and claims (I) and (II) are both classical post-processing of the measured $\{\hat c_k\}$.
Hence the whole construction is oracle-free. \qed

The relaxation of condition (b) under a finite observation window is detailed in the main text
(Sec.~\ref{sec:sample_complexity_stat}, ``Practice with a finite observation window'').
In the implementation, we avoid the ill-conditioned route of building the Jacobi matrix from power moments,
and instead run the Stieltjes procedure (Lanczos)
directly on the recovered discrete measure,
which is stable~\cite{gautschi2004}.

\section{Details of the Oracle-Free Construction}
\label{app:construction}

Three points are needed to implement the oracle-free construction experimentally. We first show how $C_\alpha(t)$ can be measured using probe-only operations. We then describe the experimental pipeline and the spectral-estimation procedure.

\subsection{Classical realization of the infinite-temperature bath}

Under the infinite-temperature bath condition,
Eq.~\eqref{eq:two_state} expresses $C_\alpha(t)$ as the difference of two probe-only Bloch vectors,
each requiring only single-probe preparation, free evolution, and single-probe measurement.
This condition is physically natural in quantum sensing,
when the probe is a controllable spin defect coupled to
its nuclear-spin bath at the platform's ordinary operating temperature and no bath preparation is needed.

On programmable hardware the bath is initialized near a pure state,
and the maximally mixed bath is emulated by classical randomization.
Averaging over uniformly random computational-basis bitstrings
$|b\rangle\in\{0,1\}^{\otimes(N-1)}$ is an unbiased estimator of the infinite-temperature average,
since $\mathbb{E}_b\langle b|O|b\rangle = \tr [O]/2^{N-1}$ for any operator $O$.
Dynamical typicality~\cite{bartsch2009} gives exponential concentration for Haar-random bath states,
but that is a distinct ensemble from the computational-basis bitstrings
used here and is not invoked to assign their system-size variance.
On engineered hardware the randomization is a single layer of bit-flip gates on the bath register,
requiring no entangling gate, no feedback, and no knowledge of $H$.

\subsection{Oracle-free pipeline}

The experimental pipeline is as follows.
(1) For each time $t_k$ and sign $\pm$, place the bath at high temperature (or use a fresh random bitstring),
prepare the probe in $|\pm_\alpha\rangle$, evolve freely for $t_k$,
and measure $\sigma^\alpha$ on the probe for $M$ shots.
(2) Form
$\hat C_\alpha(t_k) = \tfrac{1}{2}[\hat s_\alpha(t_k|\rho_+)-\hat s_\alpha(t_k|\rho_-)]$.
(3) Fit $\hat C_\alpha(t_k) = q_0 + \sum_{\omega>0} q_\omega \cos(\omega t_k)$
by non-negative least squares on a frequency grid finer than the window resolution $2\pi/T_C$,
renormalize to $\sum_\nu \hat p_\nu = 1$, and symmetrize in $\pm\omega$.
This imposes $\alphaop_n = 0$ exactly.
(4) Apply the Stieltjes procedure to $(\diag(\hat\omega_\nu),(\sqrt{\hat p_\nu})_\nu)$ to obtain $\Jop$,
the coefficients $\{\betaop_n\}$, and $\tauop(m) = \sum_{n<m} 1/\betaop_n$.
Stages (3)--(4) are classical post-processing at a cost polynomial in the number of recovered poles, independent of $D$.

\section{Sample Complexity and the Proof of Theorem~\ref{thm:sample_complexity}}
\label{app:sample_complexity}

Theorem~\ref{thm:sample_complexity} establishes
the number of measurements required to learn the spectrum to precision $\varepsilon$.
The key quantity is the minimum spacing
\begin{equation}
q := \min_{\nu \neq \nu'} \, \bigl| \omega_\nu - \omega_{\nu'} \bigr|
\label{eq:q_def}
\end{equation}
of the frequencies to be resolved.
Two components separated by $q$ can be told apart
in an observation window $T$ only once $qT \gtrsim 1$.
When $q$ is $qT\lesssim1$,
the two exponentials
$e^{i\omega_\nu t}$ and $e^{i\omega_{\nu'}t}$ are nearly proportional over the whole window,
the linear system that maps the sampled time series to the frequencies is close to degenerate,
and a small error $\delta$ in the signal is amplified into a large error in the recovered frequency.
The amplification also grows as the smallest spectral weight $p_{\min}$ decreases,
since a faint line is harder to locate than a strong one.
Theorem~\ref{thm:sample_complexity} states that this amplification
is bounded by a polynomial in $1/q$ and $1/p_{\min}$,
so that it is enough to reduce $\delta$ by averaging.

Theorem~\ref{thm:sample_complexity} is proved by constructing an explicit estimator from the measured samples and bounding its shot cost.
The resulting bound is given in Eq.~\eqref{eq:explicit_exponents}.
The estimator differs from the matrix-pencil method used in Sec.~\ref{sec:numerics}, but is closely related to it (Remark~\ref{rem:pencil_relation}).
We analyze this estimator because all constants in the resulting bounds can be obtained explicitly.
The proof uses only Hoeffding's inequality and elementary calculus.
Filtering the data first reduces the multi-atom eigenvalue problem to a scalar quotient, so no subspace-perturbation or eigenvalue-perturbation theorem is required.

The argument has four steps.
\begin{itemize}
\item[(i)] \emph{Bounding signal error from finite shots.}
Determine the number of shots $M$ per sample time required to satisfy
$|\hat C_\alpha(t_k)-C_\alpha(t_k)| \le \delta$
simultaneously at all sample times.
\item[(ii)] \emph{Isolating individual spectral atoms.}
A filter that suppresses every atom but one, so that shifting the sample index
by a single step and taking a quotient returns that atom's frequency.
The separation $q$ is introduced in this step, yielding an accuracy of $O(\delta/p_{\min})$.
\item[(iii)] \emph{Estimation of spectral weights.}
With the frequencies in hand, the weights follow from a Vandermonde least-squares solve
whose conditioning is again controlled by $q$.
\item[(iv)] \emph{Total sample complexity.}
Choose $\delta$ and the observation window so that both errors are at most $\varepsilon$,
and read off the total shot count.
\end{itemize}
Steps (ii) and (iii) rest on three estimates for the filter and one for the Vandermonde matrix.
The required estimates are collected in Appendix~\ref{app:technical}.

\subsection{Construction of the estimator and proof of Theorem~\ref{thm:sample_complexity}}
\label{app:explicit}

Throughout this appendix the sampling interval is fixed to
\begin{equation}
\Delta t = \frac{\pi}{2 \, \omega_{\max}},
\label{eq:expl_dt}
\end{equation}
half the largest step allowed by condition (a) of Theorem~\ref{thm:oracle_free}.
The factor two of margin is what makes the trigonometric inequalities of
Appendix~\ref{app:technical} valid as stated.
We write the resolved measure assumed to generate the samples as
$\mu^{(q)}_\alpha = \sum_{\nu=1}^{s} p_\nu\delta_{\omega_\nu}$ with the atoms ordered,
$\omega_1 < \dots < \omega_s$, consecutive gaps at least $q$,
$|\omega_\nu| \le \omega_{\max}$, $\sum_\nu p_\nu = 1$, and $p_{\min}=\min_\nu p_\nu$,
and put $z_\nu = e^{i\omega_\nu\Delta t}$.
The measure is even, so $c_{-k}=c_k$ and negative sample indices cost no measurement.

\paragraph{Step (i), uniform signal accuracy from finite shots.}

\begin{lemma}[Shot-noise error bound]
\label{lem:hoeff_explicit}
With $M$ shots on each of $\rho_\pm$,
$P\bigl[ | \hat C_\alpha(t_k) - C_\alpha(t_k) | > \delta \bigr] \le 2 e^{-M\delta^2}$.
\end{lemma}
\begin{proof}
Write the estimator as a single average over $2M$ independent bounded variables,
$\hat C_\alpha(t_k) = \frac{1}{2M} \sum_{i=1}^{M} X^+_i - \frac{1}{2M} \sum_{i=1}^{M} X^-_i$
with $X^\pm_i \in \{-1,+1\}$.
Each of the $2M$ summands ranges over an interval of length $1/M$,
so Hoeffding's inequality~\cite{Wainwright2019} gives
$2\exp\bigl(-2\delta^2/(2M\cdot M^{-2})\bigr) = 2 e^{-M\delta^2}$.
\end{proof}

\begin{lemma}[Uniform accuracy]
\label{lem:union_explicit}
With
\begin{equation}
M = \Bigl\lceil \delta^{-2} \ln\bigl(2(N_t+2)/\eta\bigr) \Bigr\rceil
\label{eq:expl_M}
\end{equation}
shots per preparation, $|\hat c_k-c_k| \le \delta$ holds simultaneously
for all $0 \le k \le N_t + 1$ with probability at least $1-\eta$,
at a total cost $N_{\mathrm{tot}} = 2(N_t+2)M$.
\end{lemma}
\begin{proof}
Union bound over the $N_t+2$ sample times, with two preparations at each.
\end{proof}

Equation~\eqref{eq:expl_M} already fixes the exponent of $\varepsilon$.
Since $\delta \propto \varepsilon$ below,
$M \propto \varepsilon^{-2}$,
and the rest of the proof only decides the prefactor and the number of sample times.

\paragraph{Step (ii), isolating a single spectral atom.}

For integers $\ell \ge 2$ and $m_F \ge 1$,
let $\mathcal{F}_\ell$ be the Fej\'er kernel normalized to unit height,
and let the filter $\Phi$ be its $m_F$-th power,
\begin{equation}
\textstyle
\mathcal{F}_\ell(\theta)
= \Bigl| \frac{1}{\ell} \sum_{n=0}^{\ell-1} e^{in\theta} \Bigr|^2
= \sum_{|j|<\ell} f_j e^{ij\theta}, \quad
f_j = \tfrac{1}{\ell} \Bigl( 1 - \tfrac{|j|}{\ell} \Bigr), \qquad
\Phi := \mathcal{F}_\ell^{\,m_F} = \sum_{|k|\le m_F(\ell-1)} \phi_k e^{ik\theta}.
\label{eq:expl_kernel}
\end{equation}
Since $\phi$ is the $m_F$-fold convolution of the non-negative $f$ with itself,
\begin{equation}
\phi_k \ge 0, \qquad
\sum_k \phi_k = 1, \qquad
\operatorname{supp} \phi \subseteq \{|k| \le m_F(\ell-1)\},
\label{eq:expl_phi}
\end{equation}
and these three facts are all that is used.
Apply the filter to the data at a trial frequency $w$,
\begin{equation}
A_j(w) := \sum_{|k|\le m_F(\ell-1)} \phi_k \, \hat c_{k+j} \, e^{-iwk\Delta t}, \qquad j = 0,1.
\label{eq:expl_A}
\end{equation}
Substituting $c_{k+j} = \sum_\nu p_\nu z_\nu^{k+j}$
and exchanging the two finite sums gives the exact identity
\begin{equation}
A_j(w) = \sum_{\nu}p_\nu z_\nu^{\,j}\,\Phi_\nu+\xi_j, \qquad
\Phi_\nu := \Phi\bigl((\omega_\nu-w)\Delta t\bigr), \qquad
|\xi_j| \le \delta \sum_k \phi_k = \delta.
\label{eq:expl_Aexact}
\end{equation}
Here $\Phi_\nu$ is real and non-negative because $\Phi = \mathcal{F}^{m_F}$ with $\mathcal{F} \ge 0$,
and the noise bound uses only $\phi_k \ge 0$ and $\sum_k \phi_k = 1$.
The filter is an average and therefore cannot amplify a uniform data error.
Equation~\eqref{eq:expl_Aexact} represents the core mechanism of the construction.
Shifting the sample index by one multiplies each exponential by its own $z_\nu$,
so as soon as the filter isolates a single spectral atom,
the quotient $A_1/A_0$ yields $z_\nu$.
This is the matrix pencil of Sec.~\ref{sec:oracle_free_construction}
restricted to a single atom,
thereby eliminating the need for eigenvalue perturbation theory.

The filter is effective when its peak aligns with the target atom
and its tail decays sufficiently for all other atoms.
Both are quantified in Appendix~\ref{app:technical}.
Lemma~\ref{lem:expl_peak} gives $\Phi_\nu \ge 2/3$ once $w$ is within
$1/(\sqrt{m_F}L)$ of $\omega_\nu$, where $L := \ell\Delta t$,
and Lemma~\ref{lem:expl_leak} gives, for $qL \ge 2\pi e$,
\begin{equation}
\mathrm{leak}_\nu := \sum_{\nu'\ne\nu} p_{\nu'} \Phi_{\nu'} \; \le \; e^{-2m_F}.
\label{eq:expl_leak_ref}
\end{equation}
The content of Eq.~\eqref{eq:expl_leak_ref} is that the leakage is exponentially small
in the kernel order $m_F$ at a record length $L = 2\pi e/q$
that does \emph{not} grow with the target accuracy.
Improving accuracy requires only a logarithmically growing $m_F$.
A single Fej\'er kernel ($m_F=1$) would instead give leakage $\sim(qL)^{-2}$,
forcing $L \propto \delta^{-1/2}$ and a sample complexity $\varepsilon^{-5/2}$.

\begin{proposition}[Error bound for frequency estimation]
\label{prop:expl_freq}
Assume $qL \ge 2\pi e$, $|w-\omega_\nu| \le 1/(\sqrt{m_F}L)$,
$m_F \ge \frac{1}{2}\ln(\frac{1}{\delta})$ (so that $\mathrm{leak}_\nu \le \delta$)
and $\delta \le p_{\min}/12$.
Put $\hat z_\nu := A_1(w)/A_0(w)$ and $\hat\omega_\nu := \arg\hat z_\nu/\Delta t$.
Then
\begin{equation}
|\, \hat z_\nu - z_\nu \,| \le \frac{8\,\delta}{p_{\min}}, \qquad
|\, \hat\omega_\nu - \omega_\nu \,| \le \frac{16\,\omega_{\max}\,\delta}{p_{\min}}.
\label{eq:expl_freqbound}
\end{equation}
\end{proposition}
\begin{proof}
(i) By Eq.~\eqref{eq:expl_Aexact},
Lemma~\ref{lem:expl_peak} and Eq.~\eqref{eq:expl_leak_ref},
$|A_0| \ge p_\nu \Phi_\nu - \mathrm{leak}_\nu - \delta \ge \frac{2}{3} p_{\min} - 2\delta \ge \frac{p_{\min}}{2}$.
(ii) Multiplying the $j = 0$ identity by $z_\nu$ and subtracting it from the $j = 1$ one,
the $\nu$ term cancels identically, so that
$A_1 - z_\nu A_0 = \sum_{\nu' \ne \nu} p_{\nu'} (z_{\nu'} - z_\nu) \Phi_{\nu'} + (\xi_1 - z_\nu \xi_0)$.
Since $|z_{\nu'} - z_\nu| \le 2$ and $\Phi_{\nu'} \ge 0$,
$|A_1 - z_\nu A_0| \le 2\,\mathrm{leak}_\nu + 2 \delta \le 4 \delta$.
(iii) Dividing, $|\hat z_\nu - z_\nu| \le 4 \delta / (p_{\min}/2) = 8 \delta / p_{\min}$.
(iv) Let $\varphi = \arg\hat z_\nu$ and $\psi = \omega_\nu \Delta t$.
Radial projection onto the unit circle is the nearest-point map, so
$|e^{i\varphi} - \hat z_\nu| \le |z_\nu - \hat z_\nu|$ and hence
$|e^{i\varphi} - e^{i\psi}| \le 2|\hat z_\nu - z_\nu|$.
On the other hand $|e^{i\varphi} - e^{i\psi}| = 2|\sin\frac{\varphi-\psi}{2}| \ge \frac{2}{\pi}|\varphi-\psi|$.
Therefore $|\varphi-\psi| \le \pi |\hat z_\nu-z_\nu|$ and
$|\hat\omega_\nu-\omega_\nu| \le (\pi/\Delta t) |\hat z_\nu-z_\nu| = 2 \omega_{\max} |\hat z_\nu-z_\nu|$.
\end{proof}

\begin{remark}[Determination of the trial frequency]
\label{rem:expl_grid}
A grid search of $|A_0(w)|$ at spacing $1/(\sqrt{m_F}L)$ supplies $w$ and costs no additional measurements,
since it is classical post-processing of data already taken.
It is well posed.
Within $1/(\sqrt{m_F}L)$ of any atom
$|A_0| \ge \frac{2}{3} p_{\min}-2\delta \ge \frac{1}{2} p_{\min}$,
whereas at a point farther than $q/2$ from every atom the argument behind Eq.~\eqref{eq:expl_leak_ref}
applies to \emph{all} atoms and gives $|A_0| \le e^{-2m_F} + \delta \le 2 \delta \le \frac{1}{6} p_{\min}$.
Thresholding $|A_0|$ at $p_{\min}/3$ therefore separates the two cases and retains one local maximum per atom.
This constraint determines the constant $12$ in Proposition~\ref{prop:expl_freq}.
\end{remark}

\paragraph{Step (iii), estimation of the spectral weights.}

The filter alone does not deliver the weights to accuracy $\delta$,
because $A_0$ estimates $p_\nu\Phi_\nu$ and $\Phi_\nu\ge2/3$ is not close enough to $1$.
They are instead recovered together by least squares on the Vandermonde system built
from the frequencies just found, which is well conditioned once the window exceeds $O(1/q)$.
Proposition~\ref{prop:expl_vandermonde} of Appendix~\ref{app:technical} gives
$\|V^{+}\| \le \sqrt{2/(N_t + 1)}$ under the stated oversampling.

\begin{proposition}[Error bound for weight estimation]
\label{prop:expl_weights}
Let $\max_\nu|\hat\omega_\nu-\omega_\nu| \le \varepsilon_\omega \le q/4$,
let $\hat V$ be the Vandermonde matrix built from $\hat\omega$ with $N_t+1$ rows
satisfying Proposition~\ref{prop:expl_vandermonde} with $\kappa = q/2$, and set
$\hat p := \hat V^{+} \hat c$.
Then, with $T := N_t \Delta t$,
\begin{equation}
\| \hat p - p \|_2 \le \sqrt{2} \, \bigl( \delta + T \varepsilon_\omega \bigr).
\label{eq:expl_weightbound}
\end{equation}
\end{proposition}
\begin{proof}
The estimated nodes inherit consecutive gaps at least $q-2\varepsilon_\omega \ge q/2$,
so $\|\hat V^{+}\| \le \sqrt{2/(N_t+1)}$.
Writing $\hat c = V p + \xi$ with $\|\xi\|_2 \le \delta\sqrt{N_t+1}$ and using
$\hat V^{+} \hat V = I$, one has $\hat p - p = \hat V^{+} [ (V - \hat V) p + \xi]$.
For unit-modulus numbers $| \zeta^k - \hat\zeta^k | \le k| \arg\zeta - \arg\hat\zeta |$,
so each component of $(V - \hat V) p$ is at most $T \varepsilon_\omega$ and
$\|(V - \hat V) p\|_2 \le \sqrt{N_t+1} \, T \varepsilon_\omega$.
\end{proof}

The term $T\varepsilon_\omega$ is intrinsic to the problem.
A node misplaced by $\varepsilon_\omega$
accumulates a phase error $T\varepsilon_\omega$ across the record,
which an amplitude fit cannot distinguish from a change of weight.
Consequently, estimating the weights requires
two more powers of $1/q$ compared to the frequencies.

\paragraph{Step (iv), combination of bounds and total complexity.}

Collecting the hypotheses of Proposition~\ref{prop:expl_freq},
Eq.~\eqref{eq:expl_leak_ref}, and Proposition~\ref{prop:expl_weights},
\begin{equation}
\Delta t = \frac{\pi}{2\omega_{\max}}, \quad
L = \frac{2\pi e}{q}, \quad
\ell = \left\lceil \frac{L}{\Delta t} \right\rceil = \left\lceil \frac{4e\,\omega_{\max}}{q} \right\rceil, \quad
m_F = \left\lceil \frac{1}{2} \ln\left(\frac{1}{\delta}\right) \right\rceil,
\label{eq:expl_params}
\end{equation}
and the number of sampled times must satisfy both the filter support $N_t\ge m_F(\ell-1)$
and Proposition~\ref{prop:expl_vandermonde} with $\kappa = q/2$.
Both hold for
\begin{equation}
N_t = \left\lceil \frac{2 \pi \bigl( e \, m_F + 4 (1 + \ln(s)) \bigr)}{q \, \Delta t} \right\rceil
\le \frac{4 \, \Gamma \, \omega_{\max}}{q}, \qquad
\Gamma := e \, m_F + 4 (1 + \ln(s)) + 1,
\label{eq:expl_Nt}
\end{equation}
so that $T = N_t\Delta t \le 2\pi\Gamma/q$.
Requiring the frequency error of Eq.~\eqref{eq:expl_freqbound} to be at most
$\varepsilon$ gives $\delta \le p_{\min} \varepsilon / (16\,\omega_{\max})$.
Substituting Eq.~\eqref{eq:expl_freqbound} and $T \le 2\pi\Gamma/q$ into
Eq.~\eqref{eq:expl_weightbound},
\begin{equation}
\| \hat p - p \|_2
\le \sqrt{2} \, \delta \Bigl( 1 + \frac{32\pi\Gamma\omega_{\max}}{q\,p_{\min}} \Bigr)
\le \frac{147\,\Gamma\,\omega_{\max}}{q\,p_{\min}}\,\delta,
\end{equation}
absorbing the $1$ using $q\,p_{\min}\le\pi\Gamma\omega_{\max}$,
so that the weights require
$\delta \le q\,p_{\min}\varepsilon/(147\,\Gamma\,\omega_{\max})$.
Since $\Gamma>7$ and $q\le2\omega_{\max}$ the weight condition is always the binding one,
by a factor $147\Gamma/(16q) \gtrsim 64/q$.

\begin{theorem}[Theorem~\ref{thm:sample_complexity} with explicit constants]
\label{thm:explicit}
Let $\mu^{(q)}_\alpha$ be as above and $\varepsilon,\eta\in(0,1)$.
Choose
\begin{equation}
\delta = \min\Bigl\{ \frac{q\,p_{\min}\,\varepsilon}{147\,\Gamma\,\omega_{\max}}, \, \frac{p_{\min}}{12} \Bigr\},
\label{eq:expl_delta}
\end{equation}
and $\Delta t, L, \ell, m_F, N_t, M$
by Eqs.~\eqref{eq:expl_params}, \eqref{eq:expl_Nt}, and \eqref{eq:expl_M},
with $\Gamma$ fixed self-consistently.
Then the estimator
{\rm (a)} grid search of $|A_0(w)|$ at spacing $1/(\sqrt{m_F}L)$, retaining the local maxima above $p_{\min}/3$,
{\rm (b)} $\hat\omega_\nu = \arg[A_1/A_0]/\Delta t$ at each retained point, and
{\rm (c)} $\hat p = \hat V^{+}\hat c$,
returns, with probability at least $1-\eta$,
$\max_\nu|\hat\omega_\nu-\omega_\nu| \le \varepsilon$ and $\|\hat p-p\|_2 \le \varepsilon$,
using a total number of shots given by Eq.~\eqref{eq:explicit_exponents}.
In particular, Eq.~\eqref{eq:sample_complexity} holds.
\end{theorem}
\begin{proof}
Condition on the event of Lemma~\ref{lem:union_explicit}.
The choice~\eqref{eq:expl_params} gives $qL = 2\pi e$,
so Eq.~\eqref{eq:expl_leak_ref} applies,
and $m_F\ge\frac12\ln\frac1\delta$ gives
$\mathrm{leak}_\nu \le e^{-2m_F} \le \delta$,
and $\delta \le p_{\min}/12$ is imposed in Eq.~\eqref{eq:expl_delta}.
Remark~\ref{rem:expl_grid} guarantees that
step (a) returns exactly one trial point per atom and no others.
Proposition~\ref{prop:expl_freq} then bounds the frequency error by $\varepsilon$.
The hypothesis $\varepsilon_\omega \le q/4$ of Proposition~\ref{prop:expl_weights} is automatic,
since Eqs.~\eqref{eq:expl_freqbound} and~\eqref{eq:expl_delta} give
$\varepsilon_\omega \le 16 q\varepsilon/(147\Gamma) \le q\varepsilon/(9\Gamma)$
with $\Gamma > 7$ and $\varepsilon < 1$.
Proposition~\ref{prop:expl_weights} then bounds the weight error by $\varepsilon$.
Finally, by Lemma~\ref{lem:union_explicit} and Eqs.~\eqref{eq:expl_Nt}, \eqref{eq:expl_delta},
\begin{equation}
N_{\mathrm{tot}} = 2 (N_t + 2) M \le
2 \cdot \frac{4 \, \Gamma \, \omega_{\max}}{q} \cdot
\frac{147^2 \, \Gamma^2 \, \omega_{\max}^2}{q^2 \, p_{\min}^2 \, \varepsilon^2} \cdot
\ln\left( \frac{2(N_t+2)}{\eta} \right),
\end{equation}
and $8 \times 147^2 = 172{,}872 \le 1.8 \times 10^5$.
\end{proof}

\begin{remark}[Frequencies alone require fewer measurements]
\label{rem:expl_freqonly}
If only the Bohr frequencies are required to accuracy $\varepsilon$,
the weight step is dropped and Eq.~\eqref{eq:expl_freqbound} alone fixes $\delta$,
so that \\
$\delta = \min \{ p_{\min} \varepsilon / (16 \, \omega_{\max}), \, p_{\min} / 12 \}$
and the same count $N_{\mathrm{tot}} = 2(N_t+2)M$ reads
\begin{equation}
N_{\mathrm{tot}}
\le 2 \cdot \underbrace{\frac{4 \Gamma \omega_{\max}}{q}}_{N_t+2} \cdot
\underbrace{\frac{16^2 \omega_{\max}^2}{p_{\min}^2 \varepsilon^2} \cdot
\ln\left( \frac{2(N_t+2)}{\eta} \right)}_{M}
= \frac{2048 \, \Gamma \, \omega_{\max}^{3}}{q \, p_{\min}^{2} \, \varepsilon^{2}} \cdot
\ln\left( \frac{2(N_t+2)}{\eta} \right),
\label{eq:expl_freqonly}
\end{equation}
that is, $q^{-1}$ in place of $q^{-3}$ and $\Gamma$ in place of $\Gamma^3$.
\end{remark}

\begin{remark}[Relation to the matrix-pencil implementation]
\label{rem:pencil_relation}
The numerics of Sec.~\ref{sec:numerics} use the matrix pencil
and a band-limited non-negative fit rather than the filter analyzed above,
and Fig.~\ref{fig:estimator} compares them.
The same polynomial form, without explicit constants,
can be obtained for the matrix pencil along the standard approach of spectral estimation,
and we record the result below,
as it represents the standard theoretical expectation.
There the relevant matrix is the Hankel matrix $H_0 = [\hat c_{i+j}]$,
which factorizes as $V\diag(p)V^{\mathrm T}$ with $V$ of size $N_t \times 2r$ and therefore has rank $2r$.
Its smallest \emph{nonzero} singular value obeys
$\sigma_{2r}(H_0) \ge p_{\min} \sigma_{\min}(V)^2$,
with $\sigma_{\min}(V)$ bounded below in terms of $q$
by Moitra's super-resolution estimate~\cite{Moitra2015}
(which removes the $\ln(s)$ of Proposition~\ref{prop:expl_vandermonde}
at the cost of a non-elementary argument).
Wedin's $\sin\Theta$ theorem~\cite{wedin1972} then bounds the rotation of the signal subspace
extracted from the noisy Hankel matrix by $2\|E\|_2/\sigma_{2r}(H_0)$,
and Bauer--Fike~\cite{golub2009} converts that into a displacement of the pencil eigenvalues,
with the Bauer--Fike constant equal to the condition number of the Vandermonde factor,
$\kappa(V) \le \sigma_{\max}(V)/\sigma_{\min}(V) \le 2 \sqrt r$.
Composing the three gives $\max_\nu|\hat z_\nu-z_\nu| = O(\sqrt r\,\delta/p_{\min})$
and hence Eq.~\eqref{eq:sample_complexity} again.
We do not employ this approach for the rigorous bound
because explicit constants for these three constituent theorems
are not typically provided in the literature,
preventing the derivation of the closed-form bound in Eq.~\eqref{eq:explicit_exponents}.
\end{remark}

\subsection{Optimality of the scaling and genericity of the gap}
\label{app:optimal_scaling}

We show that the $\varepsilon$-scaling obtained above is information-theoretically optimal
for a single frequency, hence also as a lower bound for the joint problem,
and that the worst-case spectral degeneracy is non-generic.

\begin{proposition}[Optimality of the $1/\varepsilon^2$ scaling at fixed evolution time]
\label{prop:optimal}
Any estimator that recovers a single Bohr frequency to precision $\varepsilon$
with success probability $\geq2/3$ from probe measurements
with evolution times bounded by a fixed $t_{\max}$ requires
$N_{\mathrm{tot}} = \Omega(1/(p^2t_{\max}^2\varepsilon^2))$ shots,
where $p$ is the spectral weight of that frequency.
If instead the total evolution time is the resource and $t_{\max}$ may grow
as $1/\varepsilon$, Heisenberg-limited protocols achieve total-time scaling
$O(1/\varepsilon)$~\cite{huang2023}.
The statement above is optimality of the shot count
in the fixed-window setting relevant to our protocol.
\end{proposition}
\begin{proof}
By Le~Cam's two-point method (see, e.g., Ref.~\cite{Tsybakov2009}),
the autocorrelations of $\mu = p\delta_\omega$ and
$\mu' = p\delta_{\omega+\varepsilon}$ differ by at most $p\,t_{\max}\varepsilon$
for any $t \leq t_{\max}$, giving a total Kullback--Leibler divergence
$O(N_{\mathrm{tot}}(p\,t_{\max}\varepsilon)^2)$ over $N_{\mathrm{tot}}$ shots.
Reliable discrimination requires this to be $\Omega(1)$.
\end{proof}

Proposition~\ref{prop:optimal} is a single-frequency statement.
Since recovering the whole measure includes recovering any one of its atoms,
the same $\varepsilon^{-2}$ lower bound applies to the joint problem,
and Theorem~\ref{thm:explicit} attains it.
We do not claim optimality in $q$, $p_{\min}$, or $r$.

\begin{remark}[Genericity of the minimum gap]
\label{remark:smoothing}
Exact frequency degeneracy $q = 0$ is non-generic.
For $H + \gamma G$ with $G$ a Gaussian Unitary Ensemble (GUE) matrix and any
fixed $\gamma > 0$, eigenvalue repulsion (see, e.g., Ref.~\cite{Mehta2004})
gives, with high probability, a minimum Bohr-frequency spacing
$q \ge \gamma/\mathrm{poly}(D)$.
Two limitations of this statement should be noted.
First, it concerns the genericity of the spectrum of the given Hamiltonian.
Perturbing $H$ would change the object being learned.
Second, since $D=2^N$, a spacing of order $1/\mathrm{poly}(D)$ is still exponentially small in $N$.
Genericity therefore keeps the bound~\eqref{eq:sample_complexity} finite but does not make the full-resolution problem tractable.
This is why Theorem~\ref{thm:sample_complexity} is formulated for the resolved measure $\mu^{(q)}_\alpha$.
\end{remark}

\subsection{Technical estimates}
\label{app:technical}

This appendix collects four elementary estimates used in Appendix~\ref{app:explicit}.
Three concern the filter $\Phi=\mathcal{F}_\ell^{m_F}$ of Eq.~\eqref{eq:expl_kernel}, and one concerns the Vandermonde matrix.
The proofs use only the concavity of $\sin$ on $[0,\pi/2]$, the bound $\cos x \ge 1-x^2/2$, and the inequality $\sum_{j \le K} 1/j \le 1+\ln(K)$.

\begin{lemma}[Lower bound at the filter peak]
\label{lem:expl_peak}
If $\ell|\theta| \le 1/\sqrt{m_F}$ then
$\Phi(\theta) \ge \bigl(1-\tfrac{1}{6m_F}\bigr)^{2m_F} \ge \tfrac{2}{3}$.
\end{lemma}
\begin{proof}
Bound the Dirichlet average below by its real part and use $\cos x \ge 1-x^2/2$, which gives
$\bigl|\tfrac{1}{\ell} \sum_{n<\ell} e^{in\theta}\bigr|
\ge \tfrac{1}{\ell} \sum_{n<\ell} \bigl(1-\tfrac{n^2\theta^2}{2}\bigr)
\ge 1-\tfrac{(\ell\theta)^2}{6}$,
using $\sum_{i<\ell}i^2\le\ell^3/3$.
All terms are non-negative under the hypothesis,
so no absolute value is lost.
Squaring and raising to the $m_F$-th power gives
$\Phi \ge (1-(\ell\theta)^2/6)^{2m_F} \ge (1-1/6m_F)^{2m_F}$,
and $m_F \mapsto (1-1/6m_F)^{2m_F}$ increases
from $(5/6)^2 = 0.694$ at $m_F = 1$ to $e^{-1/3} = 0.717$.
\end{proof}

\begin{lemma}[Upper bound for the filter tail]
\label{lem:expl_tail}
If $0 < |\theta| \le \pi$ then $\Phi(\theta) \le \bigl(\pi^2/(\ell\theta)^2\bigr)^{m_F}$.
\end{lemma}
\begin{proof}
We have $\mathcal{F}_\ell(\theta) \le (\ell^2\sin^2(\theta/2))^{-1}$
because $\sin^2(\ell\theta/2) \le 1$,
and on $|\theta| \le \pi$ concavity of $\sin$ on $[0,\pi/2]$
gives $|\sin(\theta/2)| \ge |\theta|/\pi$.
Raising this inequality to the $m_F$-th power yields the result.
\end{proof}

\begin{lemma}[Bound on spectral leakage]
\label{lem:expl_leak}
Put $L := \ell \Delta t$. If $qL \ge 2\pi e$ and $|w-\omega_\nu| \le 1/(\sqrt{m_F}L)$,
then Eq.~\eqref{eq:expl_leak_ref} holds, i.e.,
$\mathrm{leak}_\nu = \sum_{\nu'\ne\nu} p_{\nu'} \Phi_{\nu'} \le e^{-2m_F}$.
\end{lemma}
\begin{proof}
Ordering of the atoms gives $|\omega_{\nu'}-\omega_\nu| \ge q|\nu-\nu'|$,
and $w$ lies within $1/(\sqrt{m_F}L) \le q/2$ of $\omega_\nu$,
so $|\omega_{\nu'}-w| \ge q|\nu-\nu'|/2$ for $\nu' \ne \nu$.
All these differences satisfy $|(\omega_{\nu'}-w)\Delta t| \le \pi$
by Eq.~\eqref{eq:expl_dt},
so Lemma~\ref{lem:expl_tail} applies and
\begin{equation}
\mathrm{leak}_\nu
\le \left( \frac{2\pi}{qL} \right)^{2 m_F} \sum_{\nu'\ne\nu}\frac{p_{\nu'}}{|\nu-\nu'|^{2m_F}}
\le \left( \frac{2\pi}{qL} \right)^{2 m_F}
\le e^{-2 m_F},
\end{equation}
using $|\nu-\nu'| \ge 1$, $\sum_{\nu'}p_{\nu'} \le 1$ and $qL \ge 2\pi e$.
\end{proof}

\begin{lemma}[Dirichlet kernel bound]
\label{lem:expl_dirichlet}
For $0 < |\theta| \le \pi$,
$\bigl|\sum_{k=0}^{N}e^{ik\theta}\bigr|
=\bigl|\sin((N+1)\theta/2)/\sin(\theta/2)\bigr|\le\pi/|\theta|$.
\end{lemma}
\begin{proof}
Summing the geometric series and taking the modulus yields the equality. \\
Applying $1/|\sin(\theta/2)| \le \pi/|\theta|$
as in Lemma~\ref{lem:expl_tail} provides the upper bound.
\end{proof}

\begin{proposition}[Smallest singular value of the Vandermonde matrix]
\label{prop:expl_vandermonde}
Let $\zeta_\nu = e^{i\lambda_\nu\Delta t}$ have ordered frequencies with consecutive gaps at least $\kappa$,
and let $V \in \mathbb{C}^{(N_t+1)\times s}$, $V_{k\nu} = \zeta_\nu^{\,k}$.
If $N_t + 1 \ge 4\pi(1+\ln(s))/(\kappa\Delta t)$ then
$\sigma_{\min}(V)^2 \ge (N_t+1)/2$, i.e.,
$\|V^{+}\| \le \sqrt{2/(N_t+1)}$.
\end{proposition}
\begin{proof}
The Gram matrix $(V^*V)_{\nu\nu'} = \sum_{k=0}^{N_t} e^{ik(\lambda_{\nu'}-\lambda_\nu)\Delta t}$
has diagonal $N_t+1$ and, by Lemma~\ref{lem:expl_dirichlet} together with
$|\lambda_\nu-\lambda_{\nu'}| \ge \kappa|\nu-\nu'|$,
$|(V^*V)_{\nu\nu'}| \le \pi/(\kappa\Delta t\,|\nu-\nu'|)$ for $\nu \ne \nu'$.
The off-diagonal row sum is therefore at most
$\frac{\pi}{\kappa\Delta t} \cdot 2 \sum_{j=1}^{s-1}\frac{1}{j} \le \frac{2\pi(1+\ln(s))}{\kappa\Delta t}$,
and Gershgorin's theorem applied to the Hermitian matrix $V^{*}V$ gives
$\lambda_{\min}(V^{*}V) \ge (N_t+1)-\max_\nu R_\nu$, which the hypothesis makes at least $(N_t+1)/2$.
\end{proof}

This is the only place where the number of atoms enters the bound,
and it enters logarithmically, through $\Gamma$.

\paragraph{Numerical verification.}
Every estimate above was checked numerically.
The empirical Hoeffding tail lies below $2e^{-M\delta^2}$ in all cells tested
($M \in \{50, 200, 1000\}$, $\delta \in \{0.05, 0.1, 0.2\}$).
The Dirichlet bound is attained as $\theta \to 0$ and cannot be improved.
Proposition~\ref{prop:expl_vandermonde} holds
with a margin of approximately a factor of two
for equally spaced nodes,
the worst case for the row sum.
For the filter, $\phi_k \ge 0$ and $\sum_k \phi_k = 1$ to machine precision,
and Lemma~\ref{lem:expl_peak} holds with true minimum $0.920$ against the claimed $0.694$.
The measured leakage at $qL = 2\pi e$ falls like $e^{-9m_F}$ rather than $e^{-2m_F}$,
because Lemma~\ref{lem:expl_tail} drops $\sin^2(\ell\theta/2) \le 1$ and
Lemma~\ref{lem:expl_leak} omits the $|\nu-\nu'|^{-2m_F}$ suppression of distant atoms.
For the estimator as a whole, on a synthetic even measure with adversarially signed data
errors the frequency error is exactly linear in $\delta$
(fitted slope $0.999$ at fixed $m_F$) and the product
$|\Delta\omega|\,p_{\min}/\delta$ is constant as $p_{\min}$ is varied,
confirming both scalings of Eq.~\eqref{eq:expl_freqbound}.
The exponents of $1/\varepsilon$ and $1/p_{\min}$ are therefore attained.
The prefactor remains loose by a factor of about $10^5$ for randomly signed noise,
and by the factor $\Gamma^3$,
which a sharper leakage estimate would reduce by roughly a factor four in $m_F$.

\section{Finite-Temperature Correction}
\label{app:finite_temp}

This appendix details the quantitative analysis of the systematic error
introduced when the unobserved bath is at a finite temperature.
We identify the physical origin of this bias,
evaluate its first-order effect on the Jacobi coefficients,
bound its propagation to the transport time,
and show how the odd-parity sector of the signal can be repurposed as an in-situ thermometer.

\subsection{Effect of finite temperature on the measured signal}

The exact reconstruction theorem relies on the bath being maximally mixed.
At a finite bath temperature, we write
\begin{equation}
\rho_b(\beta)
= \frac{e^{-\beta H_b}}{Z_b}
= \frac{\id_b}{d_b} - \frac{\beta}{d_b}\,\Theta + O(\beta^2), \qquad
\Theta := H_b - \frac{\tr H_b}{d_b} \, \id_b.
\label{eq:ftc_gibbs}
\end{equation}
The two-state protocol therefore measures
\begin{equation}
\tilde C_\alpha(t;\beta)
= \tr \Bigl[ \tfrac{1}{2} \, \sigma^\alpha_p \otimes \rho_b(\beta) \; \sigma^\alpha_p(t) \Bigr]
= C_\alpha(t) - \frac{\beta}{D} \, \tr \bigl[ (\sigma^\alpha_p \otimes \Theta) \, \sigma^\alpha_p(t) \bigr] + O(\beta^2).
\end{equation}
In the energy eigenbasis
$\tilde C_\alpha(t;\beta) = \sum_{jk} w_{jk}(\beta)\,e^{i\omega_{jk}t}$ with
\begin{equation}
w_{jk}(\beta) = \frac{|M^\alpha_{jk}|^2}{D}
- \frac{\beta}{D}\,\langle E_k|\sigma^\alpha_p\otimes\Theta|E_j\rangle\,M^\alpha_{jk}
+ O(\beta^2).
\label{eq:ftc_weights}
\end{equation}
The measured signal is therefore the characteristic function of a \emph{deformed} measure.
Writing
\begin{equation}
\tilde\mu^{(\beta)}_\alpha := \sum_{j,k} w_{jk}(\beta)\,\delta_{\omega_{jk}}, \qquad
\tilde C_\alpha(t;\beta) = \int e^{i\omega t}\,d\tilde\mu^{(\beta)}_\alpha(\omega),
\label{eq:ftc_measure_def}
\end{equation}
Eq.~\eqref{eq:ftc_weights} gives its expansion in $\beta$,
\begin{equation}
\tilde\mu^{(\beta)}_\alpha = \mu_\alpha + \beta \, \delta\mu_\alpha + O(\beta^2), \qquad
\delta\mu_\alpha := -\frac{1}{D}\sum_{j,k} \langle E_k|\sigma^\alpha_p \otimes \Theta | E_j \rangle \, M^\alpha_{jk} \, \delta_{\omega_{jk}}.
\label{eq:ftc_dmu_def}
\end{equation}
The atoms of $\delta\mu_\alpha$ sit at the same Bohr frequencies as those of $\mu_\alpha$.
Only the weights differ, and the total mass vanishes,
$\delta\mu_\alpha(\mathbb{R}) = 0$,
because $\tilde C_\alpha(0;\beta) = 1$ exactly at every $\beta$.

Thus finite temperature changes only the weights of the Bohr frequencies and leaves their positions unchanged.
The crucial point is that the $O(\beta)$ correction depends on matrix elements of
$\sigma^\alpha_p\otimes\Theta$,
which are \emph{not} determined by the Hilbert--Schmidt spectral measure $\mu_\alpha$.
Hence, for a nontrivial bath Hamiltonian,
the measured finite-temperature signal is not the characteristic function of $\mu_\alpha$,
and achieving exact reconstruction via Theorem~\ref{thm:oracle_free}
requires additional information or control.
This is a limitation of the access model and does not affect the Bohr frequencies themselves.
Single-temperature, single-channel data do not determine the correction,
because the $O(\beta)$ term contains matrix elements that $\mu_\alpha$ does not fix.
The three ways of removing the difficulty (tunable $\beta$, independent knowledge of $H_b$,
and bitstring randomization) are listed in Sec.~\ref{sec:finite_temp}.

\subsection{First-order effect on the Jacobi coefficients}

Decompose the correction measure of Eq.~\eqref{eq:ftc_dmu_def} as
$\delta\mu_\alpha = \delta\mu_{\mathrm e} + i\,\delta\mu_{\mathrm o}$, where
$\delta\mu_{\mathrm e}$ is even and $\delta\mu_{\mathrm o}$ is odd.
This decomposition
follows from $w_{kj}=\overline{w_{jk}}$ together with $\omega_{kj}=-\omega_{jk}$.
The measured signal stays real for every $\beta$,
and the odd part appears as a component odd in $t$ (a sine component).

Since the infinite-temperature measure is even, the orthonormal polynomials obey
$\pi_n(-\omega) = (-1)^n \pi_n(\omega)$~\cite{gautschi2004}.
First-order perturbation of the three-term recurrence~\eqref{eq:three_term},
retaining the induced rotation of the polynomials,
gives the variations of the coefficients
as integrals of the unperturbed orthonormal polynomials
against the \emph{perturbation} measure $\delta\mu_\alpha$,
\begin{align}
\delta\betaop_n &= \int \omega \, \pi_n \pi_{n-1} \, d\,\delta\mu_\alpha
 - \frac{\betaop_n}{2} \left( \int \pi_n^2 \, d\,\delta\mu_\alpha + \int \pi_{n-1}^2 \, d\,\delta\mu_\alpha \right)
 - \betaop_{n-1} \int \pi_n \pi_{n-2} \, d\,\delta\mu_\alpha,
\label{eq:ftc_dbeta_full} \\
\delta\alphaop_n &= \int \omega \, \pi_n^2 \, d\,\delta\mu_\alpha
- 2 \betaop_n \int \pi_n \pi_{n-1} \, d\,\delta\mu_\alpha.
\label{eq:ftc_dalpha_full}
\end{align}
Two distinct things vary when $\mu_\alpha \mapsto \mu_\alpha + \beta\,\delta\mu_\alpha$,
and Eq.~\eqref{eq:ftc_dbeta_full} collects both.
The first is the measure in the defining integral
$\betaop_n = \int \omega\,\pi_n\pi_{n-1}\,d\mu_\alpha$.
The second is the orthonormal system itself, which must be re-orthonormalized in
the perturbed space.
Writing $\pi_n \mapsto \pi_n + \beta\,\dot\pi_n$ with
$\dot\pi_n = \sum_{m\le n} c_{nm}\pi_m$ and differentiating
$\int \pi_n\pi_m \, d\mu_\alpha = \delta_{nm}$ gives
$c_{nm} + c_{mn} = -\int \pi_n\pi_m \, d\,\delta\mu_\alpha$.
Since $\dot\pi_m$ has degree at most $m$,
the coefficient $c_{mn}$ vanishes for $m<n$, so that
\begin{equation}
c_{nm} = -\int \pi_n \pi_m \, d\,\delta\mu_\alpha \quad (m < n), \qquad
c_{nn} = -\frac{1}{2}\int \pi_n^2 \, d\,\delta\mu_\alpha.
\label{eq:ftc_rotation_coeff}
\end{equation}
Substituting these into the variation of the defining integral
and using the recurrence~\eqref{eq:three_term} twice
produces the three terms of Eq.~\eqref{eq:ftc_dbeta_full}.
They are the direct variation of the measure,
the renormalization of $\pi_n$ and $\pi_{n-1}$ (the $c_{nn}$ terms),
and the rotation of $\pi_n$ toward $\pi_{n-2}$ (an off-diagonal $c_{nm}$ term).

Denote by $\mathcal{V}_n[\nu]$ the right-hand side of Eq.~\eqref{eq:ftc_dbeta_full}
viewed as a functional of the perturbing measure,
so that $\delta\betaop_n = \mathcal{V}_n[\delta\mu_\alpha]$.
Because $\mathcal{V}_n$ is \emph{linear} in $\nu$,
it isolates the effect of the perturbation
and vanishes when evaluated on the unperturbed measure $\mu_\alpha$.
Indeed,
\begin{equation}
\mathcal{V}_n[\mu_\alpha] = \betaop_n - \frac{\betaop_n}{2}(1+1) - 0 = 0,
\qquad \text{hence} \qquad
\mathcal{V}_n[\nu] = \mathcal{V}_n\bigl[\nu - \nu(\mathbb{R})\,\mu_\alpha\bigr],
\label{eq:ftc_scale_inv}
\end{equation}
the first equality being the infinitesimal form of the exact invariance
$\betaop_n[c\,\mu_\alpha] = \betaop_n[\mu_\alpha]$, $c>0$
(rescaling a measure rescales $\pi_n$ by $1/\sqrt c$ and leaves the recurrence coefficients unchanged).
Only the \emph{shape} of the thermal correction biases the chain,
never its normalization, consistently with $\delta\mu_\alpha(\mathbb{R}) = 0$ in Eq.~\eqref{eq:ftc_dmu_def}.
The polynomial factors in the first line are even and those in the second line are odd,
so the parity of the contributing part of $\delta\mu_\alpha$ is fixed.
The hopping coefficients $\betaop_n$ are affected
only by the even part of the thermal correction,
and the diagonal coefficients $\alphaop_n$ only by the odd part.
The even-measure projection used in our non-negative least-squares pipeline consequently
removes the $O(\beta)$ odd-sector error and restores $\alphaop_n = 0$ exactly,
but it does not remove the physical bias in $\betaop_n$.

If the retained spectral atoms satisfy $p_\nu \ge p_{\mathrm{cut}}$, define
\begin{equation}
\varrho_\star := \Bigl\|\frac{d\,\delta\mu_\alpha}{d\mu_\alpha}\Bigr\|_{L^\infty(\mu_\alpha)}
\le 2\Delta H_b\sqrt{\frac{g_{\max}}{D\,p_{\mathrm{cut}}}},
\label{eq:ftc_rho_star}
\end{equation}
with $\Delta H_b$ the spectral half-spread of $H_b$
and $g_{\max}$ the maximum multiplicity of ordered pairs pooled into one retained Bohr-frequency atom.
This factor follows by applying Cauchy--Schwarz to the pooled sum over the $g_\nu$ transitions contributing to each atom.
Applying the Cauchy--Schwarz inequality in $L^2(\mu_\alpha)$
to each term of Eq.~\eqref{eq:ftc_dbeta_full}
and using $\betaop_n\le\Lambda$ gives the following bound to first order in $\beta$.
\begin{equation}
\bigl| \tbetaop_n(\beta) - \betaop_n \bigr| \le 3 \beta \, \Lambda \, \varrho_\star + O(\beta^2), \qquad
\Lambda = \max_{jk} | \omega_{jk} |.
\label{eq:jacobi_finite_temp}
\end{equation}
This reproduces the bound in Eq.~\eqref{eq:jacobi_finite_temp_main} of Sec.~\ref{sec:finite_temp},
representing the finite-temperature systematic error
relevant to the reconstructed operator-Krylov chain.

Equation~\eqref{eq:ftc_dbeta_full} was verified numerically
for the $N = 6$ ring at $g_x = 0.2$ in two ways.
For a generic even, traceless perturbation its right-hand side agrees with
a central finite difference of the Lanczos coefficients to ten significant digits,
which confirms the rotation coefficients~\eqref{eq:ftc_rotation_coeff}.
Against the exact finite-$\beta$ evaluation it reproduces $\delta\betaop_n$ to $0.5\%$,
whereas keeping only its first term is wrong by factors of $2.0$--$6.4$
and has the wrong sign at $n=3$.
The re-orthonormalization terms therefore contribute at leading order.
After the even projection $\max_n|\alphaop_n| \lesssim 10^{-14}$,
confirming that the odd-sector error is removed exactly.

\subsection{Quantitative bounds on the correction}
\label{app:tau_error}

Let $\epsilon_\beta := 3 \beta \Lambda \varrho_\star$ and assume
$\betaop_n \ge \betaop_{\min} > \epsilon_\beta$.
Since $\tauop(m) = \sum_{n<m} 1/\betaop_n$,
\begin{equation}
\frac{\bigl| \tauop(m;\beta) - \tauop(m) \bigr|}{\tauop(m)}
\le \frac{\betaop_{\max}\,\epsilon_\beta}{\betaop_{\min} (\betaop_{\min} - \epsilon_\beta)},
\label{eq:tau_error}
\end{equation}
which is independent of the Krylov depth $m$.
The accumulation of terms is compensated by the growth of $\tauop(m)$ itself.
The same substitution
$\delta \to \delta_\beta := \sup_t |\tilde C_\alpha(t;\beta)-C_\alpha(t)| \le 2\beta\Delta H_b$
in the stability estimate of Theorem~\ref{thm:sample_complexity} (Proposition~\ref{prop:expl_freq})
bounds the induced pole and parameter errors.
The resulting figures are quoted in the finite-temperature column of Table~\ref{tab:finite_temp}.

For the $N = 6$ XXZ benchmark at $g_x = 0.2$ an exact finite-$\beta$ calculation gives
$\max_n |\Delta\betaop_n| \approx 30\,\beta J^2$ and $|\Delta\tauop|/\tauop \approx 1.1\,\beta J$.
At $\beta J = 10^{-4}$ the resulting transport-time bias is about $10^{-4}$,
well below the finite-window and sampling errors reported in Table~\ref{tab:finite_temp}.
At room temperature (NV-center) or
at the cryogenic readout temperature typical of donor-spin sensors ($\sim$0.3--3~K),
the nuclear dipolar bath coupling ($\sim$1~kHz) is
many orders of magnitude below the thermal frequency scale $k_{\mathrm B}T/h$ in either regime,
so the finite-temperature correction is negligible for the present protocol.

\subsection{In-situ thermometry from the odd component}
\label{app:thermometry}

The odd-in-$t$ component identified above can be used as a thermometer.
Replacing step~(3) of the pipeline of Appendix~\ref{app:construction} by the augmented fit
\begin{equation}
\hat C_\alpha(t_k) \simeq q_0
+ \sum_{\omega>0} q_\omega \cos(\omega t_k)
+ \sum_{\omega>0} q'_\omega \sin(\omega t_k),
\label{eq:ftc_augmented_fit}
\end{equation}
with the cosine block constrained to $q_\omega \ge 0$ as before and the sine block left
unconstrained (the odd part is a signed measure), the diagnostic
\begin{equation}
\frac{\|q'\|_1}{\|q\|_1} = O(\beta\,\Delta H_b)
\label{eq:ftc_thermometer}
\end{equation}
returns the departure from infinite temperature from the data alone,
with no external calibration.
The cosine-only estimator used in our numerics has $q' \equiv 0$ by construction
and therefore cannot serve as a thermometer,
although it does remove the odd-sector error it cannot measure.
More fundamentally,
the odd sector vanishes identically whenever the system
possesses an antiunitary symmetry $\mathcal{T}$
with $\mathcal{T} H \mathcal{T}^{-1} = H$,
$\mathcal{T} \rho_b \mathcal{T}^{-1} = \rho_b$, and
$\mathcal{T} \sigma^\alpha_p \mathcal{T}^{-1} = \pm \sigma^\alpha_p$,
since then
$\tilde C_\alpha(t;\beta) = \overline{\tilde C_\alpha(t;\beta)} = \tilde C_\alpha(-t;\beta)$.
This covers every Hamiltonian with a real matrix representation,
including the XXZ benchmark used here,
for which the diagnostic reads zero at all temperatures
and the entire $O(\beta)$ effect sits in the even sector.
The thermometer is therefore informative precisely for systems without such a symmetry.

\section*{Acknowledgments}
S.M. thanks Franco Nori for introducing him to work on restricted-access Hamiltonian learning.
This work was supported by JSPS KAKENHI Grant Number JP24K00634.

\bibliographystyle{unsrturl}
\bibliography{refs_Op}

@article{difranco2009,
  author  = {Di Franco, Carlo and Paternostro, Mauro and Kim, M. S.},
  title   = {Hamiltonian Tomography in an Access-Limited Setting
             without State Initialization},
  journal = {Phys. Rev. Lett.},
  volume  = {102},
  pages   = {187203},
  year    = {2009},
  doi     = {10.1103/PhysRevLett.102.187203},
  eprint  = {0812.3510},
  archivePrefix = {arXiv},
}

@article{burgarth2009gateway,
  author  = {Burgarth, Daniel and Maruyama, Koji},
  title   = {Indirect {H}amiltonian identification through a small gateway},
  journal = {New J. Phys.},
  volume  = {11},
  pages   = {103019},
  year    = {2009},
  doi     = {10.1088/1367-2630/11/10/103019},
  eprint  = {0903.0612},
  archivePrefix = {arXiv},
}

@article{burgarth2009,
  author  = {Burgarth, Daniel and Maruyama, Koji and Nori, Franco},
  title   = {Coupling strength estimation for spin chains despite restricted access},
  journal = {Phys. Rev. A},
  volume  = {79},
  pages   = {020305(R)},
  year    = {2009},
  doi     = {10.1103/PhysRevA.79.020305},
  eprint  = {0810.2866},
  archivePrefix = {arXiv},
}

@article{wiebe2014,
  author  = {Wiebe, Nathan and Granade, Christopher and Ferrie, Christopher and Cory, D. G.},
  title   = {Hamiltonian Learning and Certification Using Quantum Resources},
  journal = {Phys. Rev. Lett.},
  volume  = {112},
  pages   = {190501},
  year    = {2014},
  doi     = {10.1103/PhysRevLett.112.190501},
}

@article{haah2023,
  author        = {Haah, Jeongwan and Kothari, Robin and Tang, Ewin},
  title         = {Optimal learning of quantum {H}amiltonians from high-temperature
                   {G}ibbs states},
  journal       = {Nat. Phys.},
  volume        = {20},
  pages         = {1027--1031},
  year          = {2024},
  doi           = {10.1038/s41567-023-02376-x},
  eprint        = {2108.04842},
  archivePrefix = {arXiv},
}

@article{huang2023,
  author  = {Huang, Hsin-Yuan and Tong, Yu and Fang, Di and Su, Yuan},
  title   = {Learning Many-Body {H}amiltonians with {H}eisenberg-Limited Scaling},
  journal = {Phys. Rev. Lett.},
  volume  = {130},
  pages   = {200403},
  year    = {2023},
  doi     = {10.1103/PhysRevLett.130.200403},
}

@article{bairey2019,
  author  = {Bairey, Eyal and Arad, Itai and Lindner, Netanel H.},
  title   = {Learning a local {H}amiltonian from local measurements},
  journal = {Phys. Rev. Lett.},
  volume  = {122},
  pages   = {020504},
  year    = {2019},
  doi     = {10.1103/PhysRevLett.122.020504},
}

@article{lapasar2012,
  author  = {Lapasar, Elham Hosseini and Maruyama, Koji and Burgarth, Daniel
             and Takui, Takeji and Kondo, Yasushi and Nakahara, Mikio},
  title   = {Estimation of Coupling Constants of a Three-Spin Chain:
             Case Study of {H}amiltonian Tomography with {NMR}},
  journal = {New J. Phys.},
  volume  = {14},
  pages   = {013043},
  year    = {2012},
  doi     = {10.1088/1367-2630/14/1/013043},
  eprint  = {1111.1381},
  archivePrefix = {arXiv},
}

@article{hangleiter2024,
  author  = {Hangleiter, Dominik and Roth, Ingo and Fuksa, Jonas
             and Eisert, Jens and Roushan, Pedram},
  title   = {Robustly learning the {H}amiltonian dynamics of a
             superconducting quantum processor},
  journal = {Nat. Commun.},
  volume  = {15},
  pages   = {9595},
  year    = {2024},
  doi     = {10.1038/s41467-024-52629-3},
  eprint  = {2108.08319},
  archivePrefix = {arXiv},
}

@article{nandy2025krylov,
  author  = {Nandy, Pratik and Matsoukas-Roubeas, Apollonas S.
             and Mart{\'i}nez-Azcona, Pablo and Dymarsky, Anatoly
             and del Campo, Adolfo},
  title   = {Quantum Dynamics in {K}rylov Space: Methods and Applications},
  journal = {Phys. Rep.},
  volume  = {1125},
  pages   = {1--82},
  year    = {2025},
  doi     = {10.1016/j.physrep.2025.05.001},
  eprint  = {2405.09628},
  archivePrefix = {arXiv},
}

@article{parker2019,
  author  = {Parker, Daniel E. and Cao, Xiangyu and Avdoshkin, Alexander
             and Scaffidi, Thomas and Altman, Ehud},
  title   = {A Universal Operator Growth Hypothesis},
  journal = {Phys. Rev. X},
  volume  = {9},
  pages   = {041017},
  year    = {2019},
  doi     = {10.1103/PhysRevX.9.041017},
  eprint  = {1812.08657},
  archivePrefix = {arXiv},
}

@article{rabinovici2022,
  author  = {Rabinovici, Ehud and S{\'a}nchez-Garrido, Adri{\'a}n
             and Shyam, Ronak and Sonner, Julian},
  title   = {Krylov localization and suppression of complexity},
  journal = {J. High Energy Phys.},
  volume  = {2022},
  pages   = {211},
  year    = {2022},
  doi     = {10.1007/JHEP03(2022)211},
}

@article{dymarsky2021,
  author  = {Dymarsky, Anatoly and Smolkin, Michael},
  title   = {Krylov complexity in conformal field theory},
  journal = {Phys. Rev. D},
  volume  = {104},
  pages   = {L081702},
  year    = {2021},
  doi     = {10.1103/PhysRevD.104.L081702},
  eprint  = {2104.09514},
  archivePrefix = {arXiv},
}

@article{barbon2019,
  author  = {Barb{\'o}n, J. L. F. and Rabinovici, Ehud and Shyam, Ronak
             and Sinha, Ritam},
  title   = {On the evolution of operator complexity beyond scrambling},
  journal = {J. High Energy Phys.},
  volume  = {2019},
  pages   = {264},
  year    = {2019},
  doi     = {10.1007/JHEP10(2019)264},
  eprint  = {1907.05393},
  archivePrefix = {arXiv},
}

@article{hua1990,
  author  = {Hua, Yingbo and Sarkar, Tapan K.},
  title   = {Matrix pencil method for estimating parameters of exponentially
             damped/undamped sinusoids in noise},
  journal = {IEEE Trans. Acoust. Speech Signal Process.},
  volume  = {38},
  pages   = {814--824},
  year    = {1990},
  doi     = {10.1109/29.56027},
}

@article{lanczos1950,
  author  = {Lanczos, Cornelius},
  title   = {An iteration method for the solution of the eigenvalue problem
             of linear differential and integral operators},
  journal = {J. Res. Natl. Bur. Stand.},
  volume  = {45},
  pages   = {255--282},
  year    = {1950},
}

@book{saad2003,
  author    = {Saad, Yousef},
  title     = {Iterative Methods for Sparse Linear Systems},
  edition   = {2nd},
  publisher = {SIAM},
  address   = {Philadelphia},
  year      = {2003},
}

@book{golub2009,
  author    = {Golub, Gene H. and Meurant, G{\'e}rard},
  title     = {Matrices, Moments and Quadrature with Applications},
  publisher = {Princeton University Press},
  year      = {2009},
}

@book{viswanath1994,
  author    = {Viswanath, V. S. and M{\"u}ller, G.},
  title     = {The Recursion Method: Application to Many-Body Dynamics},
  publisher = {Springer},
  address   = {Berlin},
  year      = {1994},
}

@book{chihara1978,
  author    = {Chihara, Theodore S.},
  title     = {An Introduction to Orthogonal Polynomials},
  publisher = {Gordon and Breach},
  address   = {New York},
  year      = {1978},
}

@article{lieb1972,
  author  = {Lieb, Elliott H. and Robinson, Derek W.},
  title   = {The finite group velocity of quantum spin systems},
  journal = {Commun. Math. Phys.},
  volume  = {28},
  pages   = {251--257},
  year    = {1972},
  doi     = {10.1007/BF01645779},
}

@article{nachtergaele2006,
  author  = {Nachtergaele, Bruno and Sims, Robert},
  title   = {Lieb-{R}obinson bounds and the exponential clustering theorem},
  journal = {Commun. Math. Phys.},
  volume  = {265},
  pages   = {119--130},
  year    = {2006},
  doi     = {10.1007/s00220-006-1556-1},
}

@book{gautschi2004,
  author    = {Gautschi, Walter},
  title     = {Orthogonal Polynomials: Computation and Approximation},
  publisher = {Oxford University Press},
  address   = {Oxford},
  year      = {2004}
}

@article{bartsch2009,
  author  = {Bartsch, Christian and Gemmer, Jochen},
  title   = {Dynamical Typicality of Quantum Expectation Values},
  journal = {Phys. Rev. Lett.},
  volume  = {102},
  pages   = {110403},
  year    = {2009},
  doi     = {10.1103/PhysRevLett.102.110403},
  eprint  = {0902.0927},
  archivePrefix = {arXiv}
}

@article{burgarth2012,
  author  = {Burgarth, Daniel and Yuasa, Kazuya},
  title   = {Quantum System Identification},
  journal = {Phys. Rev. Lett.},
  volume  = {108},
  pages   = {080502},
  year    = {2012},
  doi     = {10.1103/PhysRevLett.108.080502},
  eprint  = {1104.0583},
  archivePrefix = {arXiv}
}

@article{sone2017,
  author  = {Sone, Akira and Cappellaro, Paola},
  title   = {Hamiltonian identifiability assisted by a single-probe measurement},
  journal = {Phys. Rev. A},
  volume  = {95},
  pages   = {022335},
  year    = {2017},
  doi     = {10.1103/PhysRevA.95.022335},
  eprint  = {1609.09446},
  archivePrefix = {arXiv}
}

@article{chen2025probe,
  author        = {Chen, Sitan and Cotler, Jordan and Huang, Hsin-Yuan},
  title         = {Quantum Probe Tomography},
  year          = {2025},
  eprint        = {2510.08499},
  archivePrefix = {arXiv},
  primaryClass  = {quant-ph},
  journal       = {arXiv:2510.08499}
}

@article{ghosh2025local,
  author        = {Ghosh, Riddhi and Gilchrist, Alexei and Burgarth, Daniel},
  title         = {Characterizing errors in parameter estimation by local measurements},
  journal       = {Phys. Rev. Research},
  year          = {2026},
  doi           = {10.1103/4czy-q998},
  eprint        = {2505.15021},
  archivePrefix = {arXiv},
  primaryClass  = {quant-ph}
}

@article{zhang2014,
  author        = {Zhang, Jun and Sarovar, Mohan},
  title         = {Quantum Hamiltonian Identification from Measurement Time Traces},
  journal       = {Phys. Rev. Lett.},
  volume        = {113},
  pages         = {080401},
  year          = {2014},
  doi           = {10.1103/PhysRevLett.113.080401}
}

@article{liu2023krylov,
  author  = {Liu, Chang and Tang, Haifeng and Zhai, Hui},
  title   = {Krylov complexity in open quantum systems},
  journal = {Phys. Rev. Research},
  volume  = {5},
  pages   = {033085},
  year    = {2023},
  doi     = {10.1103/PhysRevResearch.5.033085},
  eprint  = {2207.13603},
  archivePrefix = {arXiv},
}

@article{bhattacharya2022open,
  author  = {Bhattacharya, Aranya and Nandy, Pratik and Nath, Pingal Pratyush
             and Sahu, Himanshu},
  title   = {Operator growth and {K}rylov construction in dissipative open
             quantum systems},
  journal = {J. High Energy Phys.},
  volume  = {2022},
  pages   = {081},
  year    = {2022},
  doi     = {10.1007/JHEP12(2022)081},
  eprint  = {2207.05347},
  archivePrefix = {arXiv},
}

@article{CortesGray2022,
  title     = {Quantum {K}rylov subspace algorithms for ground- and
               excited-state energy estimation},
  author    = {Cortes, Cristian L. and Gray, Stephen K.},
  journal   = {Phys. Rev. A},
  volume    = {105},
  pages     = {022417},
  year      = {2022},
  doi       = {10.1103/PhysRevA.105.022417},
  eprint    = {2109.06868},
  archivePrefix = {arXiv},
  primaryClass  = {quant-ph}
}

@article{Shen2023,
  title     = {Real-time {K}rylov theory for quantum computing algorithms},
  author    = {Shen, Yuning and Klymko, Katherine and Sud, Jarvis and
               Williams-Young, David B. and de Jong, Wibe A. and
               Tubman, Norm M.},
  journal   = {Quantum},
  volume    = {7},
  pages     = {1066},
  year      = {2023},
  doi       = {10.22331/q-2023-07-25-1066},
  eprint    = {2208.01063},
  archivePrefix = {arXiv},
  primaryClass  = {quant-ph}
}

@article{Kirby2024,
  title     = {Analysis of quantum {K}rylov algorithms with errors},
  author    = {Kirby, William},
  journal   = {Quantum},
  volume    = {8},
  pages     = {1457},
  year      = {2024},
  doi       = {10.22331/q-2024-08-29-1457},
  eprint    = {2401.01246},
  archivePrefix = {arXiv},
  primaryClass  = {quant-ph}
}

@article{Rabinovici2021,
  title     = {Operator complexity: a journey to the edge of {K}rylov space},
  author    = {Rabinovici, Eliezer and S\'{a}nchez-Garrido, Adri\'{a}n and
               Shir, Ruth and Sonner, Julian},
  journal   = {J. High Energy Phys.},
  volume    = {2021},
  pages     = {62},
  year      = {2021},
  doi       = {10.1007/JHEP06(2021)062},
  eprint    = {2009.01862},
  archivePrefix = {arXiv},
  primaryClass  = {hep-th}
}

@article{Simon1998,
  title     = {The classical moment problem as a self-adjoint finite difference operator},
  author    = {Simon, Barry},
  journal   = {Adv. Math.},
  volume    = {137},
  pages     = {82--203},
  year      = {1998},
  doi       = {10.1006/aima.1998.1728},
  eprint    = {math-ph/9906008},
  archivePrefix = {arXiv}
}

@article{Klymko2022,
  title     = {Real-Time Evolution for Ultracompact {H}amiltonian Eigenstates on Quantum Hardware},
  author    = {Klymko, Katherine and Mejuto-Zaera, Carlos and Cotton, Stephen J. and Wudarski, Filip and Urbanek, Miroslav and Hait, Diptarka and Head-Gordon, Martin and Whaley, K. Birgitta and Moussa, Jonathan and Wiebe, Nathan and de Jong, Wibe A. and Tubman, Norm M.},
  journal   = {PRX Quantum},
  volume    = {3},
  pages     = {020323},
  year      = {2022},
  doi       = {10.1103/PRXQuantum.3.020323},
  eprint    = {2103.08563},
  archivePrefix = {arXiv},
  primaryClass  = {quant-ph}
}

@article{Yoshioka2022,
  title     = {Generalized Quantum Subspace Expansion},
  author    = {Yoshioka, Nobuyuki and Hakoshima, Hideaki and Matsuzaki, Yuichiro and Tokunaga, Yuuki and Suzuki, Yasunari and Endo, Suguru},
  journal   = {Phys. Rev. Lett.},
  volume    = {129},
  pages     = {020502},
  year      = {2022},
  doi       = {10.1103/PhysRevLett.129.020502},
  eprint    = {2107.02611},
  archivePrefix = {arXiv},
  primaryClass  = {quant-ph}
}

@article{Rabinovici2025review,
  title     = {Krylov complexity},
  author    = {Rabinovici, Eliezer and S\'{a}nchez-Garrido, Adri\'{a}n and Shir, Ruth and Sonner, Julian},
  journal   = {arXiv:2507.06286},
  year      = {2025},
  eprint    = {2507.06286},
  archivePrefix = {arXiv},
  primaryClass  = {hep-th}
}

@inproceedings{Moitra2015,
  author        = {Moitra, Ankur},
  title         = {Super-resolution, Extremal Functions and the Condition Number of {V}andermonde Matrices},
  booktitle     = {Proceedings of the forty-seventh annual ACM symposium on Theory of Computing},
  series        = {STOC '15},
  pages         = {821–830},
  publisher     = {Association for Computing Machinery},
  year          = {2015},
  doi           = {10.1145/2746539.2746561},
  eprint        = {1408.1681},
  archivePrefix = {arXiv},
}

@book{Tsybakov2009,
  title     = {Introduction to Nonparametric Estimation},
  author    = {Tsybakov, Alexandre B.},
  series    = {Springer Series in Statistics},
  issn      = {0172-7397},
  isbn      = {9780387790527},
  year      = {2009},
  publisher = {Springer New York},
  doi       = {10.1007/b13794},
}

@book{Wainwright2019,
  title     = {High-Dimensional Statistics: A Non-Asymptotic Viewpoint},
  author    = {Wainwright, Martin J.},
  isbn      = {9781108498029},
  year      = {2019},
  publisher = {Cambridge University Press},
  doi       = {10.1017/9781108627771},
}

@book{Mehta2004,
  title     = {Random Matrices},
  author    = {Mehta, Madan Lal},
  edition   = {3rd},
  series    = {Pure and Applied Mathematics},
  volume    = {142},
  year      = {2004},
  publisher = {Elsevier/Academic Press},
  doi       = {10.1016/S0079-8169(04)80091-6},
}

@article{wedin1972,
  author  = {Wedin, Per-{\AA}ke},
  title   = {Perturbation bounds in connection with singular value decomposition},
  journal = {BIT Numerical Mathematics},
  volume  = {12},
  pages   = {99--111},
  year    = {1972},
  doi     = {10.1007/BF01932678}
}

@book{lawson1974,
  author    = {Lawson, Charles L. and Hanson, Richard J.},
  title     = {Solving Least Squares Problems},
  publisher = {Prentice-Hall},
  address   = {Englewood Cliffs, NJ},
  year      = {1974},
  doi       = {10.1137/1.9781611971217},
}

@article{ryan2010,
  author  = {Ryan, C. A. and Hodges, J. S. and Cory, D. G.},
  title   = {Robust Decoupling Techniques to Extend Quantum Coherence in Diamond},
  journal = {Phys. Rev. Lett.},
  volume  = {105},
  number  = {20},
  pages   = {200402},
  year    = {2010},
  doi     = {10.1103/PhysRevLett.105.200402},
  eprint  = {1008.2197},
  archivePrefix = {arXiv},
}

@article{kobayashi2021,
  author  = {Kobayashi, Takashi and Salfi, Joseph and Chua, Christopher
             and van der Heijden, Joost and House, Matthew G.
             and Culcer, Dmitry and Hutchison, Wayne D.
             and Johnson, Brett C. and McCallum, Julian C.
             and Riemann, Helge and Abrosimov, Nikolai V.
             and Becker, Peter and Pohl, Hans-Joachim
             and Simmons, Michelle Y. and Rogge, Sven},
  title   = {Engineering long spin coherence times of spin--orbit qubits in silicon},
  journal = {Nature Materials},
  volume  = {20},
  number  = {1},
  pages   = {38--42},
  year    = {2021},
  doi     = {10.1038/s41563-020-0743-3},
  eprint  = {1809.10859},
  archivePrefix = {arXiv},
}

@article{hetenyi2020exchange,
  author  = {Het{\'e}nyi, Bence and Kloeffel, Christoph and Loss, Daniel},
  title   = {Exchange interaction of hole-spin qubits in double quantum dots
             in highly anisotropic semiconductors},
  journal = {Phys. Rev. Research},
  volume  = {2},
  pages   = {033036},
  year    = {2020},
  doi     = {10.1103/PhysRevResearch.2.033036},
  eprint  = {2004.07658},
  archivePrefix = {arXiv},
}

@article{geyer2024anisotropic,
  author  = {Geyer, Simon and Het{\'e}nyi, Bence and Bosco, Stefano
             and Camenzind, Leon C. and Eggli, Rafael S. and Fuhrer, Andreas
             and Loss, Daniel and Warburton, Richard J. and Zumb{\"u}hl, Dominik M.
             and Kuhlmann, Andreas V.},
  title   = {Anisotropic exchange interaction of two hole-spin qubits},
  journal = {Nat. Phys.},
  volume  = {20},
  pages   = {1152--1157},
  year    = {2024},
  doi     = {10.1038/s41567-024-02481-5},
  eprint  = {2212.02308},
  archivePrefix = {arXiv},
}

@article{matsuura2026state,
  author        = {Matsuura, Shunji and Kawamura, Yoji and Salfi, Joseph and Iso, Satoshi},
  title         = {{K}rylov Break Times from an Inhomogeneous {L}ieb--{R}obinson Light Cone},
  journal       = {arXiv:2608.27399},
  year          = {2026},
  eprint        = {2608.27399},
  archivePrefix = {arXiv},
}

@article{matsuura2026hp,
  author        = {Matsuura, Shunji and Kawamura, Yoji and Salfi, Joseph and Nori, Franco and Baril, {\'E}lo{\"i}se and Iso, Satoshi},
  title         = {Coherent Quantum State Import from Uncontrolled Physical Systems},
  journal       = {in preparation},
  year          = {2026},
}

\end{document}